\documentclass[12pt]{article}
\usepackage{mathrsfs}
\usepackage{amsthm,amsmath,enumerate,amsbsy,amsfonts,mathabx,amscd,graphicx,algorithm,soul}
\usepackage{authblk}
\usepackage{rotating,caption,subcaption,multirow,epsfig,indentfirst}
\usepackage{natbib}
\usepackage{url} 
\usepackage{booktabs}
\usepackage{threeparttable}
\usepackage[colorlinks,bookmarksopen,bookmarksnumbered,citecolor=blue,urlcolor=blue,linkcolor=blue]{hyperref}
\usepackage{float}
\newcommand{\blind}{0}

\newtheorem{definition}{Definition}
\newtheorem{theorem}{Theorem}

\newtheorem{lemma}{Lemma}

\newtheorem{remark}{Remark}
\newtheorem{condition}{Condition}

\def\spacingset#1{\renewcommand{\baselinestretch}%
	{#1}\small\normalsize} \spacingset{1}

\newcommand{\bzero}{\boldsymbol 0}
\newcommand{\bLambda}{\boldsymbol \Lambda}

\newcommand{\bSigma}{\boldsymbol \Sigma}

\newcommand{\bGamma}{\boldsymbol \Gamma}

\newcommand{\bSigmak}{\boldsymbol \Sigma^{(k)}}

\def\P{{ \mathrm{p} }}

\newcommand{\bPhi}{\boldsymbol \Phi}

\newcommand{\bxi}{\boldsymbol \xi}

\newcommand{\bgamma}{\boldsymbol \gamma}
\newcommand{\bdelta}{\boldsymbol \delta}

\newcommand{\bvarepsilon}{\boldsymbol \varepsilon}

\newcommand{\bepsilon}{\boldsymbol \epsilon}

\def\sint{\begingroup\textstyle \int\endgroup}

\newcommand{\bSigy}{\bSigmak_{yy}}

\newcommand{\thresh}{\mathcal T_{\eta_k}\{\bSigy\}}

\newcommand{\ba}{{\mathbf a}}
\newcommand{\be}{{\mathbf e}}
\newcommand{\bbf}{{\mathbf f}}

\newcommand{\bk}{{\mathbf k}}

\newcommand{\bw}{{\mathbf w}}
\newcommand{\bg}{{\mathbf g}}

\newcommand{\bbb}{{\mathbf b}}

\newcommand{\bW}{{\bf W}}

\newcommand{\bA}{{\bf A}}
\newcommand{\bB}{{\bf B}}

\newcommand{\bE}{{\bf E}}
\newcommand{\bI}{{\bf I}}
\newcommand{\bK}{{\bf K}}

\newcommand{\bS}{{\bf S}}
\newcommand{\bT}{{\bf T}}
\newcommand{\bX}{{\bf X}}
\newcommand{\bY}{{\bf Y}}
\newcommand{\bZ}{{\bf Z}}

\newcommand{\bU}{{\bf U}}
\newcommand{\bV}{{\bf V}}
\newcommand{\bQ}{{\bf Q}}

\newcommand{\bM}{{\bf M}}

\newcommand{\bH}{{\bf H}}

\newcommand{\cC}{{\cal C}}
\newcommand{\cD}{{\cal D}}

\newcommand{\cL}{{\cal L}}
\newcommand{\cM}{{\cal M}}

\newcommand{\cT}{{\cal T}}
\newcommand{\cU}{{\cal U}}
\newcommand{\cV}{{\cal V}}
\newcommand{\cS}{{\cal S}}

\newcommand{\cR}{{\cal R}}
\newcommand{\cB}{{\cal B}}

\newcommand{\eZ}{\mathbb{Z}}
\newcommand{\eR}{\mathbb{R}}
\newcommand{\eH}{\mathbb{H}}
\newcommand{\eE}{\mathbb{E}}

\newcommand{\tF}{\text{F}}

\newcommand{\cov}{\text{Cov}}

\newcommand{\pr}{\mathbb{P}}
\newcommand{\tr}{\mbox{tr}}

\newenvironment{condbis}[1]
  {\renewcommand{\thecondition}{\ref{#1}$'$}%
   \addtocounter{condition}{-1}%
   \begin{condition}}
  {\end{condition}}

\def\T{{ \mathrm{\scriptscriptstyle T} }}

\begin{document}
\def\spacingset#1{\renewcommand{\baselinestretch}%
{#1}\small\normalsize} \spacingset{1}

\if0\blind
{
  \title{\bf Factor Modelling for High-dimensional Functional Time Series}
  \author[1]{Shaojun Guo}
		\author[2]{Xinghao Qiao}
		\author[1]{Qingsong Wang}
            \author[3]{Zihan Wang}
		\affil[1]{Institute of Statistics and Big Data, Renmin University of China, P.R. China}
            \affil[2]{Faculty of Business and Economics, The University of Hong Kong, Hong Kong}
            \affil[3]{Department of Statistics and Data Science, Tsinghua University, P.R. China}
		\setcounter{Maxaffil}{0}
		
		\renewcommand\Affilfont{\itshape\small}
		\date{\vspace{-5ex}}
  \maketitle
} \fi

\if1\blind
{
  \bigskip
  \bigskip
  \bigskip
  \begin{center}
    {\LARGE\bf Factor Modelling for High-dimensional Functional Time Series}
\end{center}
  \medskip
} \fi

\bigskip
\begin{abstract}
Many economic and scientific problems involve the analysis of high-dimensional functional time series, where the number of functional variables $p$ diverges as the number of serially dependent observations $n$ increases. In this paper, we present a novel functional factor model for high-dimensional functional time series that maintains and makes use of the functional and dynamic structure to achieve great dimension reduction and find the latent factor structure. To estimate the number of functional factors and the factor loadings, we propose a fully functional estimation procedure based on an eigenanalysis for a nonnegative definite and symmetric matrix. Our proposal involves a  weight matrix to improve the estimation efficiency and tackle the issue of heterogeneity, the rationale of which is illustrated by formulating the estimation from a novel regression perspective. Asymptotic properties of the proposed method are studied when $p$ diverges at some polynomial rate as $n$ increases. To provide a parsimonious model and enhance interpretability for near-zero factor loadings, we impose sparsity assumptions on the factor loading space and then develop a regularized estimation procedure with theoretical guarantees when $p$ grows exponentially fast relative to $n.$ Finally, we demonstrate that our proposed estimators significantly outperform the competing methods through both simulations and applications to a U.K. temperature data set and a Japanese mortality data set.
\end{abstract}

\noindent%
{\it Keywords:}  Dimension reduction; Functional time series; Functional thresholding; High-dimensional data; Sparse principal component analysis; Weight matrix.
\vfill

\newpage
\spacingset{1.7} 

\setlength{\abovedisplayskip}{9pt plus 2pt minus 6pt}
\setlength{\belowdisplayskip}{9pt plus 2pt minus 6pt}

\section{Introduction}
\label{sec:intro}
Functional time series, which refers to a sequential collection of curves observed over time exhibiting serial dependence, has recently received a great deal of attention. 
Despite progress being made in this field, existing literature has focused on the estimation based on a single or fixed number of functional time series, see, e.g.,
\cite{Bbosq1,hormann2010,bathia2010,panaretos2013,hormann2015,aue2015prediction,li2020long,qiao2019c} and among many others. 

With the rapid development in technology, data sets containing a large collection of functional time series are becoming increasingly available in various applications. Examples include 
cumulative intraday return trajectories \cite[]{horvath2014testing} and functional volatility processes \cite[]{muller2011} 
for hundreds of stocks, 
annual temperature curves collected at a number of stations,
daily energy consumption curves for thousands of households, age-specific mortality rates in different prefectures \cite[]{gao2019high}, to list a few. Such data, which are referred to as high-dimensional functional time series, take the form of $\bY_t(\cdot)=\{Y_{t1}(\cdot), \dots, Y_{tp}(\cdot)\}^{\T}$ ($t=1, \dots, n$) defined on a compact interval $\cU,$ where the number of random functions $p$ is comparable to, or even larger than the number of serially dependent observations $n$. 
Under such scenario, not only $p$ is large, but each $Y_{tj}(\cdot)$ is a functional object with serial dependence across observations, posing extra challenges. When jointly modelling the entire curve dynamics, it is of great interest and importance to explore the latent common component from a dimension-reduction viewpoint, while preserving the functional and time series structure. This motivates us to develop a functional factor model based on $\bY_1(\cdot), \dots, \bY_n(\cdot).$

In this paper, we deal with the factor modelling for high-dimensional functional time series, in which each functional datum $\bY_t(\cdot)$ arises as the sum of two unobservable components, one common $\bZ_t(\cdot)$ and one idiosyncratic $\bvarepsilon_t(\cdot).$ The co-movement of $p$ functional time series is assumed to be driven by a small number of functional factors and their inherited dynamic structures in $\bZ_t(\cdot)$. Specifically, our proposed factor model admits the representation:
\begin{equation}
\label{eq:model}
     \bY_t(\cdot) = \bZ_t(\cdot) + \bvarepsilon_t(\cdot) = \bA\bX_t(\cdot) + \bvarepsilon_t(\cdot), ~~t = 1,\dots,n,
\end{equation} 
where $\bX_t(\cdot) = \{X_{t1}(\cdot), \dots, X_{tr}(\cdot)\}^{\T}$ is a set of latent functional factor time series with (unknown) number of functional factors $r<p,$ $\bA=(A_{ij})_{p \times r}$ is the factor loading matrix
and $\bvarepsilon_t(\cdot)$ is $p$-vector of functional idiosyncratic errors.

For multivariate scalar time series $\{\bY_t\},$ the factor model (\ref{eq:model}) degenerates to the classical dynamic factor model $\bY_t = \bA \bX_t + \bvarepsilon_t,$ where two different types of model assumptions are considered in econometrics and statistics literature. 
One type of models assumes that each common factor has non-trivial contribution to $\{\bY_t\}$ while the idiosyncratic noise $\{\bvarepsilon_t\}$ is allowed to have weak cross-correlations and serial correlations. An incomplete list of the relevant references
includes \cite{chamberlain1982arbitrage,forni2000generalized,bai2002determining,stock2012,fan2013,fan2019}. 
However, the rigorous definition of the common and idiosyncratic components in those factor models can only be established asymptotically when the dimension of $\bY_t$ tends to infinity.
Another type of models assumes that the common factors accommodate all dynamics of $\{\bY_t\},$ thus making the idiosyncratic component white noise with no serial correlations but allowing substantial contemporary cross-correlation among it; see \cite{pena1987identifying,lam2011estimation,lam2012} and extensions to matrix-valued time series \cite[]{wang2019,chen2019constrained} and tensor-valued time series \cite[]{chen2021}.
For functional time series, the factor modelling remains less touched in the literature. \cite{hays2012functional} considered functional dynamic factor models for univariate functional time series, where the functional factors do not depend on $t$ and the dynamics are modelled by scalar coefficients that follow Gaussian autoregressive processes. \cite{kokoszka2015functional} extended this model by allowing for dynamic functional factors, but both models are still resctried to the univariate case.
\cite{gao2019high} adopted factor modelling techniques to predict high-dimensional functional time series by fitting a factor model to estimated scores obtained via eigenanalysis of the long-run covariance function. 
Compared to our model setup (\ref{eq:model}),  \cite{hallin2023a} and \cite{tavakoli2023b} provide an alternative factor model formulation to tackle a rather different situation with scalar factors and functional factor loadings. Their methodology follows the framework in \cite{bai2002determining} and substantially differs from our estimation procedure. See detailed discussion in Section~\ref{sec:dis}.

Our paper is along the line of \cite{lam2012} in the model setup. To facilitate methodological development, we assume that the idiosyncratic component $\bvarepsilon_t(\cdot)$ is a white noise sequence with $\eE\{\bvarepsilon_t(u)\}=0$ and $\cov\{\bvarepsilon_{t+k}(u), \bvarepsilon_t(v)\}=0$ for any $u,v\in \cU$ and $k \neq 0.$
Although the white noise assumption is not the most ideal for extracting those asymptotically identifiable factors, it enables the unique decomposition in (\ref{eq:model}), thus largely simplifying the tasks of model identification and inference. In our established theory, we allow for weak serial correlations in $\{\bvarepsilon_t(\cdot)\}$, which relaxes the white noise assumption.
Additionally, unlike \cite{bai2002determining}, we do not place any assumption on the covariance structure of $\bvarepsilon_t(\cdot).$ 
We propose to estimate both the number of functional factors and the factor loadings in terms of an eigenanalysis for a nonnegative definite and symmetric matrix defined as the double integral and sum of weighted quadratic forms in autocovariance functions of $\bY_t(\cdot)$ from different time lags. 
In practical applications, it is common that many estimated factor loadings are close to zero especially when $p$ is large, 
This inspires us to add certain sparsity constraints on the factor loading space defined in terms of 
the factor loading matrix $\bA.$ 
Since functional sparsity patterns in autocovariance functions of $\bY_t(\cdot)$ inherits from the imposed sparsity structures in $\bA,$ we first propose Hilbert-Schmidt-norm based functional thresholding of entries of the sample auotocovariance functions, which ensures the consistencies under $\log p/n \rightarrow 0.$ Furthermore, to recover sparse factor loading space and produce a parsimonious model with interpretable factors, we perform sparse principal component analysis \cite[]{vu2013} rather than standard eigenanalysis. 

This paper makes useful contributions at multiple fronts.
On the method side, our proposal has four main advantages.
First, it involves a weight matrix to improve the estimation efficiency and tackle the issue of heterogeneity. We illustrate the  rationale of its presence by formulating the estimation procedure based on a novel regression viewpoint. Our empirical results also demonstrate the consistent superiority of our method compared to the unweighted competitors across all simulation scenarios. 
Furthermore, our strategy of incorporating the weight matrix is general and can be applied broadly to a range of dimension reduction problems, e.g., \cite{lam2011estimation} and \cite{wang2019}.
Second, our estimation relies on the sample autocovariance functions of $\bY_t(\cdot).$ This not only gets rid of the impact from the noise $\bvarepsilon_t(\cdot),$ but also makes the good use of the serial correlation information, which is the most relevant in the time series modelling and prediction. 
Third, the double integral over $(u,v) \in\cU\times\cU$ takes advantage of the functional nature of the data by accumulating the autocovariance information as much as possible, leading to more efficient estimation. 
Such fully functional procedure avoids the information loss incurred by the dimension reduction step in \cite{gao2019high}. Fourth, aided by the enforced sparsity, we enhance the model interpretability and enlarge the dimension of the feature space that our model can handle. Our regularized estimation procedure adopts techniques of novel functional thresholding and sparse principal component analysis, each of which leads to the estimation consistency and sparsity recovery at its own step. 

On the theory side, we investigate convergence properties of relevant estimated terms under the assumption that the idiosyncratic errors are either white noise or weakly serially correlated, in the context of
high-dimensional scaling, where $p$ grows polynomially fast relative to $n.$ With additional sparsity constraints and tail assumptions, we further provide convergence analysis of the regularized estimation in an ultra-high-dimensional regime, where $p$ diverges at an exponential rate as $n$ increases. 
Empirically, we demonstrate through an extensive set of simulations that our estimators significantly outperform the competitors. 
Moreover, we integrate our factor modelling framework into the challenging task of predicting high-dimensional functional time series and use a real data example to illustrate the superiority of our strategy over existing prediction methods.

The paper is set out as follows. In Section~\ref{sec:method}, we develop a fully functional procedure with weight matrix to estimate the proposed factor model and explain it from a novel regression perspective. 
The theoretical properties of the proposed estimators are investigated in Section~\ref{sec:theory}. Under further sparsity assumptions, we develop the regularized estimation procedure and present the associated convergence analysis in Section~\ref{sec:ultrahigh}. The finite-sample performance of our methods is illustrated through extensive simulation studies and two real data examples in Sections~\ref{sec:sim} and \ref{sec:data}, respectively. Section~\ref{sec:dis} concludes the paper by discussing the relationship to relevant work as well as three extensions. 
We relegate all technical proofs and additional empirical results to the Supplementary Material.

{\bf Notation}. For any matrix $\bB=(B_{ij})_{p \times q},$ we let $\|\bB\|_{\min}=\lambda_{\min}^{1/2}(\bB^{\T}\bB),$ $\|\bB\|=\lambda_{\max}^{1/2}(\bB^{\T}\bB)$ and denote its matrix $\ell_\infty$ norm and Frobenius norm by
$\|\bB\|_1 = \max_{1\leq i\leq p}\sum_{j=1}^{q}|B_{ij}|$ and $\|\bB\|_{\tF}=(\sum_{i=1}^{p}\sum_{j=1}^{q}B_{ij}^2)^{1/2},$ respectively.
Let $L_2(\cU)$ denote the Hilbert space of squared integrable functions defined on $\cU.$ 
For $f\in L_2(\cU)$, we define its $\ell_2$-norm by $\|f\|_{2} = \{\int f(u)^2\mathrm du\}^{1/2}$.
For a matrix of bivariate functions $\bK = \big(K_{ij}(\cdot,\cdot)\big)_{p \times q},$ we define the functional versions of the matrix $\ell_{\infty}$ and Frobenius norms by $\|\bK\|_{\cS,\infty} = \max_{1\leq i\leq p}\sum_{j=1}^q\|K_{ij}\|_{\cS}$ and 
$\|\bK\|_{\cS,\tF} = (\sum_{i=1}^p\sum_{j=1}^q \|K_{ij}\|_{\cS}^2)^{1/2},$
respectively, where, for each $K_{ij}\in L_2(\cU\times \cU)$, we denote its Hilbert-Schmidt norm by $\|K_{ij}\|_{\cS} = \{\int\int K_{ij}(u,v)^2 \,\mathrm d u \mathrm d v\}^{1/2}.$ For two positive sequences $\{a_n\}$ and $\{b_n\}$, we write $a_n\lesssim b_n$ or $a_n=O(b_n)$ or $b_n\gtrsim a_n$ if there exists a positive constant $c$ such that $a_n/b_n \leq c$. We write $a_n \asymp b_n$ if and only if $a_n \lesssim b_n$ and $b_n\lesssim a_n$ hold simultaneously.

\section{Methodology}
\label{sec:method}
\subsection{Model Setup and Estimation Procedure}
\label{sec:est}
We assume that the $p$-vector of functional time series $\bY_1(\cdot), \dots, \bY_n(\cdot)$ 
satisfies the functional factor model in (\ref{eq:model}), where each $\bY_t(\cdot)$ is decomposed as the sum of two parts: a common dynamic part driven by a $r$-vector of latent factor process $\bX_t(\cdot)$ and an idiosyncratic part of white noise process $\bvarepsilon_t(\cdot)=\{\varepsilon_{t1}(\cdot), \dots, \varepsilon_{tp}(\cdot)\}^{\T}.$ 
Our formulation ensures that both functional and linear dynamic structure of $\bY_t(\cdot)$ are inherited from common factors $\bX_t(\cdot),$ while all white noise elements are absorbed into $\bvarepsilon_t(\cdot).$
Note that $\{\bX_t(\cdot)\}_{t=1}^n$ and $\{\bvarepsilon_t(\cdot)\}_{t=1}^n$ are unobservable. The white noise assumption is adopted here to facilitate the methodology but is relaxed in Section~\ref{sec:theory} when establishing the supporting theory.
We further assume that the rank of the factor loading matrix $\bA$ is $r$. If this full rank condition is violated, model \eqref{eq:model} can be expressed in terms of a lower-dimensional functional factor. When $r$ is much smaller than $p,$ we achieve an effective dimension reduction. 

There is an identifiable issue among $\bX_t(\cdot)$ and $\bA,$ since model~\eqref{eq:model} remains unchanged if $\{\bA, \bX_t(\cdot)\}$ is replaced by $\{\bA\bGamma, \bGamma^{-1}\bX_t(\cdot)\}$ for any invertible matrix $\bGamma \in {\mathbb R}^{r \times r}.$ Even if we assume that the columns of $\bA$ are orthonormal, i.e.  $\bA^{\T}\bA = \bI_r$ (the $r\times r$ identity matrix), $\bX_t(\cdot)$ and $\bA$ still can not be determined uniquely. However, the linear space spanned by the columns of $\bA$ (denoted by $\mathcal C(\bA)$) is unique regardless of the column-orthogonality of $\bA$. Hence we will focus on the estimation of the factor loading space ${\mathcal C}(\bA).$

We next develop a fully functional procedure to estimate $\mathcal C(\bA)$ and the number of factors $r.$ Under the assumption that $\{\bX_t(\cdot)\}$ is weakly stationary, we define the lag-$k$ ($k \geq 0$) autocovariance functions:
\begin{equation}\label{eq:notation}
\begin{aligned}
\bSigma^{(k)}_{yy}(u,v) &= \cov\{\bY_{t+k}(u), \bY_{t}(v)\},~~ \bSigma^{(k)}_{xx}(u,v) = \cov\{\bX_{t+k}(u), \bX_{t}(v)\},\\
\bSigma^{(k)}_{x\varepsilon}(u,v) &= \cov\{\bX_{t+k}(u), \bvarepsilon_t(v)\}.
\end{aligned}
\end{equation}
We further assume that the future white noise components are uncorrelated with the common factors up to the present. It then follows from (\ref{eq:model}) and (\ref{eq:notation}) that
\begin{equation}\label{eq:cov-relation}
\bSigma^{(k)}_{yy}(u,v) = \bA \bSigma^{(k)}_{xx}(u,v) \bA^{\T} + \bA \bSigma^{(k)}_{x\varepsilon}(u,v), ~~u,v \in \cU, ~k \geq 1.
\end{equation}

Given a prescribed and fixed integer $k_0 \geq 1,$ we define a nonnegative definite and symmetric matrix
\begin{equation*}\label{eq:weighted-M}
		\bM = \sum_{k=1}^{k_0} \int_{\cU}\int_{\cU} \bSigma^{(k)}_{yy}(u,v) \bW(v)\bSigma^{(k)}_{yy}(u,v)^\T\,\mathrm d u \mathrm d v,
\end{equation*}
where the $p \times p$ weight matrix $\bW(v)$ is nonnegative definite  and symmetric
for any $v \in \cU.$ Here $\bW(v)$ is introduced to improve the estimation efficiency. In Section~\ref{sec:weight-mat}, we will suggest specific forms of $\bW(v)$ and illustrate the rationale from a regression perspective.
Replacing $\bW(v)$ by its sample version $\widehat\bW(v),$ we define
\begin{equation}\label{eq:check-M}
		\widecheck\bM = \sum_{k=1}^{k_0} \int_{\cU}\int_{\cU} \bSigma^{(k)}_{yy}(u,v) \widehat\bW(v)\bSigma^{(k)}_{yy}(u,v)^\T\,\mathrm d u \mathrm d v.
\end{equation}
Let $\nu_1 \geq \dots \geq \nu_p \geq 0$ be eigenvalues of $\widecheck \bM$
and $\bgamma_j$ be a unit eigenvector of $\widecheck \bM$ corresponding to $\nu_j$ for each $j.$ We claim that, under mild Conditions~\ref{con:x} and \ref{con:strg} imposed in Section~\ref{sec:theory}, $\text{rank}(\widecheck\bM)=r$ (i.e., $\nu_r> 0=\nu_{r+1} = \dots=\nu_p$) holds with overwhelming probability as justified in our proof in Section~\ref{supp.pf.thm1} of the Supplementary Material. As a result, the factor loading space ${\cal C}(\bA)$ can be recovered by
${\cal C}(\bK) = \text{span}\{\bgamma_1, \dots, \bgamma_r\}$ with
$\bK= (\bgamma_1, \dots,\bgamma_r)  \in {\mathbb R}^{p \times r}.$
To see this, let $\widecheck\bA$ be a $p \times (p-r)$ orthogonal complement of matrix $\bA$ such that 
$\widecheck\bA^{\T} \bA = {\bf 0}$ and $\widecheck\bA^{\T}\widecheck\bA=\bI_{p-r}.$ It then follows from (\ref{eq:cov-relation}) and (\ref{eq:check-M}) that $\widecheck \bM\widecheck\bA = {\bf 0},$ which implies that the columns of $\widecheck\bA$ are eigenvectors of $\widecheck\bM$ corresponding to $(p-r)$ zero eigenvalues and hence ${\cal C}(\bA)$ is spanned by the eigenvectors of $\widecheck\bM$ corresponding to $r$ nonzero eigenvalues.
\begin{remark}
\label{rmk1}
(i) We take the double integral and the sum in the definition of $\widecheck\bM$ in (\ref{eq:check-M}) to accumulate the information of autocovariance functions as much as possible from each $(u,v) \in \cU\times\cU$ and from different time lags leading to more efficient estimation, whereas fixing at certain $(u,v)$ or time lag may lead to spurious estimation results.

(ii) The definition in (\ref{eq:check-M}) ensures that $\widecheck \bM$ is nonnegative definite and symmetric and there is no cancellation of information accumulated from lags 1 to $k_0.$ Hence the estimation is insensitive to the choice of $k_0.$ In practice 
we tend to select a small value, $1 \leq k_0 \leq 5,$ as the strongest autocorrelations usually appear at small lags. 
\end{remark}

To estimate ${\cal C}(\bA),$ we need an eigenanalysis of the sample estimator for $\widecheck\bM,$
\begin{equation}\label{eq:Mhat}
    \widehat \bM = \sum_{k=1}^{k_0}\int_{\cU} \int_{\cU} \widehat{\bSigma}^{(k)}_{yy}(u,v) \widehat \bW(v)\widehat{\bSigma}^{(k)}_{yy}(u,v)^{\T} \,\mathrm d u\mathrm d v,
\end{equation}
where 
\begin{equation}\label{eq:est-cov}
	\widehat{\bSigma}^{(k)}_{yy}(u,v) = \frac{1}{n-k}\sum_{t=1}^{n-k}\big\{\bY_{t+k}(u)-\bar \bY(u)\big\}\big\{\bY_{t}(v)-\bar \bY(v)\big\}^{\T}, ~~u,v \in \cU,
\end{equation}
is the estimator for ${\bSigma}^{(k)}_{yy}(u,v)$ 
and $\widebar \bY(\cdot) = n^{-1}\sum_{t=1}^n \bY_t(\cdot).$ Performing eigen-decomposition of $\widehat\bM$ leads to estimated eigenvalues $\hat \nu_1 \geq \dots \geq \hat \nu_p $ and the associated eigenvectors $\widehat \bgamma_1, \dots, \widehat \bgamma_p.$ Let $\widehat \bA =(\widehat \bgamma_1, \dots, \widehat\bgamma_r).$ Then ${\cal C}(\widehat \bA)=\text{span}\{\widehat \bgamma_1, \dots, \widehat \bgamma_r\}$ forms the estimate of ${\cal C}(\bA).$

We have developed the estimation of model (\ref{eq:model}) assuming that the number of factors $r$ is known or can be identified correctly. In practice, $r$ is unknown 
and there is vast literature on the topic of determining it, see 
\cite{bai2002determining,onatski2010determining,lam2012,ahn2013eigenvalue,fan2020estimating} and \cite{han2022}, to quote a few. Here we take the commonly adopted ratio-based estimator for $r$ as:
\begin{equation}\label{eq:fac-num}
    \hat r = \arg\min_{1\leq j\leq c_r p} \frac{\hat{\nu}_{j+1}}{\hat{\nu}_j},
\end{equation}
where $c_r \in (0, 1)$ is a prespecified constant. 
In empirical studies, we take $c_r = 0.75$ to avoid the fluctuations due to the ratios of extreme small values.

\subsection{Weight Matrix}
\label{sec:weight-mat}
Under the nonnegative-definiteness and symmetry constraint, we suggest the following weight matrix:
\begin{equation}
\label{WeightM}
        \bW(v)=\bQ\big\{\bQ^\T \bSigma_{yy}^{(0)}(v,
        v)\bQ\big\}^{-1}\bQ^\T, ~~v\in \cU,
\end{equation}
where $\bQ \in {\mathbb R}^{p \times q}$ is a full-rank matrix and $q$ is larger than $r$ but much smaller than $p.$ The entries of $\bQ$ are independently sampled from some random distribution with zero mean and unit variance, e.g., $\text{Uniform}[-\sqrt 3, \sqrt 3]$ and $N(0,1),$ which is used here to facilitate technical analysis. Our experiments indicate that the results are not sensitive to the choices of $q$ and the sampling distribution. Then we can obtain the sample version of $\bW(v)$ by 
\begin{equation}
\label{hat.W0}
    \widehat \bW(v)=\bQ\big\{\bQ^\T \widehat \bSigma_{yy}^{(0)}(v,v)\bQ\big\}^{-1}\bQ^\T.
\end{equation}

From a regression perspective, we next provide an intuitive explanation why the weight matrix takes the suggested form in (\ref{WeightM}). Despite the unobservable common factors $\bX_t(\cdot),$ we can simply view (\ref{eq:model}) fixing at $u \in \cU$ as a multiple linear regression with multivariate responses $\bY_t(u)$ and the coefficient matrix $\bA$ to be estimated. To construct the covariate vector from observed data in a regression setup, we assume that there exists some column-orthogonal matrix $\bPhi \in {\mathbb R}^{q \times r}$ with $\bPhi^{\T}\bPhi={\bf I}_r$ such that $\widecheck \bX_t(\cdot) = \bPhi \bX_t(\cdot)$ can be represented by the linear combination of past observed curves with the addition of some random noise:
$$\widecheck\bX_t(\cdot)=\bQ^{\T}\bY_{t-k}(\cdot) + \widecheck \bvarepsilon_{tk}(\cdot), ~~k=1, \dots, k_0,$$ where $\widecheck \bvarepsilon_{tk}(\cdot)$ is the white noise process, independent of $\bY_{t-1}(\cdot), \dots, \bY_{t-k_0}(\cdot).$ Then model~(\ref{eq:model}) can be rewritten as  
\begin{equation}\label{eq:new-model}
	\bY_t(\cdot) = \bB\widecheck\bX_t(\cdot) + \bvarepsilon_t(\cdot) = \bB\bQ^{\T}\bY_{t-k}(\cdot) + \be_{tk}(\cdot),~~t=k+1, \dots, n,
\end{equation}
where $\be_{tk}(\cdot) = \bvarepsilon_t(\cdot) + \bB\widecheck \bvarepsilon_{tk}(\cdot)$ and
$\bB=\bA \bPhi^{\T}$ is a $p \times q$ matrix satisfying $\text{rank}(\bB)=r$ and ${\cal C}(\bB) = {\cal C}(\bA).$
Now we can treat (\ref{eq:new-model}) fixing at $u \in \cU$ as a linear regression model with observed covariate vectors $\{\bQ^\T{\bY_{t-k}}(u)\}_{t=k+1}^n$ and unknown low-rank coefficient matrix $\bB.$ 
Based on (\ref{eq:new-model}) and the population moment equation $\cov\{\be_{tk}(u),\bQ^{\T}\bY_{t-k}(v)\}=0$ for $u,v \in \cU,$ we can solve $\bB$ by
\begin{equation}
\label{eq:cov2}
\bB = \cov \Big\{\bY_t(u), \bQ^{\T}\bY_{t-k}(v)\Big\}\cov\Big\{ \bQ^\T\bY_{t-k}(u), \bQ^\T\bY_{t-k}(v)\Big\}^{-1}. 
\end{equation}
Replacing the covariance terms in (\ref{eq:cov2}) by their sample versions, we obtain the estimator 
\begin{equation}
\label{lse_B}
	\widehat\bB_{u,v} = \Big\{\sum_{t=k+1}^{n} \bY_{t}(u)\bY_{t-k}^\T(v)\bQ \Big\} \Big\{\sum_{t=k+1}^n  \bQ^\T\bY_{t-k}(u)\bY_{t-k}^\T(v)\bQ \Big\}^{-1}.
\end{equation}

Denote the $j$-th row vector of $\widehat\bB_{u,v}$ by ${\widehat\bbb}_{j} \in {\mathbb R}^q$. Since ${\widehat\bbb}_{j}$ represents a statistic  whose variation plays a pivotal role in subsequent analyses, it is advantageous to utilize the normalized statistic $\widehat\bbb_{j} \{\cov(\widehat\bbb_j)\}^{-1/2}$ rather than $\widehat\bbb_j$ itself, where $\cov(\widehat\bbb_j) \in {\mathbb R}^{q \times q}$ denotes the invertible covariance matrix of $\widehat\bbb_j$. To simplify our subsequent derivations, we suppose $\eE\{\bX_t(u)\}={\bf 0} $ and $\cov\{\be_{tk}(u)\} = \bI_p$ uniformly across all $t,$ $k$, and $u\in\cU$. By (\ref{lse_B}), it follows that:
\begin{equation}
\label{cov_B}
\begin{split}
	\cov(\widehat\bbb_j) &= \Big\{\sum_{t=k+1}^n \bQ^\T\bY_{t-k}(v)\bY_{t-k}^\T(u)\bQ \Big\}^{-1}\Big\{\sum_{t=k+1}^n \bQ^\T\bY_{t-k}(v) \bY_{t-k}^\T(v)\bQ\Big\} \\
	&\quad \cdot\Big\{\sum_{t=k+1}^n \bQ^\T\bY_{t-k}(u)\bY_{t-k}^\T(v)\bQ \Big\}^{-1} := \widehat \bSigma_{u,v},
\end{split}
\end{equation}
indicating that $\cov(\widehat\bbb_j)$ remains the same across $j$. Our primary objective is to identify the rank of $\widehat\bB_{u,v}$ and to recover the space spanned by the columns of $\widehat\bB_{u,v}.$ Motivated by the normalized idea, this can be achieved through an eigenanalysis of the scaled nonnegative-definite and symmetric matrix $\widehat\bB_{u,v} \widehat \bSigma_{u,v}^{-1} \widehat\bB_{u,v}^{\T}.$ 
Combining (\ref{lse_B}) and (\ref{cov_B}) yields that
\begin{equation}
\label{scale_B}
	\widehat\bB_{u,v} \widehat \bSigma_{u,v}^{-1}\widehat\bB_{u,v}^\T = \widehat\bSigma^{(k)}_{yy}(u,v)\bQ\{\bQ^\T \widehat\bSigma_{yy}^{(0)}(v,v)\bQ\}^{-1}\bQ^\T\widehat\bSigma^{(k)}_{yy}(u,v)^\T. 
\end{equation}
To accumulate the information in (\ref{scale_B}) as much as possible, we integrate its right-hand side over $(u,v) \in \cU\times\cU$ and sum it over time lags $k=1, \dots, k_0,$ which culminates in (\ref{eq:Mhat}) with $\widehat\bW(v)$ as defined in (\ref{hat.W0}).


An alternative choice of the weight matrix is $\bW_1(v) = \bI_p$ for $v \in \cU,$ 
where the homogeneous weights are assigned. As a consequence, we need to perform an eigenanalysis on the estimated nonnegative definite and symmetric matrix:
\begin{equation}
\label{estM1}
    \begin{aligned}
        \widehat\bM_{1} =\sum_{k=1}^{k_0}\int_{\cU}\int_{\cU} \widehat\bSigma^{(k)}_{yy}(u,v) \widehat\bSigma^{(k)}_{yy}(u,v)^\T\,\mathrm d u \mathrm d v,
\end{aligned}
\end{equation}
which can be further simplified by integrating along the diagonal path $u=v \in \cU:$
\begin{equation}
\label{estM2}
    \begin{aligned}
        \widehat\bM_{2} =\sum_{k=1}^{k_0}\int_{\cU} \widehat\bSigma^{(k)}_{yy}(u,u)\widehat\bSigma^{(k)}_{yy}(u,u)^\T\,\mathrm d u.
\end{aligned}
\end{equation}
\begin{remark}
\label{rmk.weight}
It is noteworthy that, without the double or single integral, the unweighted estimators $\widehat\bM_{1}$ and $\widehat\bM_{2}$ in (\ref{estM1}) and (\ref{estM2}) coincide with the proposed method for multivariate scalar time series in \cite{lam2012}. Compared with the weighted estimator $\widehat\bM,$ 
the performance of $\widehat\bM_{1}$ is expected to deteriorate 
especially for the heterogeneous case as illustrated in our simulations. 
Moreover, due to the loss of autocovariance information for $u \neq v$ incurred by the single integral, we expect that $\widehat\bM_1$ outperforms $\widehat\bM_2.$ We will compare the sample performance of  $\widehat\bM,$ $\widehat\bM_{1}$ and $\widehat\bM_{2}$ in Section~\ref{sec:sim}.
\end{remark}

\section{Theoretical Properties}
\label{sec:theory}
In this section, we study asymptotic properties of the proposed method under a high-dimensional regime, where $p$ and $n$ tend to infinity together and $r$ is fixed. Before presenting the theoretical results, we impose some regularity conditions. 

\begin{condition}\label{con:white}
    The idiosyncratic component $\{\bvarepsilon_t(\cdot)\}$ is a white noise sequence with $\bSigma_{\varepsilon\varepsilon}^{(k)}(u,v)=0$ for any $(u,v) \in \cU \times \cU$ and $k=1, \dots, k_0.$
\end{condition}

\begin{condbis}{con:white}\label{con:bound}
    The idiosyncratic component $\{\bvarepsilon_t(\cdot)\}$ is second-order weakly stationary with $\sup_{(u,v)\in\cU\times\cU}\Vert\bSigma_{\varepsilon\varepsilon}^{(k)}(u,v)\Vert=O(1)$ as $p\to\infty$ for $k=1,\dots,k_0.$
\end{condbis}

The white noise assumption in Condition~\ref{con:white} is imposed to facilitate the methodological development in Section~\ref{sec:method}. However, our theoretical results can be established under a weaker Condition~\ref{con:bound}, which allows for weak serial correlations in $\{\bvarepsilon_t(\cdot)\}.$


\begin{condition}\label{con:x}
    (i) The latent functional factor process $\{\bX_t(\cdot)\}$ is second-order weakly stationary with $ \eE(\|X_{tj}\|_{2}^2)=O(1)$ for $j=1,\dots,r;$ 
	\noindent(ii) There exists at least one $k \in \{1, \dots, k_0\}$ such that the rank of  $\iint \bSigma^{(k)}_{xx}(u,v)\bSigma^{(k)}_{xx}(u,v)^{\T}\,\mathrm d u\mathrm d v$ is $r.$ 
\end{condition}

\begin{condition}\label{con:eps}
    (i) 
    $\max_j {\eE} (\|\varepsilon_{tj}\|_{2}^2)=O(1);$ 
    (ii) $\inf_{v\in \cU}\lambda_{\min}\big\{\bSigma^{(0)}_{\varepsilon \varepsilon}(v,v)\big\}$ is bounded away from zero.
\end{condition}

\color{black}

Conditions~\ref{con:x} and \ref{con:eps} contain some standard finite moment assumptions in functional data analysis literature. 
Condition~\ref{con:x}(ii) can be viewed as functional generalization of Condition~2 in \cite{wang2019} for matrix-valued time series, which ensures that the latent factor process $\bX_t$ has exactly $r$ components.
Condition~\ref{con:eps}(ii) is imposed for technical convenience. 
It precludes the case when $\bY_{t}(v)$ is non-random at some point $v \in \cU.$
However, replacing $\bY_{t}(\cdot)$ with a contaminated process $\bY_{t}(\cdot) + \bdelta_{t},$ where $\bdelta_{t}$'s are independent with zero mean and 
diagonal covariance matrix with small diagonal components and are independent of $\bY_{t+k}(\cdot)$'s for all $k,$ Condition~\ref{con:eps}(ii) is then satisfied while the autocovariance structure in $\bSigma_{yy}^{(k)}$ remains the same
in the sense of  $\cov\{\bY_{t+k}(u) + \bdelta_{t+k}, \bY_{t}(v)+\bdelta_{t}\}=\cov\{\bY_{t+k}(u), \bY_{t}(v)\}$ for $k \geq 1$ and $u,v \in \cU.$

\begin{condition}\label{con:strg}
     There exists some constant $\delta\in[0,1]$ such that $\|\bA\|^2\asymp p^{1-\delta}\asymp\|\bA\|^2_{\min}$.
\end{condition}

\begin{condition}\label{con:xeps}
    (i) For $k=1, \dots, k_0,$ $\|\bSigma^{(k)}_{x\varepsilon}\|_{{\cS},\infty} = o(p^{(1-\delta)/2});$ 
    (ii) $\cov\{\bX_t(u), \bvarepsilon_{t+k}(v)\} = {\bf 0}$ for any $k\geq 0$ and $(u,v) \in \cU\times\cU.$ 
\end{condition}

The parameter $\delta$ in Condition~\ref{con:strg} can be viewed as the strength of factors with smaller values yielding stronger factors.
It measures the relative growth rate of the amount of information contained in $\bZ_t(\cdot)$ as the dimension $p$ increases, compared to that in $\bvarepsilon_t(\cdot),$ see (\ref{eq:model}).
When $\delta=0$, Condition~\ref{con:strg} corresponds to the pervasiveness assumption in \cite{fan2013}, which means that all factors are strong. When $\delta >0,$ the factors are termed as weak factors. 
Condition~\ref{con:xeps}(i) requires that the correlation between $\bX_{t+k}(\cdot)$ and $\bvarepsilon_t(\cdot)$ is not too strong, while Condition~\ref{con:xeps}(ii) assumes that the future idiosyncratic components are uncorrelated with the common factors up to the present.

\begin{condition}\label{con:rootn}
    Let $\widehat{\bSigma}^{(k)}_{yy}(u,v) = \big\{\widehat\Sigma^{(k)}_{yy,ij}(u,v)\big\}_{p \times p}$ and 
    ${\bSigma}^{(k)}_{yy}(u,v) = \big\{\Sigma^{(k)}_{yy,ij}(u,v)\big\}_{p \times p}$ for $(u,v) \in \cU \times \cU.$ Then (i) For $k=1,\dots,k_0,$
    $\max_{i,j}{\eE} \big\{\|\widehat{\Sigma}_{yy,ij}^{(k)}-\Sigma_{yy,ij}^{(k)}\|_{\cS}\big\} = O(n^{-1/2});$ 
    (ii) $\bSigma^{(0)}_{yy}(u,v)$ is Lipschitz-continuous over $(u,v)\in \cU\times\cU;$
    (iii) $\max_{i,j}{\eE}\big\{\sup_{(u,v) \in \cU\times\cU}\big|\widehat{\Sigma}_{yy,ij}^{(0)}(u,v)-\Sigma_{yy,ij}^{(0)}(u,v) \big|\big\} = O\big\{ n^{-1/2}(\log n)^{1/2}\big\}.
    $
\end{condition}

There are several sufficient conditions that have been commonly imposed in functional time series literature when $p$ is fixed and can lead to the convergence result under expectation in Condition~\ref{con:rootn}(i), which implies 
$\|\widehat{\Sigma}_{yy,ij}^{(k)}-\Sigma_{yy,ij}^{(k)}\|_{\cS}=O_\P(n^{-1/2}).$
The key requirement in establishing the convergence is to control the temporal dependence in $\{Y_{tj}(\cdot)\}$ for each $j.$ Examples include strong mixing conditions \cite[]{Bbosq1,qiao2019c}, cumulant mixing conditions \cite[]{panaretos2013} and $L^q$-$m$-approximability \cite[]{hormann2010,hormann2015}. When $p$ diverges, we assume the rate under deterministic expectation uniformly in Condition~\ref{con:rootn}(i).
With sub-Gaussian-type tail assumption and Lipschitz-continuity in Condition~\ref{con:rootn}(ii), it is not difficult to apply the partition technique that reduces the problem from supremum over $\cU\times\cU$ to the maximum over a grid of pairs, then the uniform concentration inequalities and hence the uniform rate of $n^{-1/2}(\log n)^{1/2}$ in Condition~\ref{con:rootn}(iii) can be achieved. See the same uniform rate with detailed proof under an i.i.d. setting in \cite{qiao2019b}.

Now we are ready to present theorems about the rates of convergence for estimators of the factor loading space ${\cal C}(\bA)$ and the eigenvalues $\{\nu_j\}_{j=1}^p.$ To measure the accuracy in estimating ${\cal C}(\bA),$ we use the metric of the distance between ${\cal C}(\bA)$ and ${\cal C}(\widehat\bA).$
For two orthogonal matrices $\bK_1$ and $\bK_2$ of dimensions $p\times r_1$ and $p \times r_2,$ respectively, we define
\begin{equation*}\label{distance}
    \mathcal D\big(\mathcal C(\bK_1), \mathcal C(\bK_2)\big) = \Big\{1- \frac{1}{\max(r_1,r_2)}\tr(\bK_1 \bK_1^{\T} \bK_2\bK_2^{\T})\Big\}^{1/2}.
\end{equation*}
This distance ranges between $0$ and $1.$ It equals $0$ if and only if $\mathcal C(\bK_1)=\mathcal C(\bK_2),$ and $1$ if and only if $\bK_1$ and $\bK_2$ are orthogonal. 

\begin{theorem}\label{thm:eigvec}
Let Conditions~\ref{con:bound}, \ref{con:x}--\ref{con:rootn} hold and $p^{\delta}n^{-1/2} \to 0.$ Suppose that $r$ is known. Then as $p, n \rightarrow \infty,$ it holds that
     $       \mathcal D \big(\mathcal C(\bA), \mathcal C(\widehat \bA)\big) = O_\P(p^{\delta}n^{-1/2}+p^{\delta-1}),$
    which reduces to $O_\P(p^{\delta}n^{-1/2})$ if Condition~\ref{con:white} instead of Condition~\ref{con:bound} holds.
\end{theorem}
\begin{remark}
    (i) Theorem~\ref{thm:eigvec} implies that, as the factors become stronger with smaller $\delta,$ the estimated factor loading space achieves a faster rate of convergence. The presence of $p^{\delta-1}$ in the established rate, under Condition~\ref{con:bound}, arises from the distance between ${\cal C}(\bK)$ and ${\cal C}(\bA),$ which equals zero under the white noise assumption. If Condition~\ref{con:white} holds or $p\gtrsim n^{1/2},$ our rate reduces to $O_\P(p^{\delta}n^{-1/2}),$ which is consistent with the result under the white noise assumption in \cite{lam2012}.
    
    (ii) If Condition~\ref{con:white} holds, 
    when all factors are strong (i.e., $\delta=0$), then the $\sqrt{n}$ rate is attained, since the signal is as strong as the noise and hence enlarging $p$ will not affect the estimation efficiency, circumventing the phenomenon of ``curse of dimensionality". For weak factors (i.e., $\delta>0$), the noise increases faster than the signal, and  the increase in $p$ will result in a slower convergence rate.
\end{remark}

\begin{theorem}\label{thm:eigval}
    Let Conditions~\ref{con:bound}, \ref{con:x}--\ref{con:rootn} hold and $p^{\delta}n^{-1/2} \to 0.$ Then as $p, n \rightarrow \infty,$ the following assertions hold:\\
    (i) $|\hat{\nu}_j-\nu_j| = O_\P(p^{2-\delta}n^{-1/2})$ for $j=1,\dots,r,$ and $\hat{\nu}_j=O_\P(p^2n^{-1}+1)$ for $j=r+1,\dots,p$;\\
    (ii) $\hat{\nu}_{j+1}/\hat\nu_j\asymp 1$ with probability tending to one for $j=1,\dots,r-1$, and $\hat{\nu}_{r+1}/\hat\nu_r=O_\P(p^{2\delta}n^{-1}+p^{2\delta-2})$.
\end{theorem}

\begin{remark}

Part~(ii) of Theorem~\ref{thm:eigval} implies that the eigen-ratio $\hat{\nu}_{j+1}/\hat\nu_j$ will drop steeply at $j=r,$ thus providing partial theoretical support for the proposed ratio-based estimator $\hat r$ in (\ref{eq:fac-num}). Under Condition~\ref{con:white} or $p\gtrsim n^{1/2}$, it holds that $|\hat{\nu}_j-\nu_j| = O_\P(p^{2-\delta}n^{-1/2})$ for $j=1,\dots,r,$ and $\hat{\nu}_j=O_\P(p^2n^{-1})$ for $j=r+1,\dots,p$, which together imply that $\hat{\nu}_{j+1}/\hat\nu_j\asymp 1$ with probability tending to one for $j=1,\dots,r-1$, and $\hat{\nu}_{r+1}/\hat\nu_r=O_\P(p^{2\delta}n^{-1})$. When all factors are strong (i.e., $\delta=0$), we obtain $\hat \nu_{r+1}/\hat \nu_r = O_\P(n^{-1}),$ suggesting that $\hat r$ may not suffer from the increase in $p.$ See similar results in \cite{lam2012}. 
\end{remark}

To facilitate the consistency analysis of the ratio-based estimator $\hat r$ for $r$ and to avoid the case of ``0/0", we define a modified ratio-based estimator
\begin{equation}
\label{mod.ratio.est}
\hat r = \arg\min_{1\leq j\leq p} \frac{\hat{\nu}_{j+1} + \vartheta_n }{\hat{\nu}_j + \vartheta_n},
\end{equation}
where $\vartheta_n$ provides a lower bound
correction to $\hat \nu_j$ for $j >r$ and satisfies the conditions in Theorem~\ref{thm.r} below.

\begin{theorem}
\label{thm.r}
Let the conditions of Theorem~\ref{thm:eigval} hold, $\vartheta_np^{-2+2\delta} \rightarrow 0,\vartheta_nn^2p^{-2-2\delta} \to \infty$ and $\vartheta_np^{2-2\delta}\to \infty.$ Then as $p, n \rightarrow \infty,$ it holds that
$\pr(\hat r = r) \to 1.$
\end{theorem}

\begin{remark}
(i) Theorem~\ref{thm.r} shows that the modified $\hat r$ in (\ref{mod.ratio.est}) is a consistent estimator of $r$ 
In practice, provided that $\vartheta_n$ is usually hard to be specified, we still use (\ref{eq:fac-num}) to estimate $r,$ leading to good performance in our empirical studies.

(ii) With the aid of Theorem~\ref{thm.r}, our estimation of ${\cal C}(\bA)$ is asymptotically adaptive to $r.$ To this end, let $\widecheck \bA =(\widehat \bgamma_1, \dots, \widehat\bgamma_{\hat r})$ and ${\cal C}(\widecheck\bA)=\text{span}\{\widehat \bgamma_1, \dots, \widehat \bgamma_{\hat r}\}$ be the estimator of ${\cal C}(\bA)$ with $r$ estimated by $\hat r.$ Then it holds that for any constant $C>0$ that
\begin{eqnarray*}
\pr\big(p^{-\delta}n^{1/2}{\cal D}\big( {\cal C}(\bA), {\cal C}(\widecheck\bA)\big)>C\big)
&\leq & \pr\big(p^{-\delta}n^{1/2}{\cal D}\big( {\cal C}(\bA), {\cal C}(\widecheck\bA)\big)>C | \hat r=r\big)
+ \pr(\hat r \neq r)\\
&\leq & \pr\big(p^{-\delta}n^{1/2}{\cal D}\big( {\cal C}(\bA), {\cal C}(\widehat\bA)\big)>C | \hat r=r\big) + o(1),
\end{eqnarray*}
which together with Theorem~\ref{thm:eigvec} yield $\mathcal D \big(\mathcal C(\bA), \mathcal C(\widecheck \bA)\big) = O_\P(p^{\delta}n^{-1/2}+p^{\delta-1})$.
\end{remark}

\section{Sparse Factor Model}
\label{sec:ultrahigh}
Despite the phenomenon of ``curse of dimensionality" being avoided when all factors are strong, our method does not guarantee a parsimonious and interpretable model in the presence of weak factors. In real applications, it is quite common that many estimated factor loadings are close to zero especially when $p$ is large, see the gene expression study in \cite{carvalho2008high} and the sea surface air pressure records example in \cite{lam2012}. 
Such phenomenon motivates us to propose a sparse factor model 
by imposing sparsity assumptions on the factor loading space $\cC(\bA)$ of (\ref{eq:model}).  
In this section, we target to develop the regularized estimation under sparsity constraints, which not only leads to the enhanced intepretability in practice, but also theoretically enlarges the dimension of the feature space that our proposed sparse factor model can handle compared with the nonsparse factor model with weak factors. 

We consider two complementary notions of subspace sparsity defined in terms of the factor loading matrix $\bA$: row sparsity and column sparsity, which are consistent with the definitions in \cite{vu2013}. Intuitively, the row sparsity entails that only a small subset of components in $\bY_t(\cdot)$ are driven by the common factors $\bX_t(\cdot),$ thus making the row sparse factor loading space generated by a small number of variables, independent of the choice of the basis. The column sparsity, on the other hand, corresponds to the case that each common factor $X_{tj}(\cdot)$ has impact on only a small fraction of components of $\bY_t(\cdot)$ and hence the column sparse factor loading space has orthonormal basis consisting of sparse vectors. We begin by introducing parameter spaces of ``approximately column sparse" and ``approximately row sparse" factor loading matrices respectively defined in Conditions~\ref{con:colsparse} and \ref{con:rowsparse} below. Let $\ba_i$ be the $i$-th row vector of $\bA = (A_{ij})_{p \times r}.$ 

\begin{condition}[Column sparsity]\label{con:colsparse}
    The matrix $\bA$ is from the class 
    $$\cV(\tau,c_1(p), L) = \Big\{\bA: \max\limits_{1\leq j\leq r}\sum_{i=1}^p|A_{ij}|^\tau\leq c_1(p),\ \max\limits_{1\leq i\leq p, 1 \leq j \leq r}|A_{ij}|\leq L \Big\},\quad \tau\in [0, 1).$$
\end{condition}

\begin{condbis}{con:colsparse}[Row sparsity]\label{con:rowsparse}
 The matrix $\bA$ is from the class 
    $$\cV^*(\tau,c_1(p), L) = \Big\{\bA: \sum_{i=1}^p \|\ba_i\|^\tau\leq c_1(p),\  \max\limits_{1\leq i\leq p, 1 \leq j \leq r}|A_{ij}|\leq L \Big\},\quad\tau\in [0, 1).$$
\end{condbis}

The parameters $c_1(p)$ and $\tau$ together control the column or row sparsity in the factor loading matrix $\bA.$ In the special case of $\tau=0,$ $\cV(0,c_1(p), L)$ and $\cV^*(0,c_1(p), L)$ correspond to the truly sparse situations, in which $\bA$ has at most $c_1(p)$ nonzero entries on each column under Condition~\ref{con:colsparse} or $c_1(p)$ nonzero row vectors under Condition~\ref{con:rowsparse}. By H\"{o}lder's inequality, it is easy to check that $\cV^*(\tau, c_1(p), L)\subseteq \cV(\tau, c_{r,\tau}c_1(p), L)$, where $c_{r,\tau} = 1$ for $\tau=0$ and $r^{1-\tau/2}$ for $\tau\in(0,1).$ Provided that $r$ is fixed, our subsequent analysis will focus on the class of column sparse factor loading matrices, $\cV(\tau,c_1(p), L).$

Our estimation procedure in Section~\ref{sec:method} is developed based on the estimation of $\bSigma_{yy}^{(k)}$'s. When $p$ grows faster than $n^{1/2},$ it 
is known that the sample autocovariance functions $\widehat\bSigma_{yy}^{(k)}$'s are no longer consistent estimators. Nevertheless, under the sparse factor model setup, the decomposition of $\bSigma_{yy}^{(k)}$ in (\ref{eq:cov-relation}) suggests that our enforced column sparsity in $\bA$ is inherited by the functional sparsity structure in $\bSigma_{yy}^{(k)},$ thus enabling us to possibly construct the threhsolding-based estimator for $\bSigma_{yy}^{(k)}$ to ensure the consistency. Before introducing the estimator, we slightly modify Condition~\ref{con:xeps} to accommodate the sparse setting and present a lemma that reveals the functional sparsity pattern in $\bSigma_{yy}^{(k)}.$

\begin{condbis}{con:xeps}\label{con:xeps2}
(i) Let $\bSigma^{(k)}_{x\varepsilon}(u,v) = \big\{\Sigma^{(k)}_{x\varepsilon, ij}(u,v)\big\}_{r \times p}$ and $\max\limits_{1\leq l\leq r} \sum_{j=1}^p \|\Sigma^{(k)}_{x\varepsilon,lj}\|^\tau_{\cS}\leq c_2(p)$ for $k=1,\dots, k_0$; (ii) $\cov\{\bX_t(u), \bvarepsilon_{t+k}(v)\}={\bf 0}$ for any $k\geq 0$ and $(u,v) \in \cU\times\cU$.
\end{condbis}

Condition~\ref{con:xeps2}(i) requires relatively weak correlations between $\bX_{t+k}(\cdot)$ and $\bvarepsilon_{t}(\cdot).$ As long as $c_2(p)=o(p^{(1-\delta)/2}),$ Condition~\ref{con:xeps}(i) follows directly from Condition~\ref{con:xeps2}(i). Moreover, Condition~\ref{con:xeps2}(i)
relies on the Hilbert-Schmidt norm to encourage the functional sparsity in $\bSigma_{x\varepsilon}^{(k)}$, i.e., each common factor $X_{(t+k)l}(\cdot)$ is only correlated with a few components of $\bvarepsilon_t(\cdot).$ With imposed sparsity constraints in Conditions~\ref{con:xeps2} and  \ref{con:colsparse}, it can be inferred from (\ref{eq:cov-relation}) 
that $\bSigma^{(k)}_{yy}$ is functional sparse.
Lemma~\ref{lem:sparse} in Section~\ref{supp.lem} of Supplementary Material shows that the functional sparsity patterns in columns/rows of $\bSigma_{yy}^{(k)}$ are determined by parameters $c_1(p)$ and $c_2(p)$ with smaller values yielding functional sparser $\bSigma_{yy}^{(k)}.$ 

To obtain a functional sparse estimator for $\bSigma_{yy}^{(k)},$ we apply the hard functional thresholding rule, which combines functional versions of hard thresholding and shrinkage 
based on the Hilbert-Schmidt norm of functions, on entries of the sample autocovariance function $\widehat\bSigma_{yy}^{(k)}.$ Then the functional thresholding estimator is constructed as
\begin{equation}\label{th-est}
    \mathcal T_{\eta_k}(\widehat \bSigma^{(k)}_{yy})(u,v) = \   \Big[\widehat{\Sigma}^{(k)}_{yy,ij}(u,v) I\big\{\|\widehat{\Sigma}^{(k)}_{yy,ij}\big\|_{\cS}\geq \eta_k\big\}\Big]_{p\times p},\quad u,v \in \cU,
\end{equation}
where $I(\cdot)$ is the indicator function and $\eta_k \geq 0$ is the thresholding parameter. 
Under mild regularity conditions, Lemma~\ref{lem:nonasym} in Section~\ref{supp.lem} of the Supplementary Material yields 
\begin{equation}
\label{linfty.Sigma}
 \max_{1\leq i,j\leq p} \|\widehat{\Sigma}^{(k)}_{yy,ij}-\Sigma^{(k)}_{yy,ij}\|_{\cS}=O_\P\left(\cM_y\sqrt{\frac{\log p}{n}}\right),   
\end{equation}
where $\cM_y$ is the functional stability measure \cite[]{guo2020consistency} defined in (\ref{df.fsm}) in  Section~\ref{ap_cond} of the Supplementary Material. 
The rate under the functional version of $\ell_{\infty}$ norm in (\ref{linfty.Sigma}) plays a crucial role in our theoretical analysis under an ultra-high-dimensional regime and, in particular, suggests us to set the thresholding level as $\eta_k \asymp \cM_y(n^{-1} \log p)^{1/2}.$
Replacing $\widehat\bSigma^{(k)}_{yy}$ in (\ref{eq:Mhat}) with $\cT_{\eta_k}(\widehat \bSigma^{(k)}_{yy}),$ we obtain the corresponding estimator for $\bM:$
\begin{equation}\label{eq:Mhat-thresh}
    \widetilde \bM = \sum_{k=1}^{k_0}\int \int \mathcal T_{\eta_k}(\widehat \bSigma^{(k)}_{yy})(u,v)\widehat\bW(v)\mathcal T_{\eta_k}(\widehat \bSigma^{(k)}_{yy})(u,v)^{\T}\,\mathrm d u\mathrm d v.
\end{equation}

To recover the column sparsity structure in $\cC(\bA),$ we perform sparse {principal component analysis} (PCA) \cite[]{vu2013} on $\widetilde \bM$ rather than the standard eigenanalysis for $\widehat \bM$ in Section~\ref{sec:est}. For matrices $\bA_1$ and $\bA_2$ with the same dimension, let $\langle\bA_1,\bA_2\rangle:=\text{trace}(\bA_1^\T\bA_2).$ We define $\widetilde \bK \in \eR^{p \times r}$ as a solution to the following optimization problem:
\begin{equation}
\label{eq:sparse_1}
    \widetilde \bK = \underset{\bK=(K_{jl})_{p \times r}}{\arg\max} ~~\langle \widetilde{\bM}, \bK \bK^{\T}\rangle
    ~~\text{subject to}~~ \bK^\T \bK = \bI_r,~\max\limits_{1\leq l\leq r}\sum_{j=1}^p|K_{jl}|^\tau \leq C_{\tau},
\end{equation}
where $C_{\tau}>0$ is a regularization parameter. Alternatively, to estimate the row sparse factor loading space, we can substitute the second constraint in (\ref{eq:sparse_1}) by $\sum_{i=1}^p \|\bk_i\|^{\tau} \leq \widetilde C_{\tau},$ where $\bk_i$ denotes the $i$-th row vector of $\bK$ and $\widetilde C_{\tau}>0$ is a regularization parameter. It is worth noting that, without the sparsity constraint in 
(\ref{eq:sparse_1}), the optimization problem degenerates to the ordinary PCA. Despite being challenging to solve \eqref{eq:sparse_1} due to the non-convex constraint, some efficient and computationally tractable algorithms have been developed, see, e.g., under the truly sparse case ($\tau=0$), the combinatorial approaches  \cite[]{moghaddam2006generalized,d2008optimal,mackey2009deflation}, the semi-definite relaxation \cite[]{d2007direct} and its variants, and the random-projection-based method \cite[]{gataric2020sparse}. We refer to \cite{zou2018selective} for review on recent developments for sparse PCA.

We now present the asymptotic analysis of $\cC(\widetilde \bK)$ in the following theorem. Note that $\cC(\bK)=\cC(\bA)$ under the white noise assumption.

\begin{theorem}\label{thm:hdeigspace}
    Let Conditions~\ref{con:white}, \ref{con:x}, \ref{con:eps}, \ref{con:xeps2}, \ref{con:colsparse} and \ref{con:flp}--\ref{con:fsm} in Section~\ref{ap_cond} of 
    the Supplementary Material hold and ${\cal M}_y^2 \log p n^{-1} \to 0.$ Then as $p,n \rightarrow \infty,$ it holds that:
    $$\nu_r\mathcal D\Big(\mathcal C(\bK), \mathcal C({\widetilde \bK})\Big) = O_\P\left[c_1(p)\big\{c_1(p)+c_2(p)\big\}\mathcal{M}_{y}^{1-\tau}\Big(\frac{\log p}{n}\Big)^{\frac{1-\tau}{2}}\right].$$
\end{theorem}
\begin{remark}
(i) The convergence rate of ${\mathcal D}\big(\mathcal C(\bK), \mathcal C({\widetilde \bK})\big)$ is governed by both dimensionality parameters $\{n, p, c_1(p), c_2(p)\}$ and internal parameters (${\cal M}_y, \nu_r, \tau$). It is easy to see that the rate is better when $\nu_r$ is large and $\{c_1(p), c_2(p), {\cal M}_y, \tau\}$ are small.

(ii) Under the truly sparse case ($\tau=0$) with $c_1(p)\gtrsim c_2(p),$ and $|A_{ij}|\asymp \gamma$ for $(i,j)$ such that $A_{ij} \neq 0,$ it follows from the framework in Section~\ref{sec:theory} that $c_1(p) \gamma^2 \asymp p^{1-\delta}$ under Condition~\ref{con:strg}, which together with Theorem~\ref{thm:hdeigspace} imply that 
\begin{equation}
\label{rate.thm3}
    \nu_r\mathcal D\big(\mathcal C(\bK), \mathcal C({\widetilde \bK})\big)=O_\P \big\{p^{2-2\delta} \gamma^{-4}{\cal M}_y (\log p)^{1/2} n^{-1/2}\big\}.
\end{equation}
By comparison, Theorem~\ref{thm:eigvec} and (\ref{nu_r_bound}) 
of Supplementary Material lead to the rate $\nu_r{\mathcal D}\big(\mathcal C(\bK), \mathcal C({\widetilde \bK})\big)=O_\P \big(p^{2-\delta}n^{-1/2}\big),$ which is slower than that in (\ref{rate.thm3}) for larger values of $\gamma$ or $\delta$ (i.e., smaller values of $c_1(p)$ provided that $\gamma \asymp 1$). Hence, when the magnitudes of nonzero entries in $\bA$ become larger or the factors are weaker in the sense of Condition~\ref{con:strg} (i.e., $\bA$ is sparser in the sense of Condition~\ref{con:colsparse}), our regularized estimation benefits more from the imposed sparsity and enjoys faster convergence rate than the ordinary method in Section~\ref{sec:method}.
\end{remark}

\section{Simulation Studies}
\label{sec:sim}
In each simulated scenario, we generate $p$-vector of functional time series by 
\begin{equation}
\label{eq:sim-model}
 \bY_t(\cdot) = \kappa_0 \bA \widebar \bX_t(\cdot) + \bvarepsilon_t(\cdot),~~t = 1,\dots,n,
\end{equation} 
where the parameter $\kappa_0>0$ controls the strength of common factors $\bX_t(\cdot)=\kappa_0 \widebar \bX_t(\cdot)$ and the entries of $\bA\in\mathbb R^{p\times r}$ are sampled from $\text{Uniform}[-\sqrt 3 p^{-\delta/2}, \sqrt 3 p^{-\delta/2}]$ with $\delta \in [0,1].$ Hence Condition~\ref{con:strg} is satisfied, 
in which $\delta=0$ (or $\delta>0$) corresponds to the case of strong (or weak) factors. 
To mimic the infinite-dimensionality of functional data, we generate each scaled latent factor by $\widebar X_{tl}(\cdot)=\sum_{i=1}^{50} \xi_{tli} \phi_{i}(\cdot)$ for $l=1, \dots, r$ over $\cU=[0,1],$ where $\{\phi_{i}(\cdot)\}_{i=1}^{50}$ is a $50$-dimensional Fourier basis function and the basis coefficients $\bxi_{ti}=(\xi_{t1i}, \dots, \xi_{tri})^{\T}$ are generated from a vector autoregressive model, $\bxi_{ti}=\bV\bxi_{(t-1)i} + \bepsilon_i$ with $\bV=(\rho^{|l-l'|+1})_{1 \leq l,l' \leq r}$ and the innovation $\bepsilon_{ti}=(\epsilon_{t1i}, \dots, \epsilon_{tri})^{\T}$
consisting of independent $N(0, i^{-2\iota})$ components. We set $\rho=0.45$ and $\iota=0.75.$ For the idiosyncratic component $\bvarepsilon_t(\cdot),$ we consider the following three scenarios.
\newcounter{bean}
\setcounter{bean}{0}
\begin{center}
	\begin{list}
		{\textsc{Scenario} \arabic{bean}.}{\usecounter{bean}}
		\item For each $j=1, \dots, p,$ we generate $\varepsilon_{tj}(\cdot) = \sum_{i=1}^{20} 2^{-(i+1)} \widetilde Z_{tji}\phi_i(\cdot),$ where $\widetilde Z_{tji}$'s are independent standard normal.
		
		\item We generate $\bvarepsilon_{t}(\cdot) = \bH\widebar\bvarepsilon_{t}(\cdot),$ where $\bH =5^{-1}\text{diag}(h_1,\dots,h_p),$ $h_j$'s are sampled uniformly from $\{1,\dots,10\}$ and each $\bar\varepsilon_{tj}(\cdot)$ is generated in the same way as $\varepsilon_{tj}(\cdot)$ in Scenario~1.
		
        \item We fix $\kappa_0=1$ and generate $\varepsilon_{tj}(\cdot) = \kappa_1\varepsilon_{t0}(\cdot) + \bar \varepsilon_{tj}(\cdot),$ where $\kappa_1 > 0,$ $\varepsilon_{t0}(\cdot) = \sum_{i=1}^4 \widebar Z_{ti}\phi_i(\cdot)$ and  each $\widebar Z_{ti}$ is sampled independently from $N(0,1).$ 

        \item For each $j=1, \dots, p,$ we generate $\varepsilon_{tj}(\cdot) = \sum_{i=1}^{20} 2^{-(i+1)} \widecheck Z_{tji}\phi_i(\cdot),$ where, for each $i=1, \dots, 20,$ $\{\widecheck Z_{tji}\}$ are generated from an AR(1) model $\widecheck Z_{tji} = 0.5 \widecheck Z_{(t-1)ji} + \theta_{tji}$ with $\theta_{tji}$'s being independent standard normal.
	\end{list}
\end{center}
In Scenario~1, each $\varepsilon_{tj}(\cdot)$ is white noise with the identical covariance function across $j,$ while Scenario~2 corresponds to the heterogeneous case with different covariance functions. Scenario~3 consists of $(r+1)$ actual common factors, among which the additional factor $\varepsilon_{t0}(\cdot)$ is independent of other factors in $\bX_t(\cdot),$ and its signal strength is determined by the parameter $\kappa_1.$ The underlying factor loading matrix is $(\bA, {\bf 1}_p),$ where ${\bf 1}_p$ denotes the $p$-vector of ones. For simplicity, we still denote the factor loading matrix by $\bA.$ Scenario~4 relaxes the white noise assumption by allowing for weak serial correlations in $\{\bvarepsilon_t(\cdot)\}.$


\subsection{Ordinary case}
\label{sim:ordi}
We compare our proposed method based on $\widehat \bM$ with two competing methods based on $\widehat \bM_1$ and $\widehat \bM_2.$ 
Since our experimental results suggest that $\hat r$ and $\mathcal C(\widehat \bA)$ are generally insensitive to the choice of $q$ (cannot be too small) and $k_0,$ we set $q=12$ and $k_0=4$. A sensitivity analysis can be found in Figure~\ref{fig-rob} in Section~\ref{supp.emp} of the Supplementary Material. In each setting, we generate $n=100$ serially dependent observations of $p=50,100,200$ functional variables based on $r=4$ functional factors. We ran each simulation $100$ times. The sample performance of three approaches is examined in terms of their abilities of correctly identifying the number of factors and the estimation accuracy in recovering the factor loading space, respectively measured by the relative frequency estimate for $\pr(\hat r=r)$ and $\mathcal D \big(\mathcal C(\bA), \mathcal C(\widehat \bA)\big)$ using the correct $r.$ For each of three comparison methods, Figure~\ref{fig1} and Figure~\ref{fig3} in Section~\ref{supp.emp} of the Supplementary Material plot average relative frequencies $\hat r=r$ as the factor strength increases when the factors are strong (i.e., $\delta=0$) and weak with $\delta=0.5$, respectively. 
Figure~\ref{fig2} and Figure~\ref{fig4} in Section~\ref{supp.emp} of the Supplementary Material plot the corresponding average estimation errors for $\mathcal C(\bA).$
See also Figure~\ref{fig-cross} in Section~\ref{supp.emp} for the scenario where the idiosyncratic components are cross-correlated.

\begin{figure}[!hbt]
\captionsetup[subfigure]{labelformat=empty}
\centering
\begin{subfigure}{0.32\linewidth}
  \caption{\scriptsize{$p=50$}}
  \vspace{-0.25cm}
  \includegraphics[width=5cm,height=3.5cm]{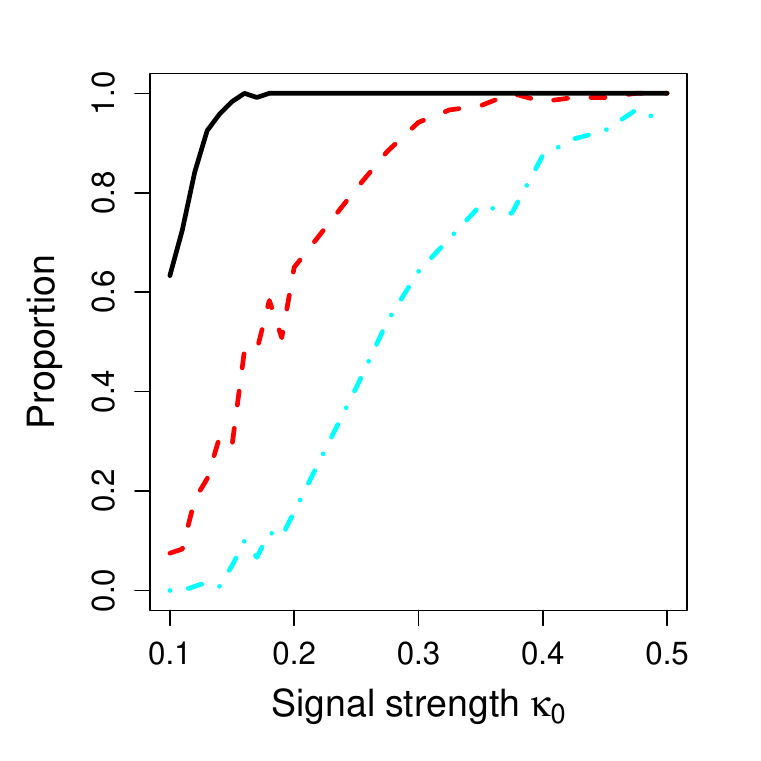}
  \label{fig:1-1-50}
  \vspace{-0.25cm}
  \caption{\scriptsize{$p=50$}}
  \vspace{-0.25cm}
  \includegraphics[width=5cm,height=3.5cm]{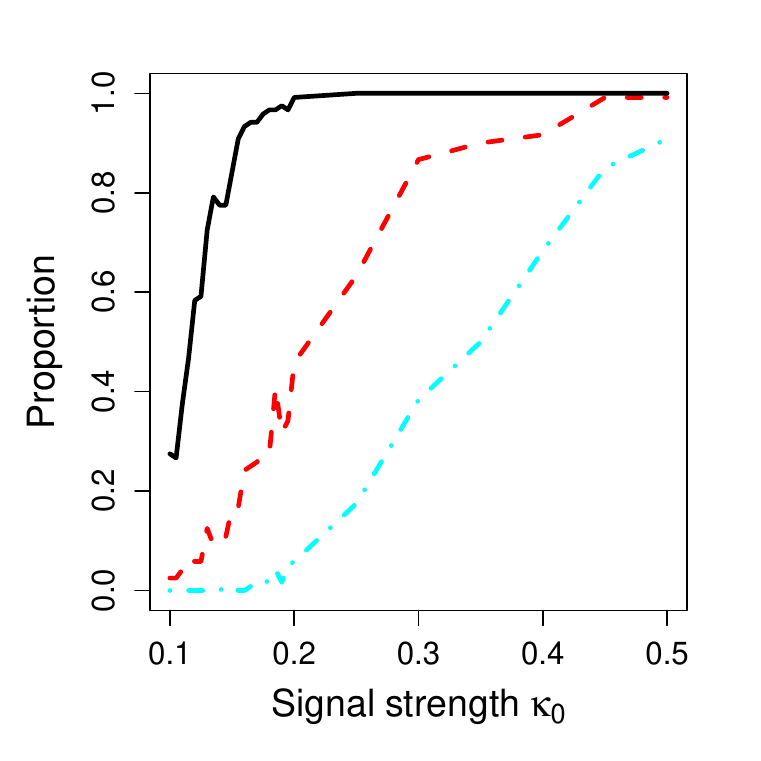}
  \label{fig:1-2-50}
  \vspace{-0.25cm}
  \caption{\scriptsize{$p=50$}}
  \vspace{-0.25cm}
  \includegraphics[width=5cm,height=3.5cm]{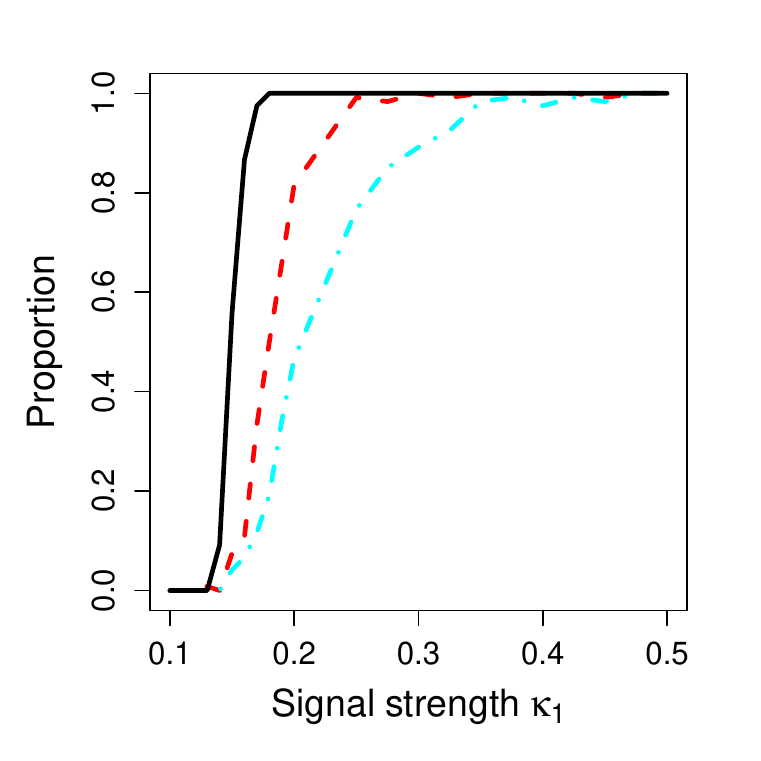}
  \label{fig:1-3-50}
  \vspace{-0.25cm}
  \caption{\scriptsize{$p=50$}}
  \vspace{-0.25cm}
  \includegraphics[width=5cm,height=3.5cm]{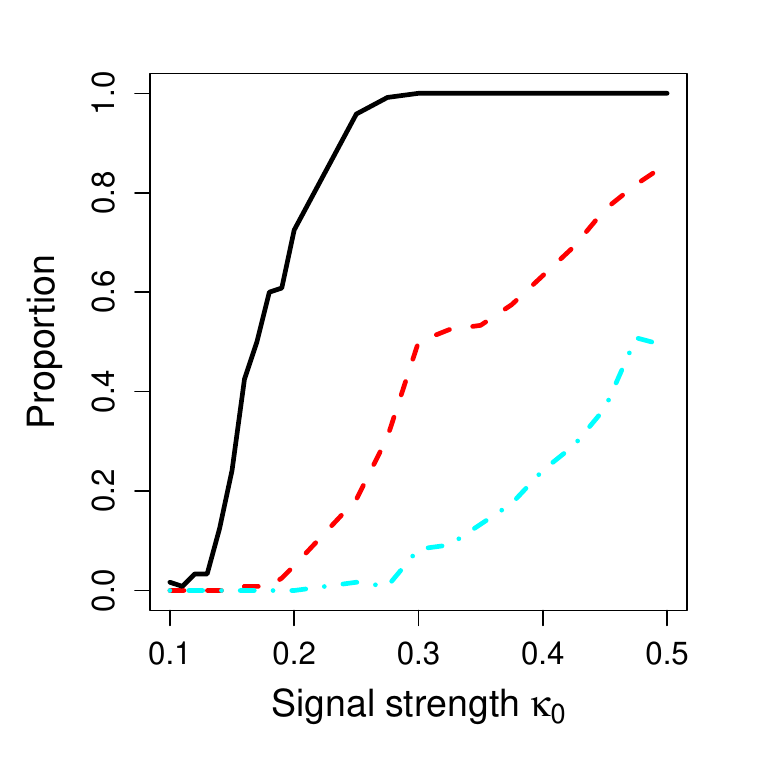}
  \label{fig:1-4-50}
\end{subfigure}
\begin{subfigure}{0.32\linewidth}
  \caption{\scriptsize{$p=100$}}
  \vspace{-0.25cm}
  \includegraphics[width=5cm,height=3.5cm]{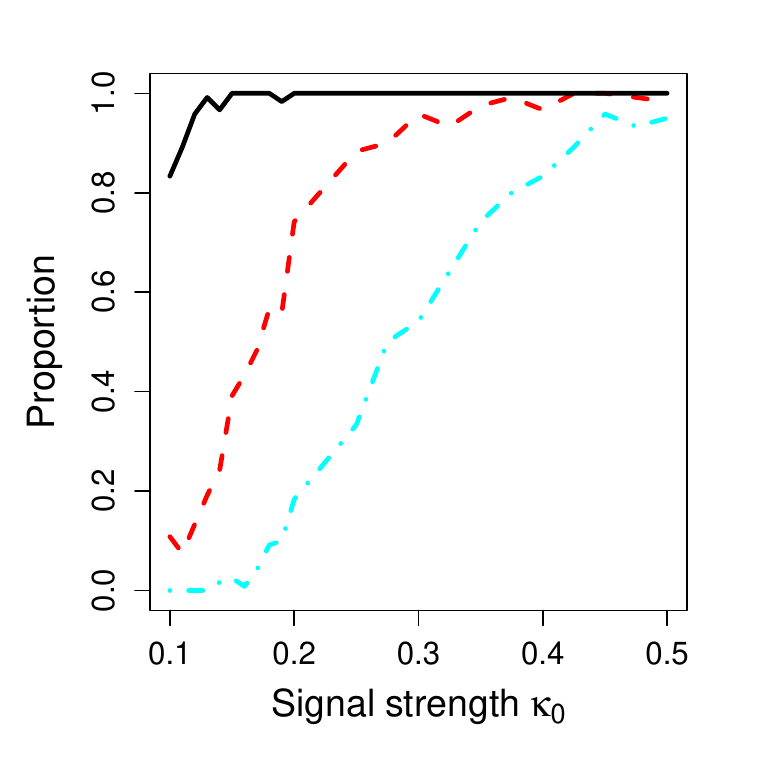}
  \label{fig:1-1-100}
  \vspace{-0.25cm}
  \caption{\scriptsize{$p=100$}} 
  \vspace{-0.25cm}
  \includegraphics[width=5cm,height=3.5cm]{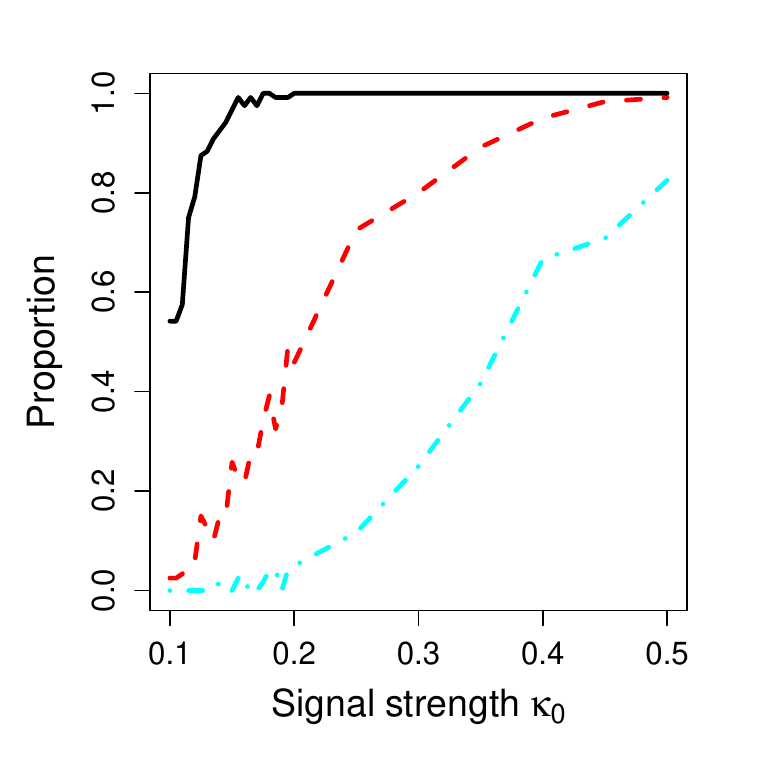}
  \label{fig:1-2-100}
  \vspace{-0.25cm}
  \caption{\scriptsize{$p=100$}}
  \vspace{-0.25cm}
  \includegraphics[width=5cm,height=3.5cm]{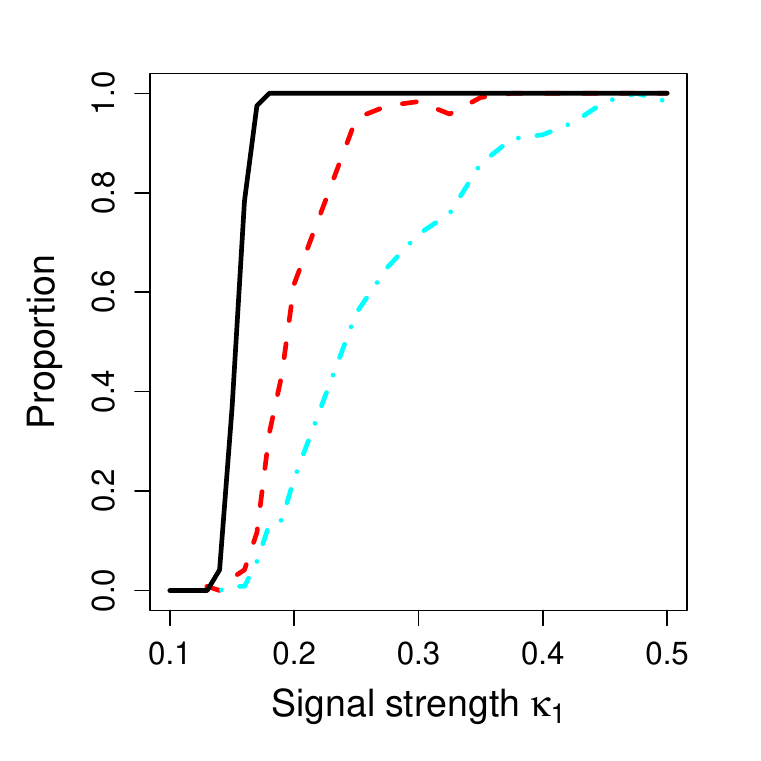}
  \label{fig:1-3-100}
  \vspace{-0.25cm}
  \caption{\scriptsize{$p=100$}}
  \vspace{-0.25cm}
  \includegraphics[width=5cm,height=3.5cm]{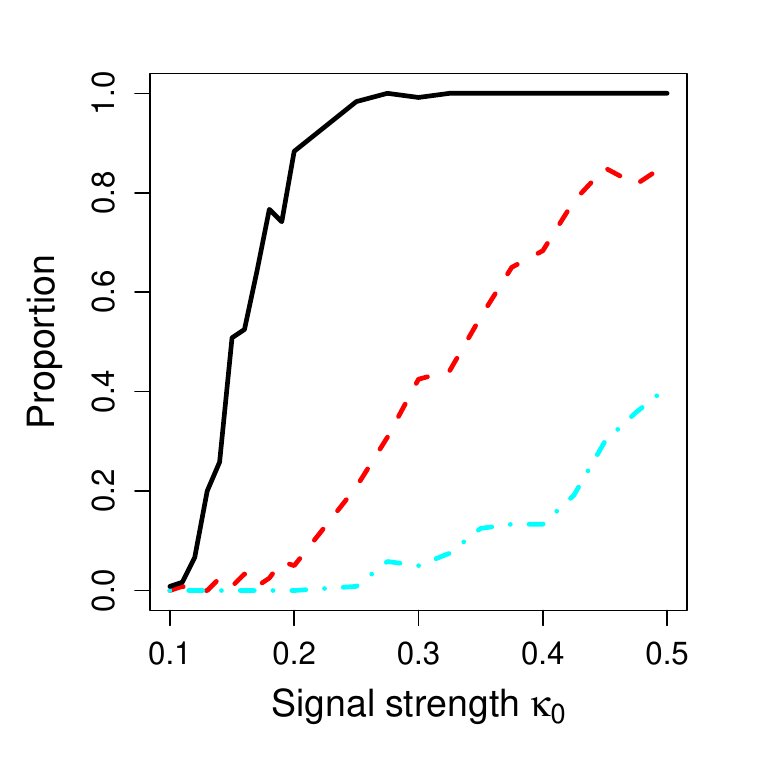}
  \label{fig:1-4-100}
\end{subfigure}
\begin{subfigure}{0.32\linewidth}
  \caption{\scriptsize{$p=200$}}
  \vspace{-0.25cm}
  \includegraphics[width=5cm,height=3.5cm]{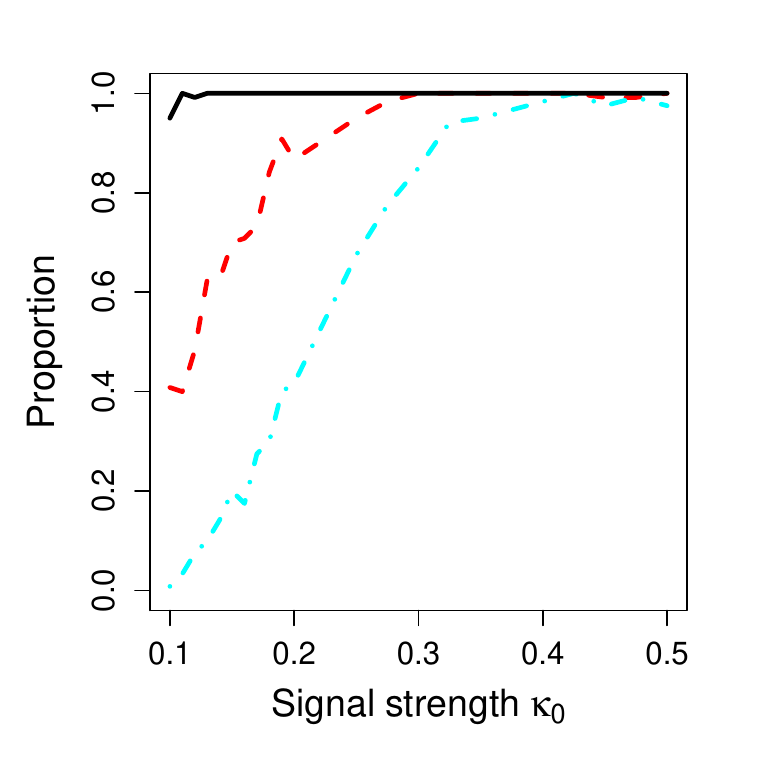}
  \label{fig:1-1-200}
  \vspace{-0.25cm}
  \caption{\scriptsize{$p=200$}}
  \vspace{-0.25cm}
  \includegraphics[width=5cm,height=3.5cm]{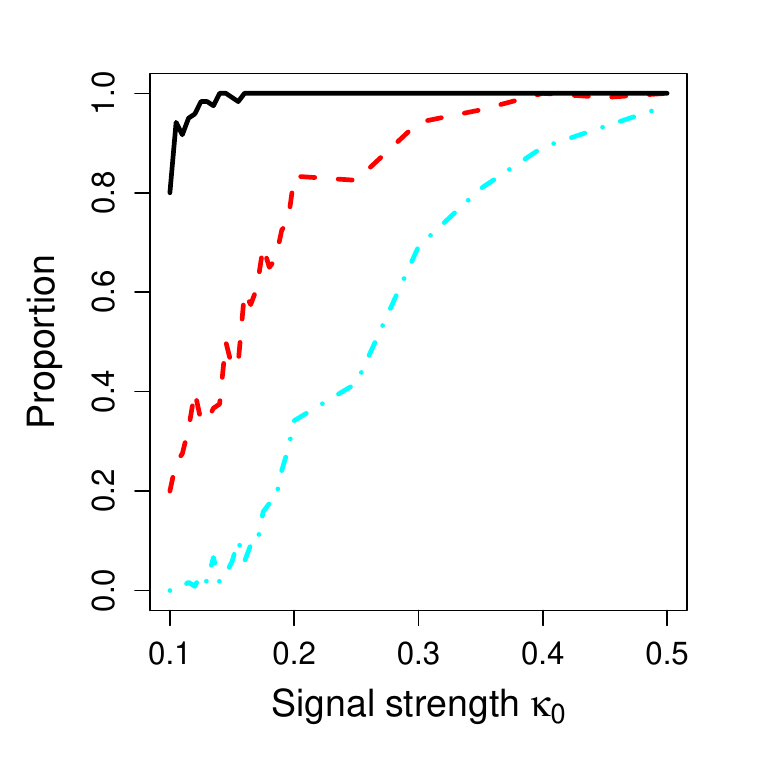}
  \label{fig:1-2-200}
  \vspace{-0.25cm}
  \caption{\scriptsize{$p=200$}}
  \vspace{-0.25cm}
  \includegraphics[width=5cm,height=3.5cm]{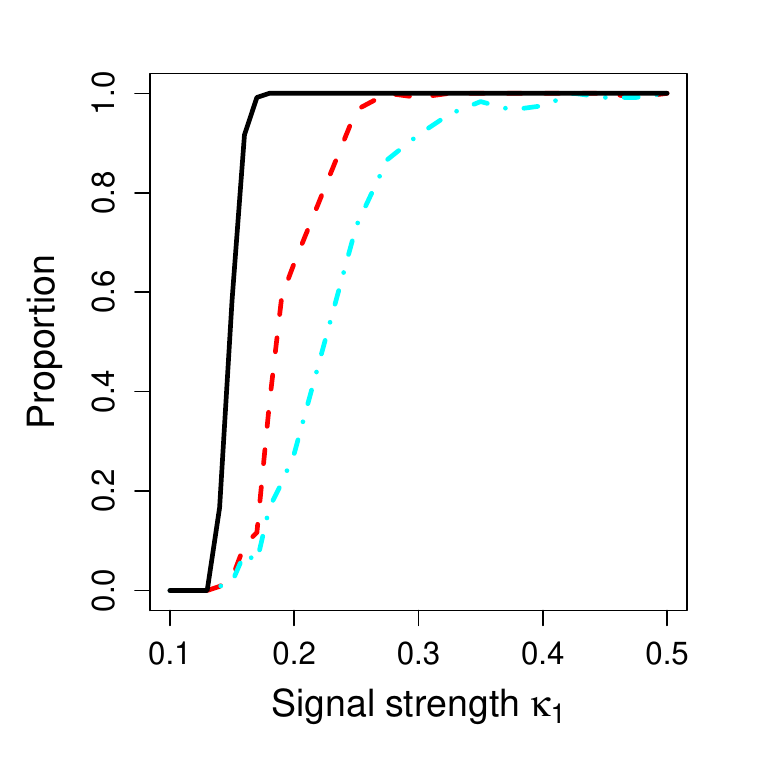}
  \label{fig:1-3-200}
  \vspace{-0.25cm}
  \caption{\scriptsize{$p=200$}}
  \vspace{-0.25cm}
  \includegraphics[width=5cm,height=3.5cm]{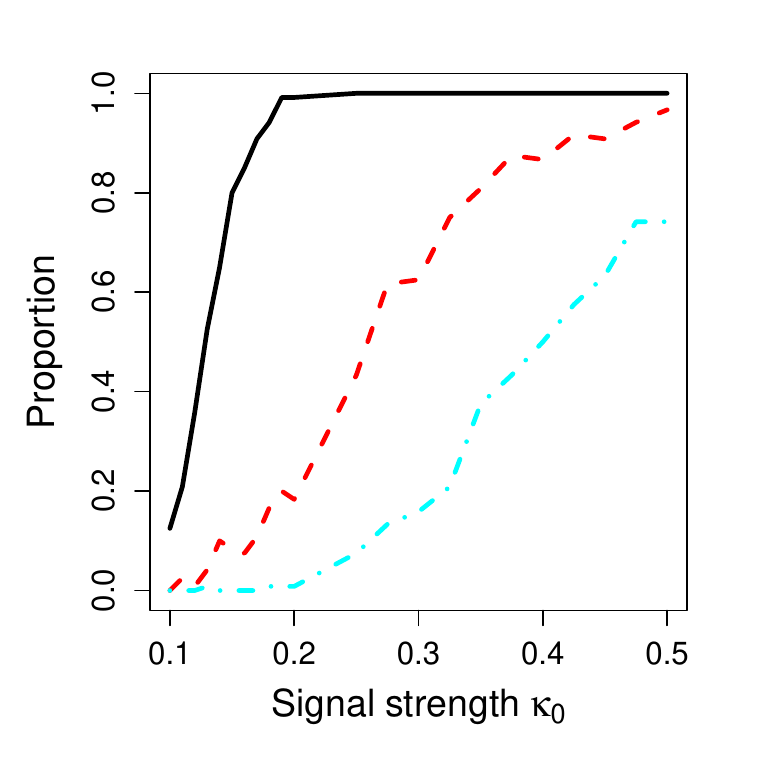}
  \label{fig:1-4-200}
\end{subfigure}
\centering
\begingroup
\linespread{1}
\caption{\label{fig1}{\it \footnotesize Scenario~1 (top row), Scenario~2 (2nd row), Scenario~3 (3rd row) and Scenario~4 (bottom row)} for $p=50,$ $100,$ $200$ when factors are strong: Plots of average relative frequency estimates for $\pr(\hat r = r)$ against $\kappa_0$ or $\kappa_1$ for three methods based on $\widehat\bM$ (black solid), $\widehat\bM_1$ (red dashed) and $\widehat\bM_2$ (cyan dash dotted).
}
\endgroup
\vspace{-0.2cm}
\end{figure}

\begin{figure}[!hbt]
\captionsetup[subfigure]{labelformat=empty}
\centering
\begin{subfigure}{0.32\linewidth}
  \caption{\scriptsize{$p=50$}}
  \vspace{-0.25cm}
  \includegraphics[width=5cm,height=3.5cm]{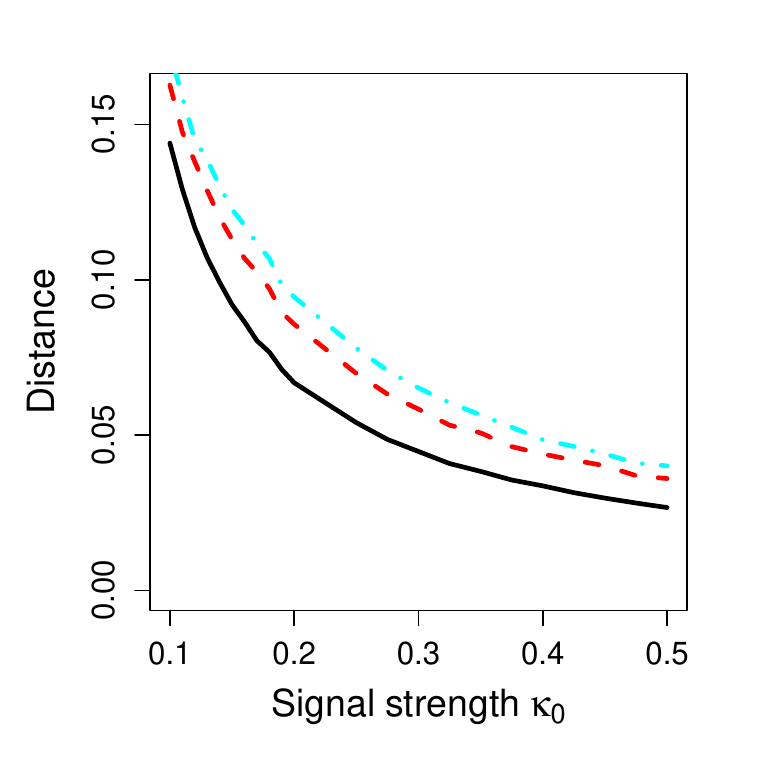}
  \label{fig:3-1-50}
  \vspace{-0.25cm}
  \caption{\scriptsize{$p=50$}}
  \vspace{-0.25cm}
  \includegraphics[width=5cm,height=3.5cm]{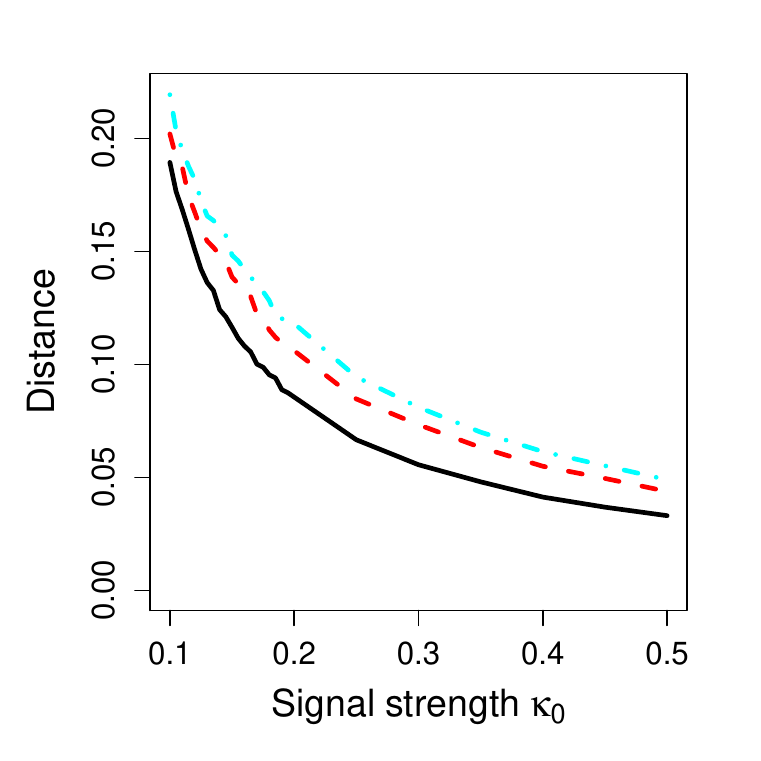}
  \label{fig:3-2-50}
  \vspace{-0.25cm}
  \caption{\scriptsize{$p=50$}}
  \vspace{-0.25cm}
  \includegraphics[width=5cm,height=3.5cm]{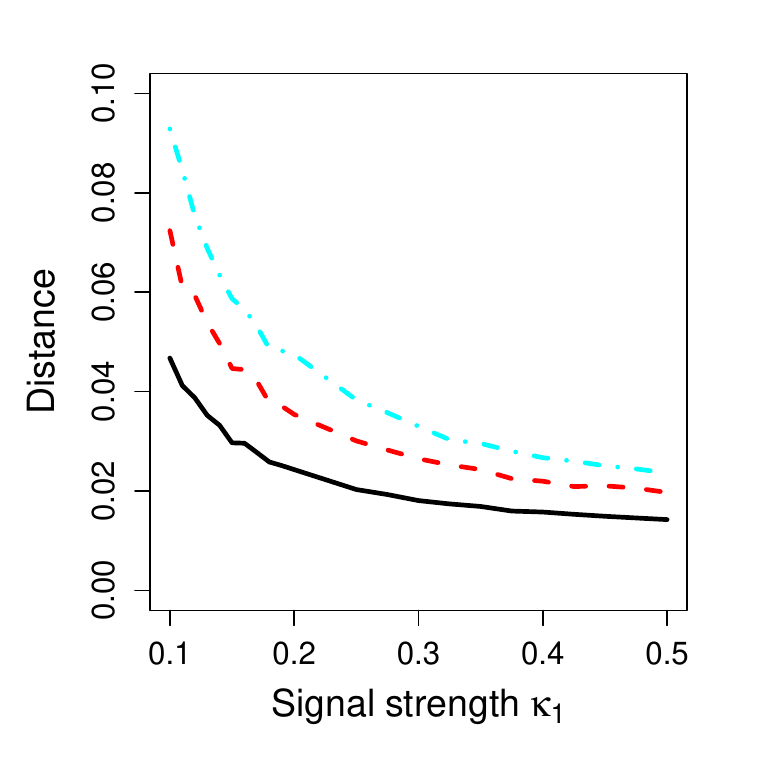}
  \label{fig:3-3-50}
  \vspace{-0.25cm}
  \caption{\scriptsize{$p=50$}}
  \vspace{-0.25cm}
  \includegraphics[width=5cm,height=3.5cm]{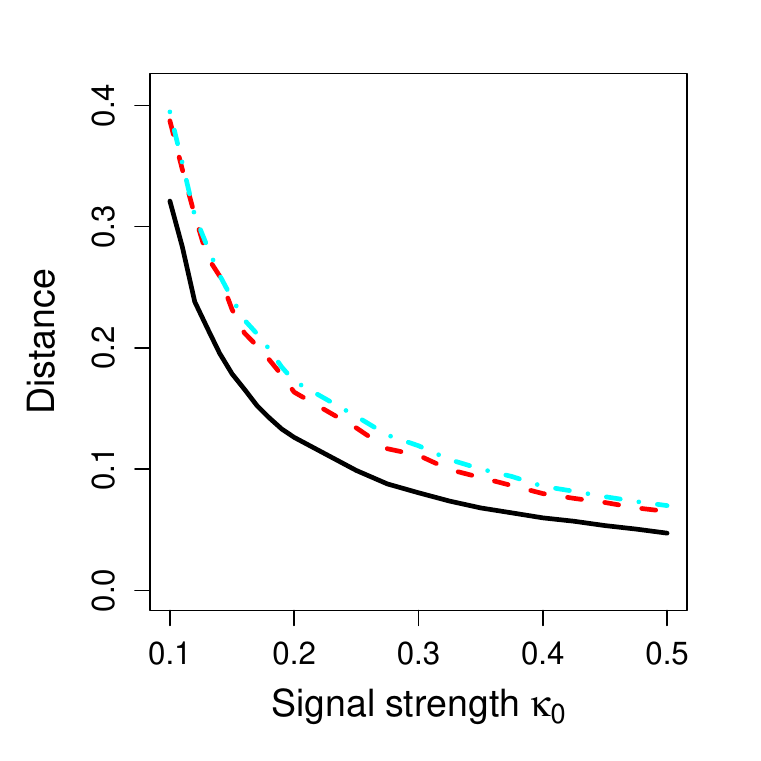}
  \label{fig:3-4-50}
\end{subfigure}
\begin{subfigure}{0.32\linewidth}
  \caption{\scriptsize{$p=100$}}
  \vspace{-0.25cm}
  \includegraphics[width=5cm,height=3.5cm]{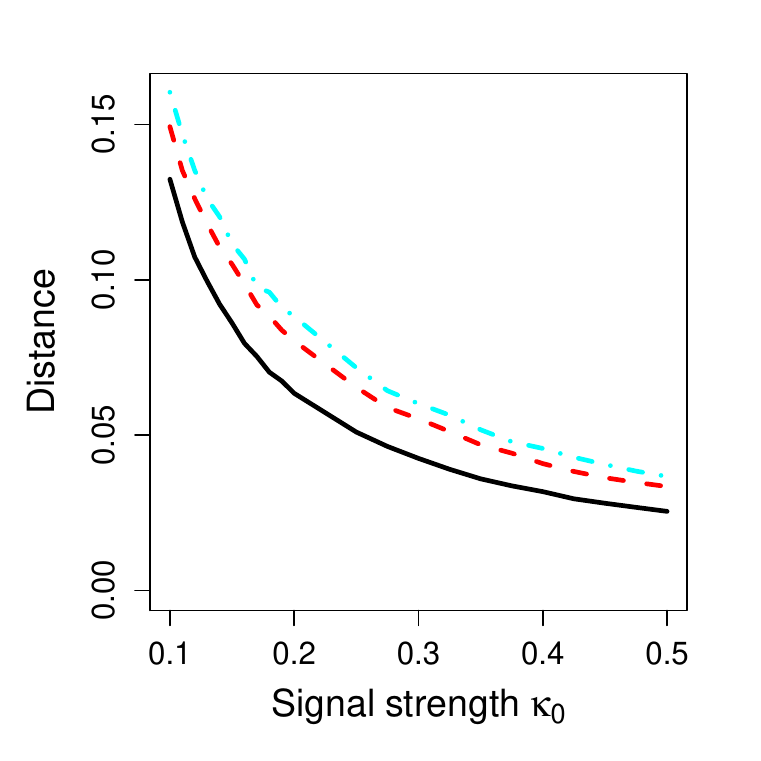}
  \label{fig:3-1-100}
  \vspace{-0.25cm}
  \caption{\scriptsize{$p=100$}} 
  \vspace{-0.25cm}
  \includegraphics[width=5cm,height=3.5cm]{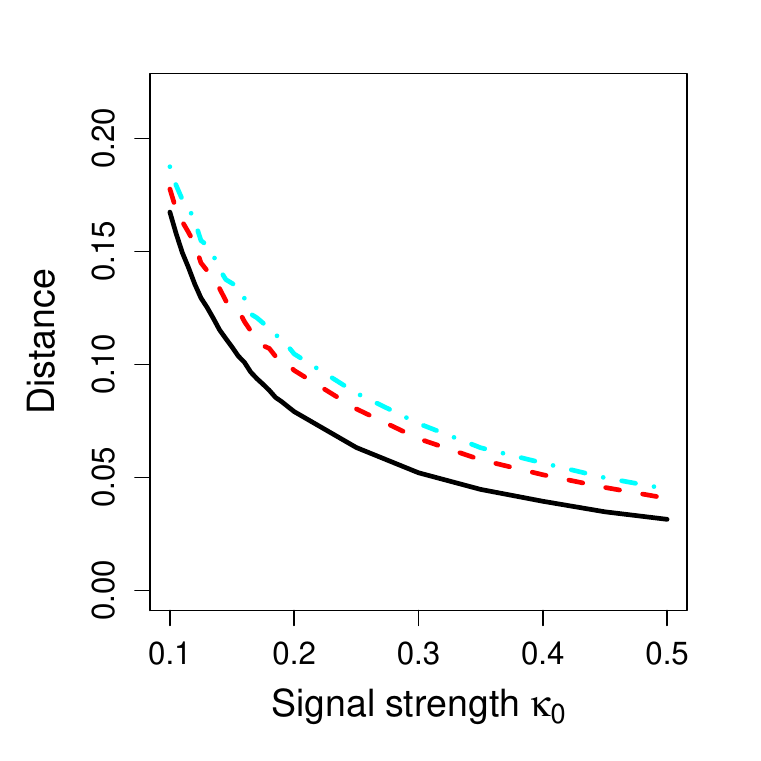}
  \label{fig:3-2-100}
  \vspace{-0.25cm}
  \caption{\scriptsize{$p=100$}}
  \vspace{-0.25cm}
  \includegraphics[width=5cm,height=3.5cm]{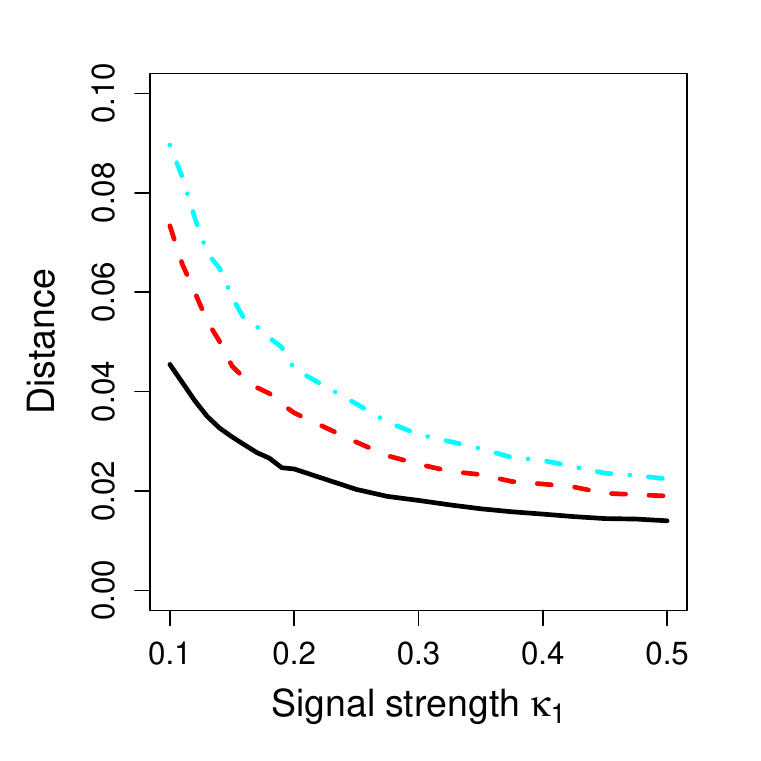}
  \label{fig:3-3-100}
  \vspace{-0.25cm}
  \caption{\scriptsize{$p=100$}}
  \vspace{-0.25cm}
  \includegraphics[width=5cm,height=3.5cm]{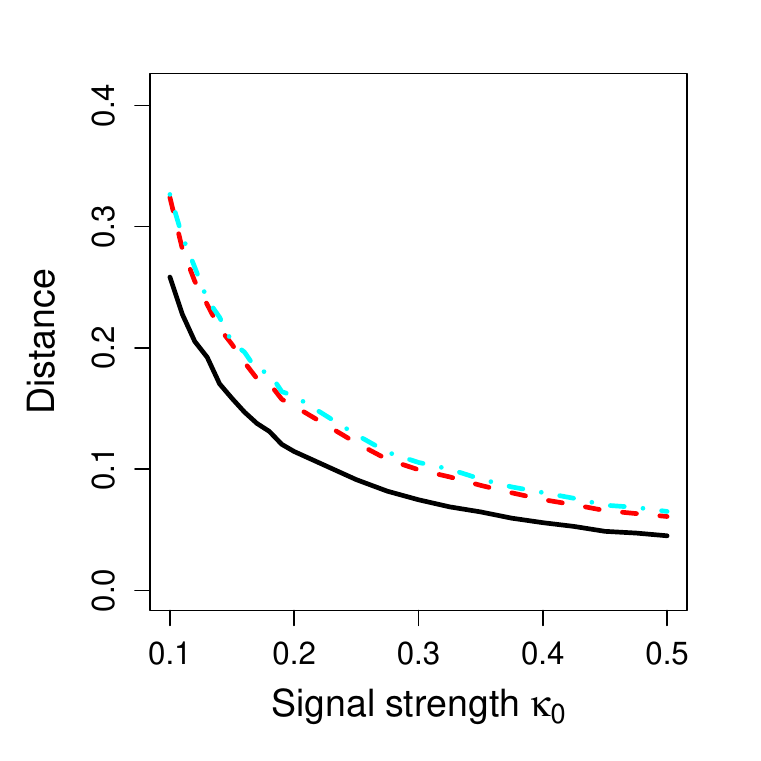}
  \label{fig:3-4-100}
\end{subfigure}
\begin{subfigure}{0.32\linewidth}
  \caption{\scriptsize{$p=200$}}
  \vspace{-0.25cm}
  \includegraphics[width=5cm,height=3.5cm]{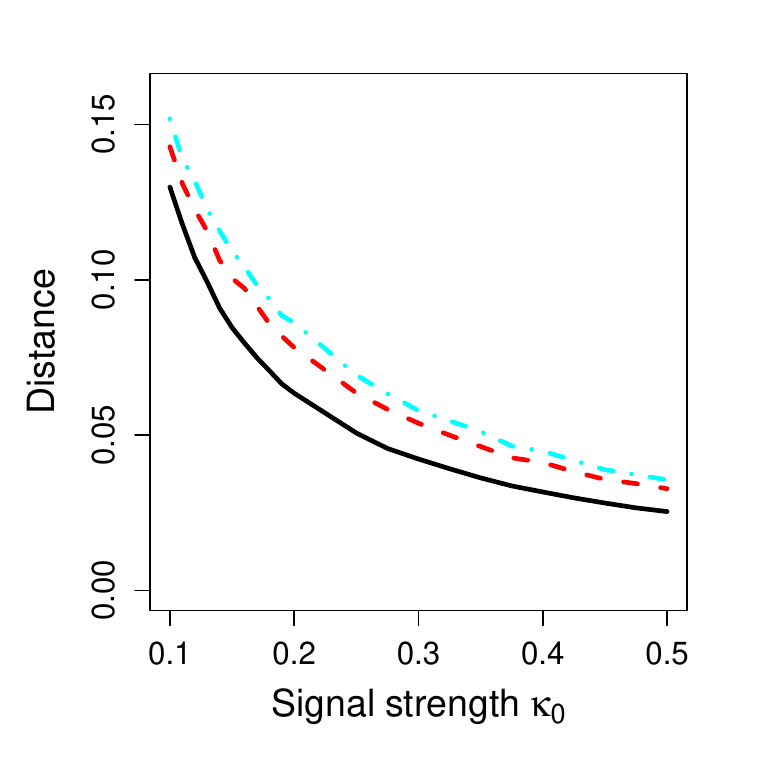}
  \label{fig:3-1-200}
  \vspace{-0.25cm}
  \caption{\scriptsize{$p=200$}}
  \vspace{-0.25cm}
  \includegraphics[width=5cm,height=3.5cm]{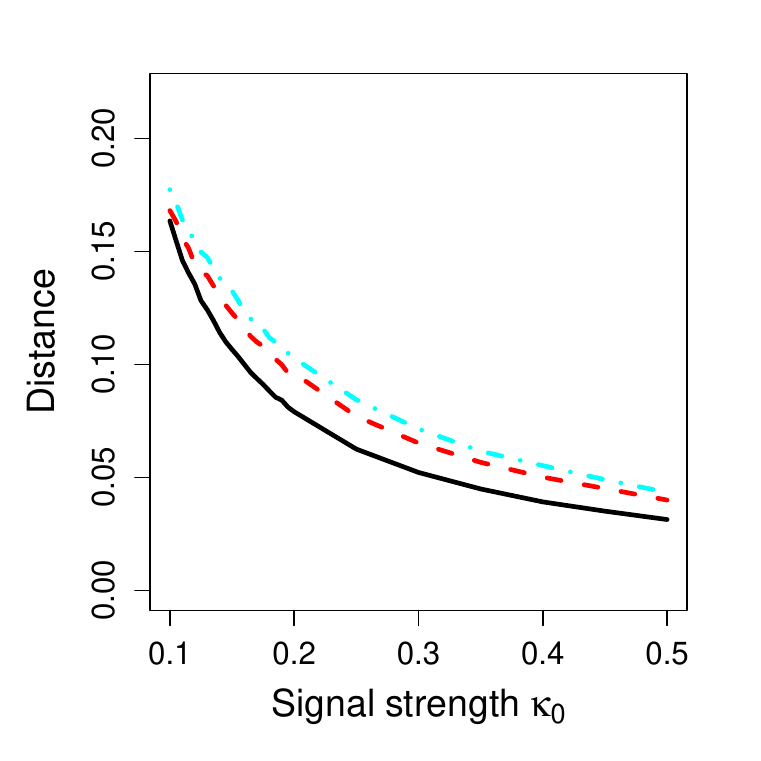}
  \label{fig:3-2-200}
  \vspace{-0.25cm}
  \caption{\scriptsize{$p=200$}}
  \vspace{-0.25cm}
  \includegraphics[width=5cm,height=3.5cm]{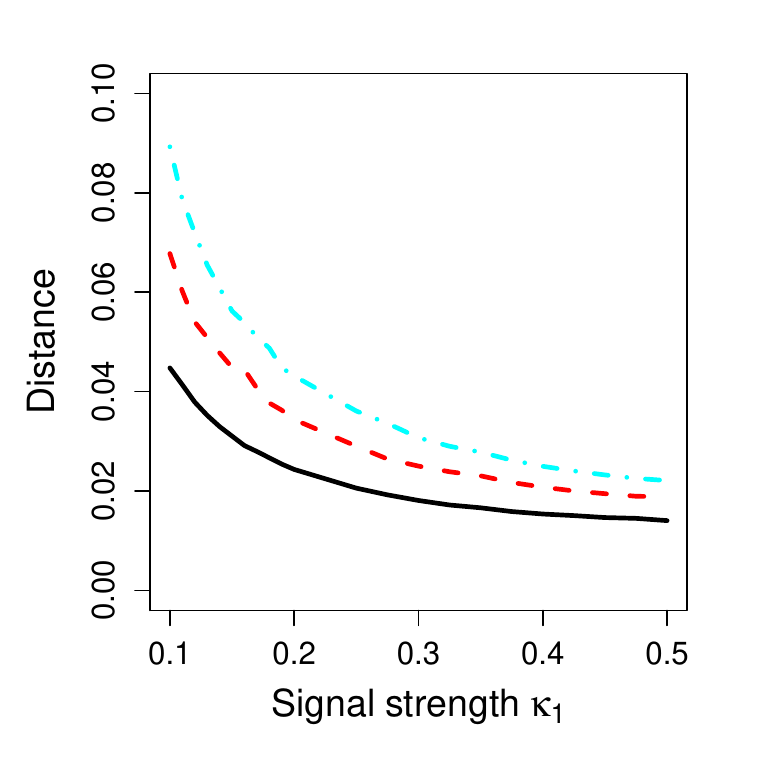}
  \label{fig:3-3-200}
  \vspace{-0.25cm}
  \caption{\scriptsize{$p=200$}}
  \vspace{-0.25cm}
  \includegraphics[width=5cm,height=3.5cm]{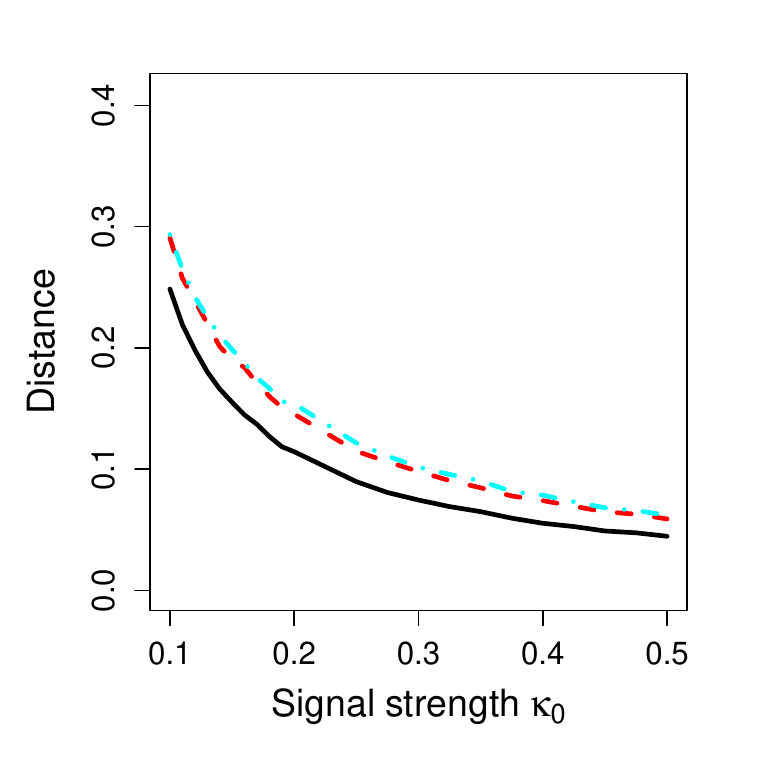}
  \label{fig:3-4-200}
\end{subfigure}
\centering
\begingroup
\linespread{1}
\caption{\label{fig2}{\it \footnotesize Scenario~1 (top row), Scenario~2 (2nd row), Scenario~3 (3rd row) and Scenario~4 (bottom row)} for $p=50,$ $100,$ $200$ when the factors are strong: Plots of average $\mathcal D \big(\mathcal C(\bA), \mathcal C(\widehat \bA)\big)$ against $\kappa_0$ or $\kappa_1$ for three methods based on $\widehat\bM$ (black solid), $\widehat\bM_1$ (red dashed), $\widehat\bM_2$ (cyan dash dotted).
}
\endgroup
\vspace{-0.2cm}
\end{figure}

Several conclusions can be drawn from Figures~\ref{fig1}--\ref{fig2} and \ref{fig3}--\ref{fig-cross}. 
First, our proposed method based on $\widehat\bM$ provides highly significant improvements in accuracies for the identification of $r$ and the recovery of $\cC(\bA)$ over the competing methods in all scenarios we consider.
The improvement of our method involving the weight matrix is larger for the heterogeneous case and even more substantial for the case when the factors are weak.
Between two competitors, the method based on $\widehat\bM_1$ outperforms that based on $\widehat\bM_2,$ providing empirical evidence for Remark~\ref{rmk.weight}.
Second, the estimation for $r$ and $\cC(\bA)$ performs better as the strength of factors increases (i.e., $\kappa_0$ or $\kappa_1$ increases or $\delta$ decreases), which is in line with our theoretical results in Section~\ref{sec:theory}. In Figure~\ref{fig1}, our method makes the sharpest progress as the factors become stronger, while the two comparison methods require much higher strength of factors to compete. 
Third, when the factors are strong, we observe the phenomenon of ``blessing of dimensionality'' in Figures~\ref{fig1}--\ref{fig2} in the sense that the estimation improves as $p$ increases from $50$ to $200.$ The improvement is due to the increase of the information from added components on the factors. Under Scenario~3 with an extra factor independent of the others, the information and noise on factors increase simultaneously as $p$ grows, and hence the estimation does not necessarily improve.
Fourth, under the weak factor setting, many entries of $\bA$ are quite close to zero. As $p$ enlarges, the enhanced information on the factors is accompanied with the increase of noise. While Figure~\ref{fig3} reveals that increasing $p$ does not necessarily lead to the improved estimation for $r$, we observe in Figure~\ref{fig4} that the estimation for $\cC(\bA)$ gets worse for larger values of $p$, complying with the result in Theorem~\ref{thm:eigvec}.

\subsection{Sparse case}\label{sim:sparse}
In this section, we conduct some simulations to evaluate the performance of the {functional-thresholding-and-sparse-PCA-based approach} (TSPCA) developed in Section~\ref{sec:ultrahigh} to estimate the sparse factor model. Two kinds of subspace sparsity constraints are imposed on $\bA.$ Specifically, for the row sparsity, we randomly select $80\%$ rows of $\bA$ to be zero vectors, while for the column sparsity, we randomly set $80\%$ elements within each column of $\bA$ as zeros. We focus on Scenario~1 with $\kappa_0=1,$ $n=100,$
and generate nonzero entries of $\bA$ from $\text{Uniform}[-\sqrt 3, \sqrt 3].$ 

Implementing TSPCA requires choosing the thresholding parameter $\eta_k$ in (\ref{th-est}) and the columnwise number of nonzero elements $C_{0}$ in (\ref{eq:sparse_1}). To select the optimal $\hat\eta_k$ for each $k=1, \dots, k_0,$ we implement a $G$-fold {cross-validation} approach \cite[]{cai2011adaptive}. To be specific, we first sequentially divide the set $\{1, \dots, n\}$ into $G$ blockwise groups $\cD_1, \dots, \cD_G$ of approximately equal size. We then treat
the $g$-th group as a validation set, compute the sample lag-$k$ autocovariance functions $\widehat\bSigma_{yy}^{(k),(g)}(u,v)$ and $\widehat\bSigma_{yy}^{(k),(-g)}(u,v)$ for $u,v \in \cU$ based on the validation set and the remaining $G-1$ groups, respectively, and repeat the above procedure $G$ times. We finally select $\widehat \eta_k$ by minimizing 
$R_k(\eta_k) = {G}^{-1}\sum_{g=1}^G \big\|{\mathcal T}_{\eta_k}\big(\widehat \bSigma^{(k),(-g)}_{yy}\big)- \widehat \bSigma^{(k),(g)}_{yy}\big\|_{\cS,\tF}^2.$
We adopt a similar {cross-validation} method to select the optimal $\widehat C_0.$
Given the $g$-th group as a validation set, we obtain the solution $\widetilde \bK^{(-g)}$ to the constrained optimization problem in (\ref{eq:sparse_1}) with $\tau=0$ based on the remaining $G-1$ groups. To solve this problem, we apply sparse PCA \cite[]{moghaddam2006generalized,mackey2009deflation} to $\widetilde \bM^{(-g)},$ which is formed by (\ref{eq:Mhat-thresh}) using ${\mathcal T}_{\widehat\eta_k}\big(\widehat \bSigma^{(k),(-g)}_{yy}\big)$ for $k=1, \dots, k_0.$
We also obtain $\widehat \bK^{(g)}$ by carrying out an eigenanalysis for $\widehat \bM^{(g)}$ based on the validation set. The above procedure is repeated $G$ times and $\widehat C_0$ is selected by minimizing 
$D(C_0) = {G}^{-1}\sum_{g=1}^G \cD\big(\cC(\widetilde \bK^{(-g)}), \cC(\widehat \bK^{(g)})\big).$ Although the time break created by each leave-out validation set possibly jeopardize the autocovariance structure on the remaining $G-1$ groups via $k_0$ mis-utilized lagged terms, its effect on the estimation of $\bSigma_{yy}^{(k)(-g)}$'s is negligible especially when $n$ is sufficiently large. Hence our proposed {cross-validation} approach does not place a practical constraint.
\begin{table}[ht]
\centering
\captionsetup{width=\textwidth}
\begingroup
\linespread{1}
\caption{\footnotesize The mean and standard error (in parentheses) of the distance between the estimated and true sparse factor loading spaces 
over $100$ simulation runs. All entries have been multiplied by $10^2$ for formatting reasons.}\label{tab.tspca}
\endgroup
\vspace{0.15cm}
\resizebox{5.7in}{!}{
    \scriptsize
    \begin{tabular}{cccccccc}
    \hline
    Sparsity &$p$ & PCA & TSPCA & Sparsity &$p$ & PCA & TSPCA \\
    \hline
    &$100$ & $3.86 (0.53)$ & $1.71 (0.31)$&           &$100$ & $4.27 (0.83)$ & $3.55 (1.09)$\\

    Row &$200$ & $3.79 (0.56)$ & $1.74 (0.26)$&Column &$200$ & $4.25 (0.65)$ & $3.69 (0.59)$\\
    &$400$ & $3.73 (0.66)$ & $1.76 (0.29)$&           &$400$ & $4.03 (0.74)$ & $3.17(0.88)$\\
    \hline
    \end{tabular}
}
\end{table}

We compare TSPCA with the ordinary PCA-based approach by performing an eigenanlysis on $\widehat \bM$. The estimation quality is measured in terms of the distance between the estimated and true sparse factor loading spaces using the correct $r.$ We report the numerical summaries in Table~\ref{tab.tspca} for $p=100, 200, 400.$ It is obvious to see that TSPCA 
uniformly improves the PCA-based estimation, demonstrating the advantage of our regularized estimation procedure to fit sparse factor models. The improvement for the row sparsity case is more significant than that for the column sparsity case. Compared with the functional sparsity patterns in Lemma~\ref{lem:sparse} under the column sparsity constraint, it follows from the decomposition of $\bSigma_{yy}^{(k)}$ in (\ref{eq:cov-relation}) that the row sparsity in $\bA$ can lead to functional sparser $\bSigma_{yy}^{(k)}$'s. In this sense, TSPCA benefits more from the row sparsity structure, thus resulting in enhanced improvement.

\section{Real Data Analysis}
\label{sec:data}
\subsection{U.K. Temperature Data}
\label{sec:temp.data}
Our first data set, which is available at \url{https://www.metoffice.gov.uk/research/climate/maps-and-data/historic-station-data}, consists of monthly average temperature collected at $p=22$ measuring stations in the U.K. from 1959 to 2020 ($n=62$). Let $Y_{tj}(u_k)$ ($t=1, \dots, n, j=1, \dots,p, k=1, \dots, 12$) denotes the average temperature during month $u_k=k \in \cU=[1,12]$ of year $1958+t$ at the $j$-th measuring station. The observed temperature curves are smoothed based on $10$-dimensional Fourier basis to capture the periodic structures over the annual cycle. The smoothed curve series exhibits very small autocorrelations beyond $k=2,$ so we use $k_0=2$ in computing $\widehat \bM.$ 
The ratio-based estimator suggests $\hat r=1,$ though we illustrate further results using $\hat r=3.$ We also conduct multiple testing of univariate functional white noise using \cite{zhang2016white}'s method on component residual curve series after taking out the dynamical component driven by $\hat r$ factors. Within multiple testing framework, we apply the Benjamini--Hochberg correction in adjusting p-values. With the target false discovery rate level 10\%, we obtain 0 (or 8) out of 22 nulls are rejected when $\hat r=3$ (or $\hat r=1$), suggesting no overwhelming evidence against the component white noise assumption for the idiosyncratics with $\hat r=3.$

Figure~\ref{fig:temp} displays spatial heatmaps of estimated factor loading matrix and the rotated matrix using the varimax procedure to maximize the sum of the variances of the squared loadings. Some interesting patterns can be observed from those heatmaps. Compared with the original results in Figure~\ref{fig:temp}, the varimax rotation brings the factor loading matrix closer to a ``simple structure", where (i) each component in $\bY_t(\cdot)$ has a high loading on one specific factor but near-zero loadings on other factors and (ii) each factor has high impact on a few components of $\bY_t(\cdot)$ with high loadings on this factor while the remaining variables have near-zero loadings on this factor. The varimax rotation leads to larger color variations in heatmaps for better interpretation.
For example, it is apparent that the first factor via varimax rotation influences the dynamics in the southeast, while the original factor has uniform impact in all locations. Hence we focus on the interpretation of remaining factors after varimax rotation. Specifically, the second factor mainly impacts the dynamics of the northern region. The third factor can roughly be viewed as the main driving force for the dynamics in the middle north.
\begin{figure}[!hbt]
	\centering
	\begin{subfigure}{0.32\linewidth}
	  \caption{\scriptsize{1st factor (original)}}
	  \vspace{-0.9cm}
	  \includegraphics[width = 4.5cm,height = 6cm]{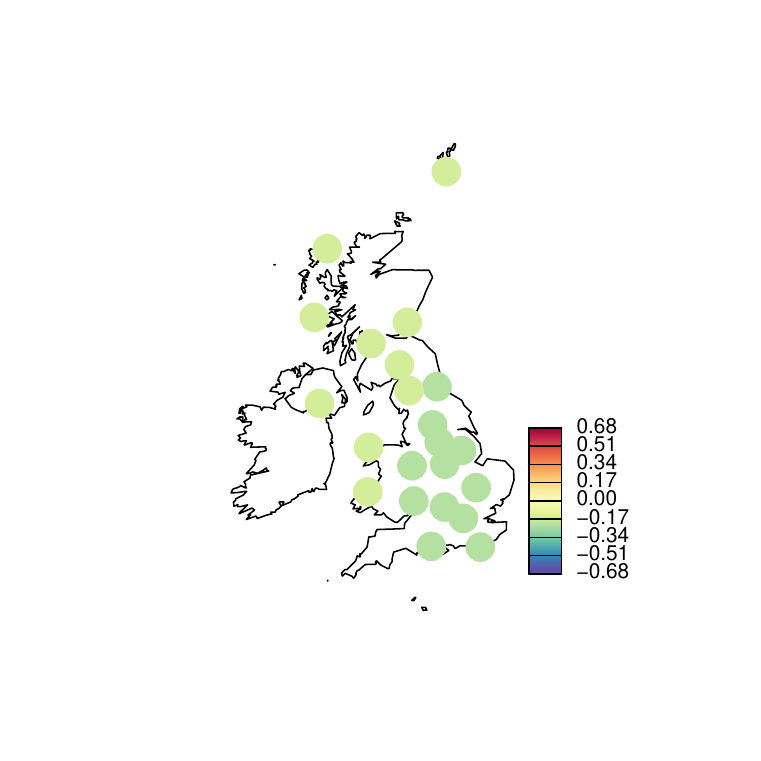}
	  \vspace{-1.5cm}

	  \caption{\scriptsize{1st factor (rotated)}}
	  \vspace{-0.9cm}
	  \includegraphics[width = 4.5cm,height = 6cm]{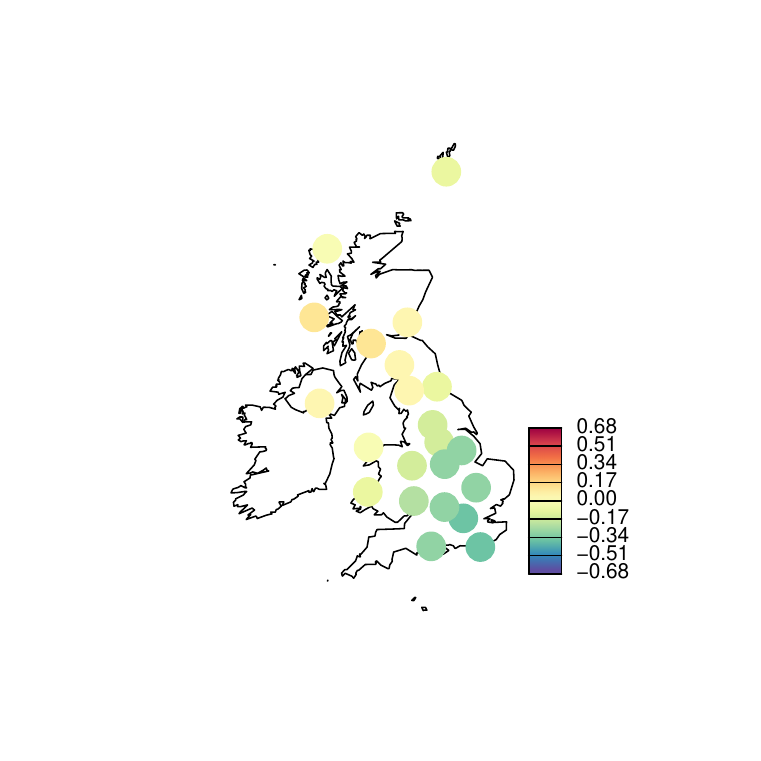}
	\end{subfigure}\hspace{-0.15cm}	
	\begin{subfigure}{0.32\linewidth}
		\caption{\scriptsize{2nd factor (original)}}
		\vspace{-0.9cm}
		\includegraphics[width = 4.5cm,height = 6cm]{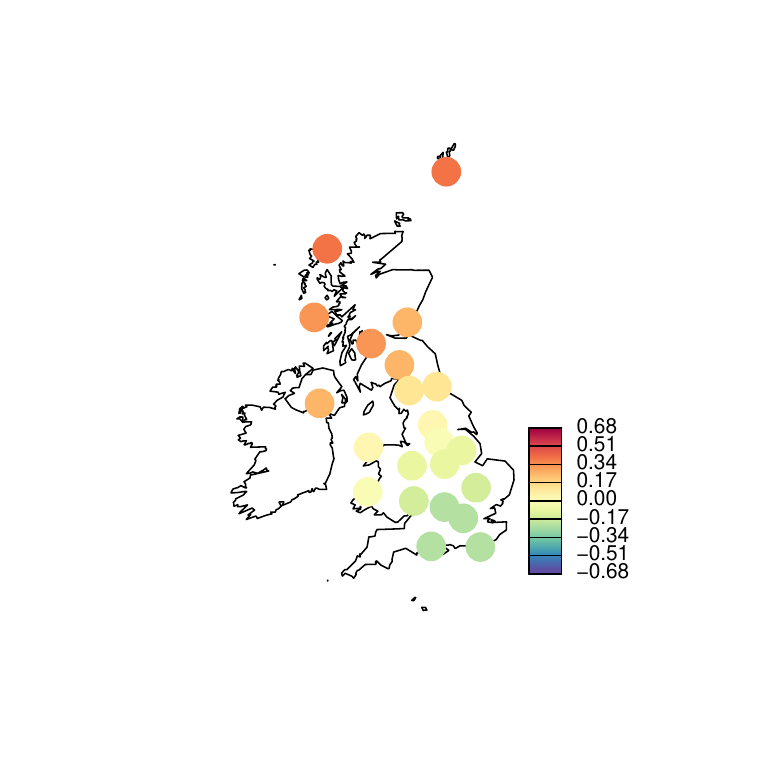}
	  \vspace{-1.5cm}
  
		\caption{\scriptsize{2nd factor (rotated)}}
		\vspace{-0.9cm}
		\includegraphics[width = 4.5cm,height = 6cm]{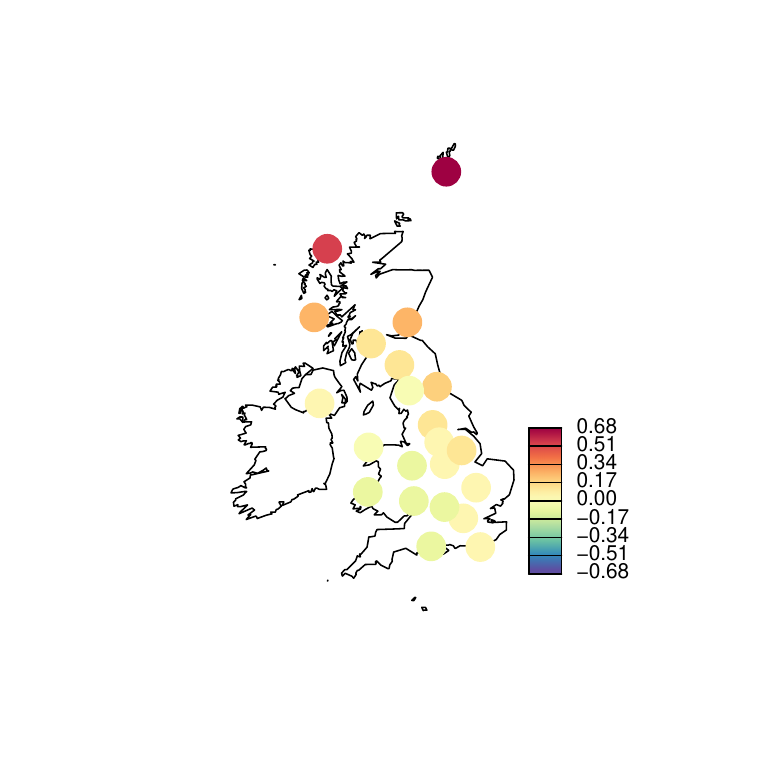}
	\end{subfigure}\hspace{-0.15cm}	
	\begin{subfigure}{0.32\linewidth}
		\caption{\scriptsize{3rd factor (original)}}
		\vspace{-0.9cm}
		\includegraphics[width = 4.5cm,height = 6cm]{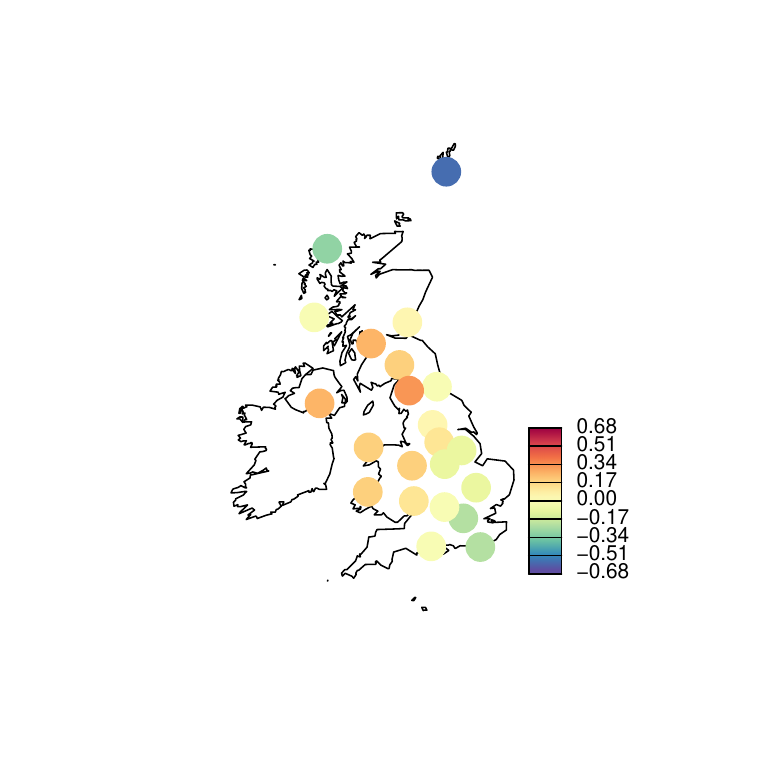}
	  \vspace{-1.5cm}
  
		\caption{\scriptsize{3rd factor (rotated)}}
		\vspace{-0.9cm}
		\includegraphics[width = 4.5cm,height = 6cm]{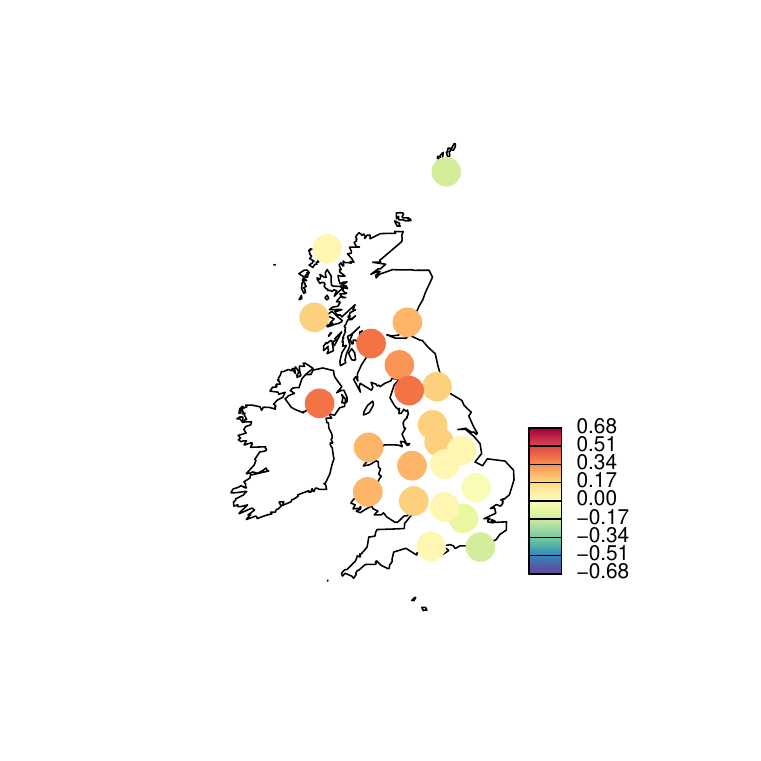}
	\end{subfigure}	
        \vspace{-1.5cm}
	\centering
        \begingroup
        \linespread{1}
	\caption{\label{fig:temp}{\it \footnotesize
	Spatial heatmaps based on estimated factor loading matrix (top row) and varimax-rotated loading matrix (bottom row) of 22 U.K. locations.}}
        \endgroup
\end{figure}

\subsection{Japanese Mortality Data}
\label{sec:mort.data}

Our second data set contains age-specific and gender-specific mortality rates for $p=47$ Japanese prefectures from 1975 to 2017 ($n=43$). This data set was also analyzed in \cite{gao2019high}. Due to those sparse observations at old ages, we focus on data for ages below 96. We apply a log transformation to mortality rates and denote by $Y_{tj}(u_k)$ ($t=1, \dots, n, j=1, \dots, p, k=1, \dots, 96$) the log mortality rate of people aged $u_k = k-1 \in \cU=[0,95]$ living in the $j$-th prefecture during the year $1974+t.$ We then perform smoothing for observed curves and replace the missing values via smoothing splines. The estimation of model (\ref{eq:model}) is done by choosing $k_0=2$ 
and we use $\hat r=2$ in subsequent analysis. To enhance interpretability for identified factors under a high-dimensional $p>n$ regime, we also implement TSPCA to estimate the sparse factor loading matrix.

Figure~\ref{fig:female} and Figure~\ref{fig:male} in 
the Supplementary Material show spatial heatmaps of varimax-rotated loading matrix and sparse loading matrices with different sparsity levels for Japanese females and males, respectively. 
For TSPCA, we implement the cross-validation method to select the optimal $\hat\eta_k$'s and set the columnwise sparsity in estimated loading matrix to 27/47 or 32/47 for better visualization of the results. A more systematic method to determine the sparsity level, e.g., via a significance testing, needs to be developed.
Compared with the original results, both varimax rotation and TSPCA lead to enhanced interpretability in the sense of recovering the factor loading matrix with ``simple structure". Among two competitors, TSPCA tends to reduce some near-zero loadings to exactly zero, thus providing a parsimonious model with more interpretable results than varimax rotation. In particular, we can observe clear regional patterns from those heatmaps for TSPCA, where each of two factors influences the dynamics of some complementary regions of Japan in all settings. Take the results under the sparsity 27/47 in Figure~\ref{fig:female} as an example, the first factor serves as the main driving force in regions of Hokkaido, Kanto, Chubu and Kansal, while Tohoku and Shikoku are heavily loaded regions on the second factor.
\begin{figure}[!htb]

	\centering
	\begin{subfigure}{0.32\linewidth}
	  \caption{\scriptsize{1st factor (rotated)}}
	  \vspace{-0.5cm}
	  \includegraphics[width=6cm,height=5cm]{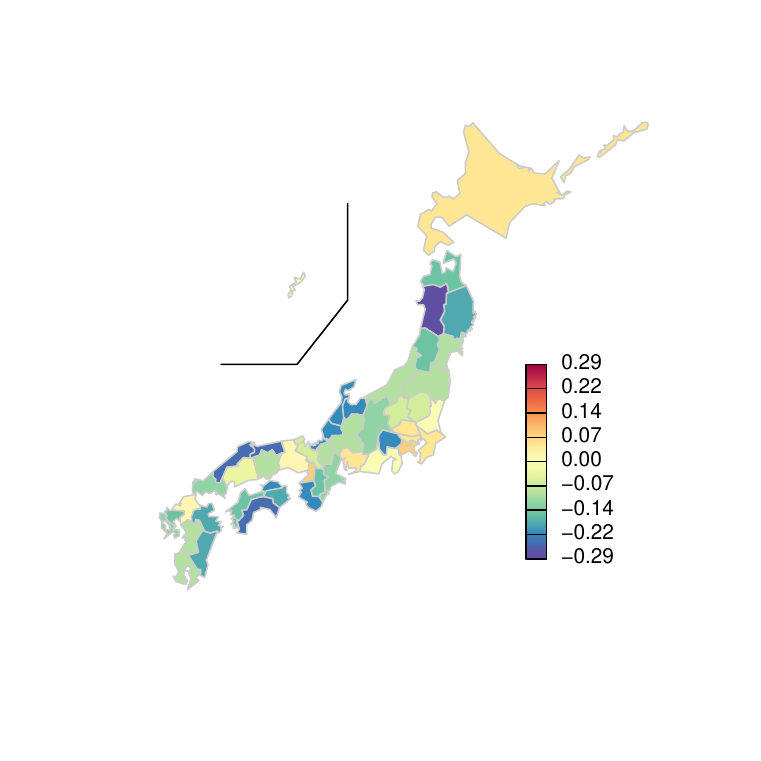}
	  \vspace{-2cm}

	  \caption{\scriptsize{2nd factor (rotated)}}
	  \vspace{-0.5cm}
	  \includegraphics[width=6cm,height=5cm]{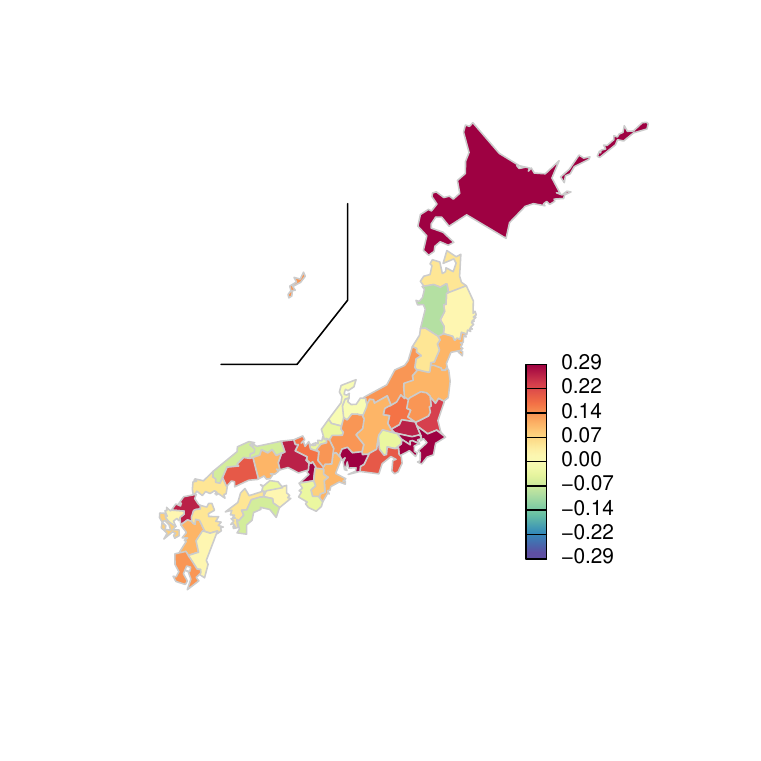}
	\end{subfigure}\hspace{-0.15cm}	
	\begin{subfigure}{0.32\linewidth}
	  \caption{\scriptsize{1st factor (27 zeros)}}
	  \vspace{-0.5cm}
	  \includegraphics[width=6cm,height=5cm]{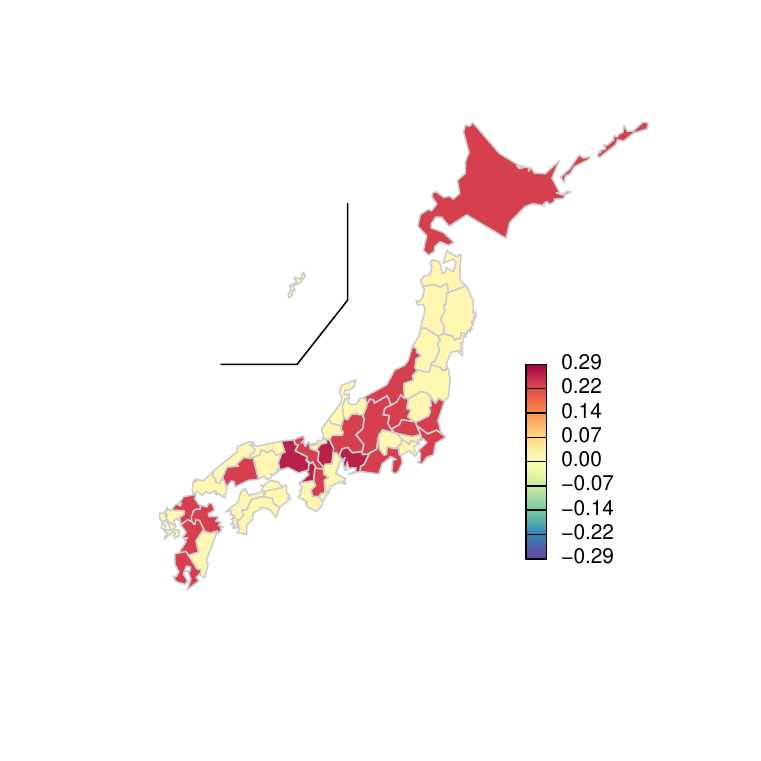}
	  \vspace{-2cm}

	  \caption{\scriptsize{2nd factor (27 zeros)}}
	  \vspace{-0.5cm}
	  \includegraphics[width=6cm,height=5cm]{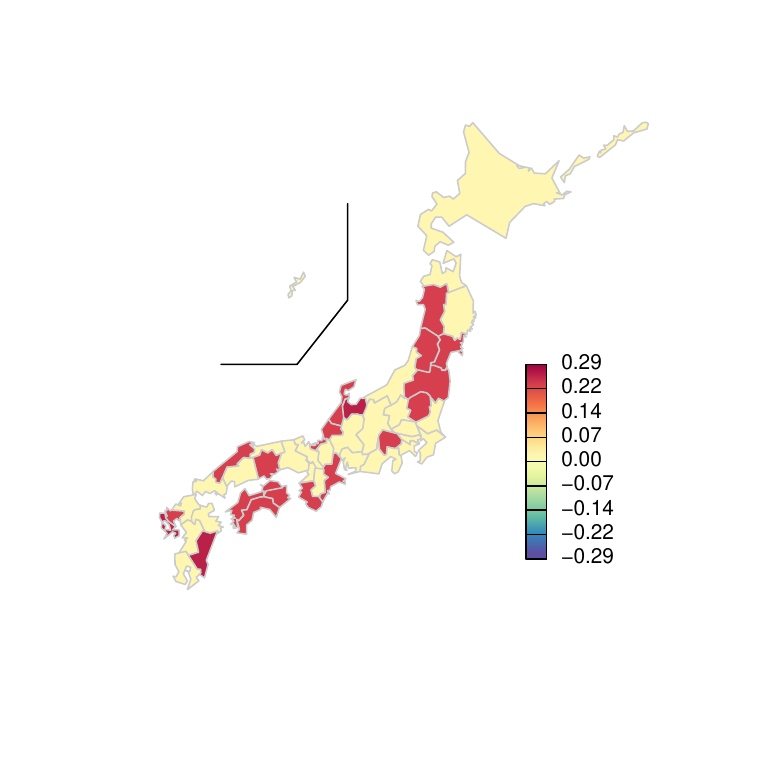}
	\end{subfigure}\hspace{-0.15cm}
	\begin{subfigure}{0.32\linewidth}
		\caption{\scriptsize{1st factor (32 zeros)}}
	  \vspace{-0.5cm}
	  \includegraphics[width=6cm,height=5cm]{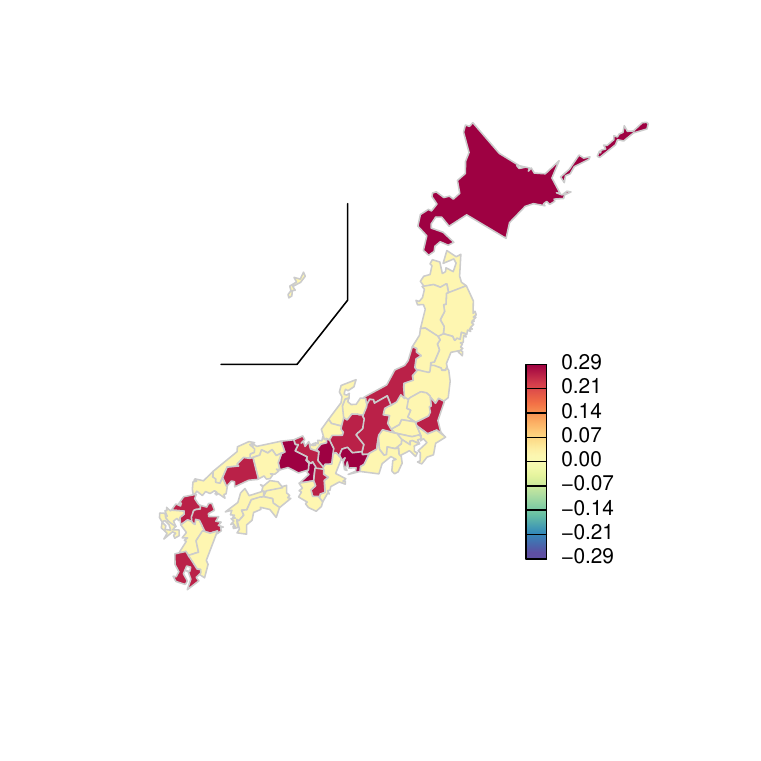}
	  \vspace{-2cm}

	  \caption{\scriptsize{2nd factor (32 zeros)}}
	  \vspace{-0.5cm}
	  \includegraphics[width=6cm,height=5cm]{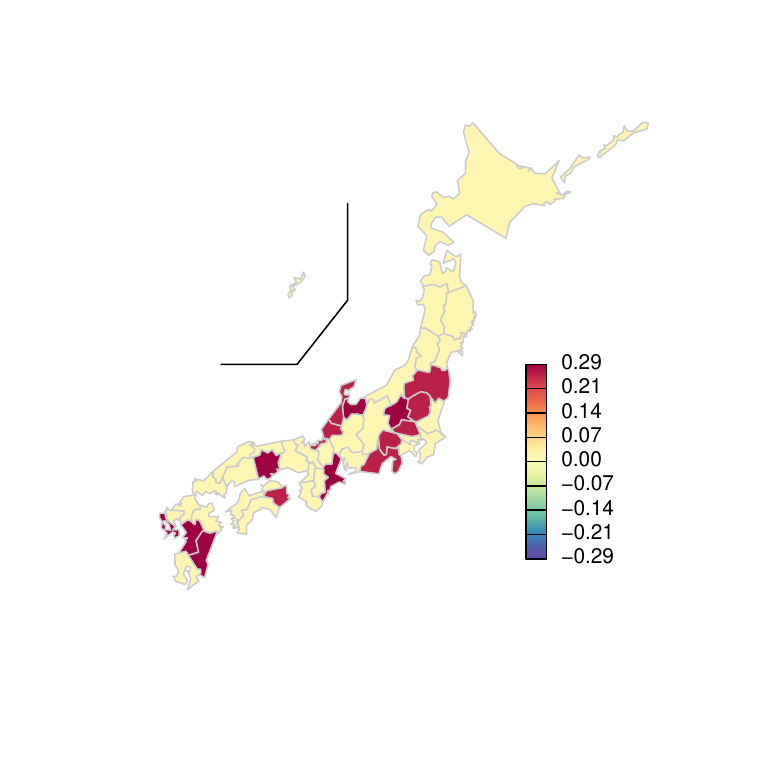}
	\end{subfigure}
 	  \vspace{-1.25cm}
	\centering
        \begingroup
        \linespread{1}
	\caption{\label{fig:female}{\it  \footnotesize Spatial heatmaps based on varimax-rotated loadings (left column) and sparse loadings with 27 zeros (middle column) and 32 zeros (right column) of 47 prefectures on two factors for Japanese females.
	}}
        \endgroup
\end{figure}

To further illustrate the developed methodology, we set upon the task of predicting high-dimensional functional time series $\bY_t(\cdot).$ Specifically, we incorporate the functional factor model framework into the $h$-step-ahead prediction (denoted as FFM) consisting of three steps below.
\begin{enumerate}
    \item[1.] Apply our proposal to estimate model (\ref{eq:model}) based on past observations $\{\bY_t(\cdot)\}_{t=1}^T,$ thus obtaining the estimated factor loading matrix $\widehat \bA$ and number of functional factors $\hat r.$
    \item[2.]  Compute $\hat r$ estimated factors by $\widehat \bX_t(\cdot) = \widehat \bA^\T \bY_t(\cdot).$ 
    For each $k \in \{1, \dots, \hat r\},$ predict $\widehat X_{(t+h)k}(\cdot)$ based on past fitted values for the $k$-th common factor, $\{\widehat X_{tk}(\cdot)\}_{t=1}^T.$ 
    \item[3.]  The $h$-step head predict for $\bY_{T+h}(\cdot)$ is $\widehat \bY_{T+h}(\cdot) = \widehat \bA\widehat \bX_{T+h}(\cdot).$
\end{enumerate}
The second step is on the prediction of univariate functional time series. For each common factor, we obtain the $h$-step-ahead prediction based on the best fitted functional ARMA model \cite[]{klepsch2017} according to the BIC criterion. We then develop a sparse version of FFM (denoted as SFFM) by performing TSPCA in the first step to estimate sparse $\widehat\bA.$ For comparison, we implement an alternative factor model based prediction method of \cite{tavakoli2023b} (denoted as TNH) by firstly estimating their factor model (see (\ref{ffm1}) in Section~\ref{sec:dis}), then making use of the corresponding best fitted ARMA models to predict the scalar factor time series $\widetilde \bX_{t}$ and the functional idiosyncratic component $\widetilde \bvarepsilon_t(\cdot),$ and finally predicting $\bY_t(\cdot)$ based on (\ref{ffm1}). We also adopt the method of \cite{gao2019high} (denoted as GSY) to predict $p$ component series of $\bY_t(\cdot)$ jointly, as well as the univariate prediction method of \cite{aue2015prediction} (denoted as ANH) to predict each $Y_{tj}(\cdot)$ separately. 

To evaluate the predictive accuracy, we use the expanding window approach. The data is divided into a training set and a test set consisting of the first $n_1$ and the last $n_2$ observations, respectively.
For any integer $h>0,$ we implement each fitting method on the training set, obtain $h$-step-ahead prediction on the test data based on the fitted model, increase the training size by one and repeat this procedure $n_2-h+1$ times to compute the $h$-step-ahead {mean absolute prediction error} (MAPE) and {mean squared prediction error} (MSPE) by
$         \text{MAPE}(h) = \{p\times(n_2-h+1)\times 96 \}^{-1}\sum_{j=1}^p\sum_{t=n_1+h}^{n}\sum_{k=1}^{96} \big|\widehat Y_{tj}(u_k) - Y_{tj}(u_k)\big|$ and
$\text{MSPE}(h) = \{p\times(n_2-h+1)\times 96 \}^{-1}\sum_{j=1}^p\sum_{t=n_1+h}^{n}\sum_{k=1}^{96} \big\{\widehat Y_{tj}(u_k) - Y_{tj}(u_k)\big\}^2,$ respectively. 
The resulting MAPE and MSPE values are summarized in Table~\ref{tab.forecast}. It is obvious that FFM, SFFM and TNH significantly outperform GSY and ANH in all settings. Among the three winners, our proposed FFM provides the highest predictive accuracies for all cases, and, at the same time, SFFM is slightly inferior possibly due to the bias introduced by the enforced sparsity. It is worth mentioning that our additional experiments show that the predictive performance of FFM and SFFM can be further improved as $\hat r$ increases beyond $2,$ whereas the best predictive performance of TNH is already attained with only one scalar factor as presented in Table~\ref{tab.forecast}.


\begin{table}[!htbp]
	\centering
        \begingroup
        \linespread{1}
	\caption{\footnotesize MAPEs and MSPEs of FFM, SFFM and three competing methods for Japanese male and female mortality rates. All entries have been multiplied by $10$ for formatting reasons.}\label{tab.forecast}
        \endgroup
	\vspace{0.15cm}
	\resizebox{5.7in}{!}{
	\footnotesize
	\begin{tabular}{cc|ccccc|ccccc}
		\hline
		& &\multicolumn{5}{c|}{MAPE}&\multicolumn{5}{c}{MSPE}\\
		\hline
	    %

        \multirow{4}{*}{Male} & $h$ & FFM & SFFM & TNH & GSY & ANH & FFM & SFFM & TNH & GSY & ANH\\
		\multirow{4}{*}{} & 1 & 1.33 & 1.36 & 1.33 & 1.75 & 1.42 & 0.48 & 0.49 & 0.49 & 0.84 & 0.59\\
        \multirow{4}{*}{} & 2 & 1.37 & 1.43 & 1.37 & 1.88 & 1.47 & 0.51 & 0.52 & 0.52 & 0.92 & 0.63\\
        \multirow{4}{*}{} & 3 & 1.43 & 1.49 & 1.44 & 1.87 & 1.53 & 0.54 & 0.55 & 0.55 & 0.86 & 0.67\\\hline
	    
	    \multirow{4}{*}{Female} & $h$ & FFM & SFFM & TNH & GSY & ANH & FFM & SFFM & TNH & GSY &ANH\\
		\multirow{4}{*}{} & 1 & 1.40 & 1.44 & 1.41 & 1.77 & 1.52 & 0.61 & 0.64 & 0.68 & 0.87 & 0.82\\
        \multirow{4}{*}{} & 2 & 1.45 & 1.51 & 1.45 & 1.84 & 1.56 & 0.66 & 0.69 & 0.68 & 0.94 & 0.99\\
        \multirow{4}{*}{} & 3 & 1.50 & 1.56 & 1.50 & 1.77 & 1.66 & 0.71 & 0.74 & 0.74 & 0.89 & 1.14\\\hline
		\end{tabular}
	}
\end{table}

\section{Discussion}
\label{sec:dis}
It is noteworthy that \cite{hallin2023a} develops an alternative functional factor model with scalar factors and functional loadings. Their methodology allows different components of multivariate functional time series to reside in different Hilbert spaces or a mix of functional and scalar time series, which is not possible with our approach. Nevertheless, for all applications mentioned in Section~\ref{sec:intro} and explored in Section~\ref{sec:data}, the corresponding components of multivariate functional time series all lie in the same Hilbert space. In such a common practical setting, their model reduces to
\begin{equation}
    \label{ffm1}
    \bY_t(\cdot) = \widetilde \bA(\cdot) \widetilde \bX_t + \widetilde \bvarepsilon_t(\cdot),~~t=1, \dots, n,
\end{equation}
where $\widetilde\bX_t = (\widetilde X_{t1}, \dots, \widetilde X_{t\tilde r})^{\T}$ is a set of latent factor time series and $\widetilde\bA(\cdot)=\{\widetilde A_{jl}(\cdot)\}_{p \times \tilde r}$ is the unknown functional factor loading matrix. 
Both (\ref{eq:model}) and (\ref{ffm1}) generate useful factor models for high-dimensional functional time series $\bY_t(\cdot),$ but are designed to tackle rather different situations. 
Note one crucial question in functional time series modelling is how to characterize the functional and time series structures. Our factor model (\ref{eq:model}) with static factor loadings assumes that both structures are inherited from $r$ common functional time series factors $\bX_t(\cdot)$ with reduced dimension from $p \times \infty$ to $r \times \infty$ before subsequent analysis. By comparison, factor model (\ref{ffm1}) treats $\tilde r$ common factors in $\widetilde \bX_t$ as $\tilde r$-dimensional vector time series, which reduces the dimension from $p \times \infty$ to a much lower value $\tilde r,$ while the infinite-dimensional structure is maintained  in the functional factor loading matrix $\widetilde \bA(\cdot).$

We next discuss a possible way to apply our method to estimate model (\ref{ffm1}). 
Given an orthonormal basis $\{\phi_i(\cdot)\}_{i=1}^{\infty},$
we expand the functional objects in (\ref{ffm1}) by $Y_{tj}(\cdot)=\sum_{i=1}^{\infty} \xi_{tji} \phi_{i}(\cdot),$
$\widetilde A_{jl}(\cdot)=\sum_{i=1}^{\infty} a_{jli} \phi_{i}(\cdot),$
$\widetilde \varepsilon_{tj}(\cdot)=\sum_{i=1}^{\infty} \epsilon_{tji}\phi_{i}(\cdot)$
and hence rewrite (\ref{ffm1}) as
\begin{equation}
    \label{ffm2}
    \xi_{tji} = \sum_{l=1}^{\tilde r} a_{jli} \widetilde X_{tl} + \epsilon_{tji},~~i=1, \dots, \infty.
\end{equation}
To simplify notation, we assume the same truncated dimension $M$ across $j$ and stack the basis coefficients $\{\xi_{tjl}\}_{1 \leq j \leq p, 1 \leq l \leq L},$  $\{a_{jkl}\}_{1 \leq j \leq p, 1 \leq l \leq L}$ and $\{\epsilon_{tjl}\}_{1 \leq j \leq p, 1 \leq l \leq L}$ 
$\xi_{tji},$ $a_{jli}$ and $\epsilon_{tji}$ for $j=1, \dots, p$ and $i=1, \dots, M$ to $pM$-dimensional vectors $\bxi_t, \ba_{l}$ and $\bepsilon_t,$ respectively. As a result, (\ref{ffm2}) can be equivalently represented as the following factor model for $pM$-dimensional vector time series $\bxi_1, \dots, \bxi_n$:
\begin{equation}
    \label{ffm3}
    \bxi_t = \bLambda \widetilde \bX_{t} + \widetilde \bepsilon_{t},
\end{equation}
where the factor loading matrix $\bLambda=(\ba_1, \dots, \ba_{\tilde r}) \in {\mathbb R}^{pM \times \tilde r}$
and the noise vector $\widetilde \bepsilon_{t} \in {\mathbb R}^{pM}$ is decomposed as the sum of $\bepsilon_{t}$ and $\bepsilon'_t$ formed by truncation errors. Then we can integrate our suggested weight matrix  into the estimation of $\tilde r$ and ${\cal C}(\bLambda)$  for the factor model (\ref{ffm3}). Hence the space spanned by columns of $\widetilde \bA(\cdot)$ can be recovered accordingly. 


We finally identify three important directions for future research.
The first direction considers the extension of (\ref{eq:model}) to other forms of functional factor models. One possible extension admits the representation
$\bY_t(u)=\bA(u)\bX_t(u)+ \bvarepsilon_t(u),u \in \cU,$
which treats both the latent factor time series $\{\bX_t(\cdot)\}_{t=1}^n$ and the factor loading matrix $\bA(\cdot)=\{A_{jl}(\cdot)\}_{p \times r}$ as functional objects in the sense of the concurrent functional model \cite[]{Bramsay1}. On the other hand, inspired from the functional linear regression with functional response, we can consider another more generalized functional factor model in the form of
$\bY_t(u) = \int_{\cU} \bA(u,v)\bX_t(v) {\rm d}v+ \bvarepsilon_t(u),u \in \cU,$ where $\bA(\cdot,\cdot)=\{A_{jl}(\cdot, \cdot)\}_{p \times r}$ is the operator-valued factor loading matrix with $p \times r$ entries of bivariate functions. It is interesting to develop estimation procedures to fit the above two models. However, compared to the fittings for models~\eqref{eq:model} and (\ref{ffm1}), this would pose more complicated challenges that require further investigation.

Second, our estimation procedure is developed by assuming that the idiosyncratic component is a white noise sequence. It is also interesting to estimate the factor model~\eqref{eq:model} following the framework in \cite{bai2002determining} and \cite{fan2013}, where the idiosyncratic noise is allowed to exhibit serial correlations.
Note the covariance function $\bSigma_{yy}^{(0)}(u,v)$ or its integral over $\cU\times\cU$ can not be used directly in the decomposition since they are not nonnegative definite and symmetric.
Under the column-orthogonality of $\bA$ and the uncorrelatedness between $\bX_t(\cdot)$ and $\bvarepsilon_t(\cdot),$ we can slightly modify our scheme and consider decomposing the nonnegative definite and symmetric matrix $\int_{\cU}\int_{\cU} \bSigma^{(0)}_{yy}(u,v) \bSigma^{(0)}_{yy}(u,v)^\T\,\mathrm d u \mathrm d v$ as the sum of the leading term $\bA\big\{\int_{\cU}\int_{\cU} \bSigma^{(0)}_{xx}(u,v) \bSigma^{(0)}_{xx}(u,v)^\T\,\mathrm d u \mathrm d v\big\}\bA^{\T}$ and the remaining three terms of smaller orders. By imposing suitable eigenvalue conditions in a similar spirit to those in \cite{fan2013}, the common factors are asymptotically identifiable and hence the factor loading space can be recovered by carrying out an eigenanalysis for $\int_{\cU}\int_{\cU} \widehat\bSigma^{(0)}_{yy}(u,v) \widehat\bSigma^{(0)}_{yy}(u,v)^\T\,\mathrm d u \mathrm d v.$

Third, our estimation procedure is naturally adaptable to fit the factor model for high-dimensional scalar time series \cite[]{lam2012}. We also believe that such procedure can be extended to deal with the factor modelling for high-dimensional matrix-valued time series \cite[]{wang2019} or even tensor-valued time series \cite[]{chen2021}. Despite the integration step being no longer needed, we can still incorporate the suggested weight matrix to account for the heterogeneous effect and improve the estimation efficiency.
These topics are beyond the scope of the current paper and will be pursued elsewhere.

\linespread{1}\selectfont
\bibliographystyle{chicago}
\bibliography{ref}

\newpage
\spacingset{1.7}
\begin{center}
	{\noindent \bf \Large Supplementary material  to ``Factor Modelling for High-Dimensional Functional Time Series"}\\
\end{center}
\begin{center}
	{\noindent Shaojun Guo, Xinghao Qiao, Qingsong Wang and Zihan Wang}
\end{center}
\bigskip

\setcounter{page}{1}
\setcounter{section}{0}
\renewcommand\thesection{S.\arabic{section}}
\setcounter{equation}{0}
\renewcommand{\theequation}{S.\arabic{equation}}

This supplementary material contains additional regularity conditions in Section~\ref{ap_cond}, proofs of main theoretical results in Section~\ref{supp.pf}, some auxiliary lemmas and their proofs in Section~\ref{supp.lem}, and additional empirical results supporting Sections~\ref{sim:ordi} and \ref{sec:mort.data} in Section~\ref{supp.emp}.

\section{Additional regularity conditions}
\label{ap_cond}
This section contains some additional regularity conditions mentioned in Section~\ref{sec:ultrahigh}. To study theoretical properties in an ultra-high-dimensional regime, we need to rely on non-asymptotic error bounds on $\widehat\Sigma_{yy,ij}^{(k)}$'s and will use the functional stability measure of $\{\bY_t(\cdot)\}_{t \in \eZ}$ \textcolor{blue}{(Guo and Qiao, 2023)}. 

Before presenting regularity conditions, we first solidify some notation. We denote the $p$-fold Cartesian product of $\mathbb H=L_2(\cU)$ by
$\eH^p = \eH \times \cdots \times \eH.$ For $\bw = (w_1, \dots, w_g)^{\T}, \bg = (g_1, \dots, g_p)^{\T} \in \eH^p,$ we define $\langle \bw, \bg \rangle = \sum_{j=1}^p w_j g_j.$ For an integral operator $\bK: \eH^p \rightarrow \eH^p$ induced from the kernel function $\bK=(K_{ij})_{p \times p},$ 
$$\bK(\bw)(u)= \Big(\sum_{j=1}^p \langle K_{1j}(u, \cdot), w_j(\cdot)\rangle, \dots,  \sum_{j=1}^p \langle K_{pj}(u, \cdot), w_j(\cdot)\rangle\Big)^{\T}$$ for $\bw \in \eH^p.$ To simplify notation, we use $\bK$ to denote both the kernel function and the operator.
We then define Hilbert space valued sub-Gaussian random variable, which corresponds to an infinite-dimensional analog of the sub-Gaussian random vector in $\eR^p.$

\begin{definition}\label{def-subgauss} Let $W(\cdot)$ be a random variable in $\eH$ and $\Sigma_W: \eH \rightarrow \eH$ be the covariance operator of $W(\cdot)$. Then $W(\cdot)$ is a sub-Gaussian process if there exists a constant $\alpha >0$ such that for all $w \in \eH,$
$\eE[\exp\{\langle w, W - \eE(W)\rangle\}] \leq \exp\{2^{-1}\alpha^2\langle w,\Sigma_{W}(w)\rangle\}.$
\end{definition}

\begin{condition}\label{con:flp} 
(i) $\{\bY_t(\cdot)\}_{t \in \eZ}$ is a sequence of multivariate
functional linear processes with sub-Gaussian errors, i.e., 
$  \bY_t(\cdot) = \sum_{l=0}^{\infty}\bB_l(\bvarepsilon_{t-l})
$ for any $t \in \eZ,$
where $\bB_l=(B_{l,jk})_{p \times p}$ with each $B_{l,jk} \in \eH \otimes \eH,$ $\bvarepsilon_t = (\varepsilon_{t1}, \dots, \varepsilon_{tp})^{\T} \in \eH^p$ and the components in $\{\bvarepsilon_t\}_{t \in \eZ}$ are independent mean-zero sub-Gaussian processes satisfying Definition~\ref{def-subgauss};
(ii) The coefficient functions satisfy $\sum_{l=0}^{\infty} \|B_{l,ij}\|_{\cS, \infty}=O(1);$
(iii) The covariance functions of $\{\varepsilon_{tj}(u)\}_{j=1}^p$ satisfy $\max_{j}\int_{\cU}\cov\{\varepsilon_{tj}(u),\varepsilon_{tj}(u)\}du=O(1).$
\end{condition}

\begin{condition}\label{con:fsm}
For $\{\bY_t(\cdot)\}_{t \in \eZ},$ its spectral density function $\bbf_{y,\theta}= (2\pi)^{-1} \sum_{k \in \eZ}\bSigma_{yy}^{(k)} e^{-ik\theta}$ with the frequency $\theta \in [-\pi, \pi]$ exists, and its functional stability measure defined as 
    \begin{equation}
        \label{df.fsm}
        \cM_{y}=2 \pi \cdot \underset{\theta \in [-\pi, \pi], \bw \in \eH_0^p}{\text{esssup}} \frac{\langle \bw, \bbf_{y,\theta}(\bw)\rangle}{\langle \bw, \bSigma_{yy}^{(0)}(\bw) \rangle} < \infty, 
    \end{equation}
    where $\eH_0^p = \big\{\bw \in \eH: \langle \bw,  \bSigma_{yy}^{(0)}(\bw)\rangle \in (0, \infty)\big\}.$
\end{condition}

The multivariate functional linear process in Condition~\ref{con:flp}(i) generalizes the multivariate (or functional) linear process to the functional (or multivariate) setting. Under this setting, the functional stability measure $\cM_y$ in (\ref{df.fsm}) can be expressed based on $\bbf_{y,\theta}$ and $\bSigma_{yy}^{(0)},$ both of which have explicit expressions \textcolor{blue}{(Fang et al., 2022)}. See other discussions about $\cM_y$ in \textcolor{blue}{Guo and Qiao (2023)}.  
Under Conditions~\ref{con:flp} and \ref{con:fsm}, the rate in (\ref{linfty.Sigma}) can be implied from Lemma~\ref{lem:nonasym} in Section~\ref{supp.lem} of the Supplementary Material.  
In general, we can relax Conditions~\ref{con:flp}(ii) and (iii) by allowing both $\sum_{l=0}^{\infty} \|B_{l,ij}\|_{\cS, \infty}$ and $\max_{j}\int_{\cU}\cov\{\varepsilon_{tj}(u),\varepsilon_{tj}(u)\}du$ to grow at very slow rates as $p$ diverges, then the rate in (\ref{linfty.Sigma}) will depend on these two terms.

\section{Proofs of main theoretical results}
\label{supp.pf}
This section contains all technical proofs of main theoretical results.
In addition to the notations defined in the main paper, we summarize here more notation to be used throughout this supplementary material. For a matrix $\bB=(B_{ij})_{p \times q},$ we denote its matrix $\ell_1$ norm by 
$\|\bB\|_1 = \max_{1\leq j\leq q}\sum_{i=1}^{p}|B_{ij}|.$
For a matrix of bivariate functions $\bK = (K_{ij}(\cdot, \cdot))_{p\times q}$ with each $K_{ij}\in  L_2(\cU\times\cU)$, we define the functional version of 
the matrix $\ell_1$-norm by 
$\|\bK\|_{\cS,1} = \max_{1\leq j\leq q} \sum_{i=1}^p \|K_{ij}\|_{\cS}.$ 

\subsection{Proof of Theorem~\ref{thm:eigvec}}
\label{supp.pf.thm1}
We first assume that Condition~\ref{con:bound} holds, and the procedures can be extended to the case when Condition~\ref{con:white} holds. Recall that 
$$
		\widehat\bM = \sum_{k=1}^{k_0} \iint \widehat\bSigma^{(k)}_{yy}(u,v) \widehat\bW(v) \widehat\bSigma^{(k)}_{yy}(u,v)^\T \, {\rm d}u {\rm d}v,
$$		
$$
		\widecheck\bM = \sum_{k=1}^{k_0} \iint \bSigma^{(k)}_{yy}(u,v) \widehat\bW(v)\bSigma^{(k)}_{yy}(u,v)^\T \, {\rm d}u {\rm d}v,
$$
where
$
\widehat\bW(v) = \bQ\big\{\bQ^\T\widehat{\bSigma}_{yy}^{(0)}(v,v)\bQ\big\}^{-1}\bQ^\T.$ If $\{\bvarepsilon_{t}(\cdot)\}$ is not a white noise sequence, then the decomposition \eqref{eq:cov-relation} should be rewritten as
\begin{equation}
    \label{eq:cov-relation_bound}
    \bSigma^{(k)}_{yy}(u,v) = \bA \bSigma^{(k)}_{xx}(u,v) \bA^{\T} + \bA \bSigma^{(k)}_{x\varepsilon}(u,v)+\bSigma_{\varepsilon\varepsilon}^{(k)}(u,v), ~~u,v \in \cU, ~k \geq 1.
\end{equation}
In this way, $\cC(\bA)\neq\cC(\widecheck{\bM})=\cC(\bK).$ We further decompose $\widecheck\bM=\widecheck{\bM}_{\cL}+\widecheck{\bM}_{\cR}$ into two parts, where $\widecheck{\bM}_{\cL}$ contains the information of the column space of $\bA,$ i.e., $\cC(\bA)=\cC(\widecheck{\bM}_{\cL}),$ and $\widecheck{\bM}_{\cR}$ is the residual part, which is zero if Condition~\ref{con:white} holds, or $\Vert\widecheck{\bM}_{\cR}\Vert=o_\P(\Vert\widecheck{\bM}_{\cL}\Vert)$ if Condition~\ref{con:bound} holds such that $\cC(\bA)$ can be approximated by $\cC(\bK)$. Specifically,
\begin{equation}
    \label{eq:M_R}
    \begin{aligned}
        \widecheck{\bM}_{\cL}=&\bA\left[\sum_{k=1}^{k_0}\iint\Big\{\bSigma_{xx}^{(k)}(u,v)\bA^{\T}+\bSigma_{x\varepsilon}^{(k)}(u,v)\Big\}\widehat{\bW}(v)\Big\{\bSigma_{xx}^{(k)}(u,v)\bA^{\T}+\bSigma_{x\varepsilon}^{(k)}(u,v)\Big\}^{\T}{\rm d}u {\rm d}v\right]\bA^{\T},\\
        \widecheck{\bM}_{\cR}=&\sum_{k=1}^{k_0}\iint\big\{\bA\bSigma_{xx}^{(k)}(u,v)\bA^{\T}+\bA\bSigma_{x\varepsilon}^{(k)}(u,v)\big\}\widehat{\bW}(v)\bSigma_{\varepsilon\varepsilon}^{(k)}(u,v)^{\T}{\rm d}u {\rm d}v\\
        &+\sum_{k=1}^{k_0}\iint\bSigma_{\varepsilon\varepsilon}^{(k)}(u,v)\widehat{\bW}(v)\big\{\bA\bSigma_{xx}^{(k)}(u,v)\bA^{\T}+\bA\bSigma_{x\varepsilon}^{(k)}(u,v)\big\}^{\T}{\rm d}u {\rm d}v\\
        &+\sum_{k=1}^{k_0}\iint\bSigma_{\varepsilon\varepsilon}^{(k)}(u,v)\widehat{\bW}(v)\bSigma_{\varepsilon\varepsilon}^{(k)}(u,v)^{\T}{\rm d}u {\rm d}v.
    \end{aligned}
\end{equation}

\color{black}
Noting that $\widehat\bW(v)$ is scale-invariant with respect to $\bQ,$ we assume here that the entries of $\bQ = (Q_{ij})_{p \times q}$ are independent $N(0, p^{-1})$ instead of $N(0,1)$ to ensure that 
the entries of $\bQ^\T{\widehat\bSigma}_{yy}^{(0)}(v,v)\bQ$ are $O_\P(1)$ instead of $O_\P(p^2).$ Provided that the maximum lag $k_0$ is fixed, without loss of generality we only need to consider the case $k_0=1.$ We organize our proof in the following five steps. 

{\bf Step~1}.
    We will show that $\bQ^\T\widehat{\bSigma}_{yy}^{(0)}(v,v)\bQ$ is positive definite and symmetric with high probability uniformly over $\cU,$ i.e., there exists some positive constant $C>0$ such that, with probability tending to one, 
    \begin{equation}
    \label{M0_bound}
    \inf_{v \in \cU} \lambda_{\min}\big\{\bQ^\T\widehat{\bSigma}_{yy}^{(0)}(v,v)\bQ\big\} > C.  
    \end{equation}

Observe that
    $$
\bQ^\T\widehat{\bSigma}_{yy}^{(0)}(v,v)\bQ = \bQ^\T\bSigma_{yy}^{(0)}(v,v)\bQ +\bQ^\T \left\{\widehat{\bSigma}_{yy}^{(0)}(v,v) - \bSigma_{yy}^{(0)}(v,v)\right\}\bQ.
$$ 
Then we have 
    \begin{equation}
	\label{eigen.ieq}
\left|\lambda_k\Big\{\bQ^\T\widehat{\bSigma}_{yy}^{(0)}(v,v)\bQ \Big\} - \lambda_k\Big\{\bQ^\T\bSigma_{yy}^{(0)}(v,v)\bQ \Big\}\right|
 \le  \left\|\bQ^\T \left\{\widehat{\bSigma}_{yy}^{(0)}(v,v) - \bSigma_{yy}^{(0)}(v,v)\right\}\bQ\right\|_\tF,
    \end{equation}
where $\lambda_k(\bB)$ is the $k$-th largest eigenvalue of a symmetric matrix $\bB.$
By the Hanson-Wright inequality and the union bound of probability, we obtain that there exists some constant $c>0$ such that for any matrix $\bB\in\mathbb R^{p\times p}$ and $t>0,$  
\begin{eqnarray*}
\pr\left\{\frac{\big\|\bQ^\T \bB \bQ -  \mathbb{E}(\bQ^\T \bB \bQ)\big\|_\tF}{\|\bB\|_\tF} \ge t qp^{-1} \right\} \le  2q^2\exp\{ - c  \min(t^2,t)\}.
\end{eqnarray*}
It follows from the concentration inequality above with $\bB=\widehat{\bSigma}_{yy}^{(0)}(v,v)-{\bSigma}_{yy}^{(0)}(v,v),$  the independence between $\bQ$ and $\widehat{\bSigma}_{yy}^{(0)}(v,v),$ the Lipschitz-continuity of $\bSigma_{yy}^{(0)}(v,v)$ and the uniform convergence rate in Condition~\ref{con:rootn}(ii)(iii) that 
\begin{eqnarray*}
&&\sup_{v\in \cU}\left\|\bQ^\T \left\{\widehat{\bSigma}_{yy}^{(0)}(v,v) - \bSigma_{yy}^{(0)}(v,v)\right\}\bQ\right\|_\tF 
\\&\le& O_\P\Big\{qp^{-1} \{\log(qn)\}^{1/2} \sup_{v \in \cU}\|\widehat{\bSigma}_{yy}^{(0)}(v,v) - \bSigma_{yy}^{(0)}(v,v)\|_\tF\Big\}\\
&=& O_\P\big\{q \log(qn)n^{-1/2}\big\}.
\end{eqnarray*}
Combining the above result with (\ref{eigen.ieq}) and the fact that
$$
\inf_{v\in \cU}\lambda_q\{\bQ^\T \bSigma_{yy}^{(0)}(v,v) \bQ\} \ge \inf_{v \in \cU}\lambda_q\{\bQ^\T \bSigma_{\varepsilon\varepsilon}^{(0)}(v,v) \bQ\} >c>0
$$ under Condition~\ref{con:eps}(ii) implies that (\ref{M0_bound}) holds with high probability. 


{\bf Step~2}.
The  bounds of $r$ largest eigenvalues of $\widecheck \bM$ (i.e., $\nu_1\geq \cdots \geq \nu_r$) satisfy
    \begin{equation}
        \label{nu_r_bound}
        \nu_j \asymp p^{2-2\delta}, ~~j = 1,\cdots,r,
    \end{equation}
 with probability tending to one. 
 
First, note that
    \begin{equation}\label{eq:M-order}
    \nu_r\leq \nu_1\leq \|\widecheck \bM\|\leq \iint\|\bSigma^{(1)}_{yy}(u,v)\|^2\|\widehat \bW(v)\|\,{\rm d} u{\rm d} v.
    \end{equation}
    Applying the fact $\|\bQ\| \asymp 1 + o_\P(1)$ and (\ref{M0_bound}) from Step~1, we have that
    \begin{equation}
        \label{bd.Q}
        \sup_{v\in\cU}\|\widehat \bW(v)\| \le \|\bQ\|^2\big[\inf_{v\in\cU} \lambda_{\min}\{\bQ^T \widehat{\bSigma}_{yy}^{(0)}(v,v)\bQ\}\big]^{-1} \lesssim 1+o_\P(1).
    \end{equation}

 It follows from Conditions~\ref{con:strg}--\ref{con:xeps}, the decomposition (\ref{eq:cov-relation_bound}) and Lemma~\ref{lem:xeps} that 
    \begin{eqnarray}
    \nonumber
            \iint \big\|\bSigma^{(1)}_{yy}(u,v)\big\|^2\,{\rm d}u{\rm d}v&\leq& \iint \left(\big\|\bA\bSigma_{xx}^{(1)}(u,v)\bA^{\T}\big\| + \big\|\bA\bSigma_{x\varepsilon}^{(1)}(u,v)\big\|+\big\|\bSigma_{\varepsilon\varepsilon}^{(1)}(u,v)\big\|\right)^2\,{\rm d}u{\rm d}v\\\nonumber
            &\leq& 3\Big\{\big\|\bA\big\|^4\iint\big\|\bSigma_{xx}^{(1)}(u,v)\big\|^2\,{\rm d}u{\rm d}v + \big\|\bA\big\|^2\iint\big\|\bSigma_{x\varepsilon}^{(1)}(u,v)\big\|^2{\rm d}u{\rm d}v\Big.\\
            & &\Big.+\iint\big\|\bSigma_{\varepsilon\varepsilon}^{(1)}(u,v)\big\|^2{\rm d}u{\rm d}v\Big\}\\
            &=& O\Big\{p^{2-2\delta}+p^{1-\delta}\iint\big\|\bSigma_{x\varepsilon}^{(1)}(u,v)\big\|^2\,{\rm d}u{\rm d}v+1\Big\},\label{eq:cov-yy}
    \end{eqnarray}
    where, by Condition~\ref{con:bound}, $\iint\big\|\bSigma_{\varepsilon\varepsilon}^{(1)}(u,v)\big\|^2\,{\rm d}u{\rm d}v\lesssim\sup_{(u,v)\in\cU\times\cU}\Vert\bSigma_{\varepsilon\varepsilon}^{(1)}(u,v)\Vert^2=O(1)$, and by Condition~\ref{con:xeps},
    \begin{equation}\label{eq:cov-xeps}
    \begin{aligned}
        \iint\big\|\bSigma_{x\varepsilon}^{(1)}(u,v)\big\|^2 \,{\rm d} u{\rm d} v 
        &\leq \big\|\bSigma_{x\varepsilon}^{(1)}\big\|_{\cS,\tF}^2 
        = O\Big\{\max\limits_{1\leq i\leq r}\sum_{j=1}^p\big\|\Sigma_{x\epsilon,ij}^{(1)}\big\|^2_{\cS}\Big\}\\
        &\leq O\Big\{\big\|\bSigma_{x\varepsilon}^{(1)}\big\|_{\cS,\infty}^2\Big\} = o(p^{1-\delta}).
    \end{aligned}
    \end{equation} 
    Combining \eqref{eq:M-order}--\eqref{eq:cov-xeps} yields that $\nu_r\lesssim p^{2-2\delta}\{1+o_\P(1)\}.$ 

     We next give a lower bound on $\nu_r.$ By Conditions~\ref{con:x}, \ref{con:strg}, Lemma~3 of \textcolor{blue}{Wang and Xi (1997)} and (\ref{M0_bound}) from Step~1, we first obtain that
    \begin{equation*}
        \begin{aligned}
            I_1 & =\lambda_{r}\left\{\iint \bSigma^{(1)}_{xx}(u,v)\bA^{\T}\widehat\bW(v)\bA\bSigma^{(1)}_{xx}(u,v)^\T \,{\rm d}u{\rm d}v\right\}\\
            &\ge \lambda_{r}\left\{\iint \bSigma^{(1)}_{xx}(u,v)\bSigma^{(1)}_{xx}(u,v)^\T\,{\rm d}u{\rm d}v\right\} \inf_{u,v\in \cU}\lambda_{r}\left\{\bA^{\T}\widehat\bW(v) \bA\right\}\\
            & \gtrsim \lambda_{r}(\bA^\T\bQ\bQ^\T\bA) \cdot \inf_{v \in \cU}\lambda_{q - r + 1}\left(\big\{\bQ^T\widehat{\bSigma}_{yy}^{(0)}(v,v)\bQ\big\}^{-1}\right)\\
            & \gtrsim p^{1-\delta} \{1+o_\P(1)\}.
            \end{aligned}
    \end{equation*}
    
    It follows from \eqref{bd.Q}, \eqref{eq:cov-xeps}, Lemma~\ref{lem:xeps}, and Conditions~\ref{con:bound} and \ref{con:strg} that
    \begin{equation}
        \label{eq:I_234_bound}
        \begin{aligned}
            I_2 & = \left\|\iint \bSigma^{(1)}_{x\varepsilon}(u,v)\widehat\bW(v)\bSigma^{(1)}_{xx}(u,v)\bA^{\T}\,{\rm d}u{\rm d}v\right\|\\
            &\le \left\{\iint \|\bSigma^{(1)}_{x\varepsilon}(u,v))\|^2 \,{\rm d}u{\rm d}v\iint \|\bSigma^{(1)}_{xx}(u,v)\|^2 \,{\rm d}u{\rm d}v \right\}^{1/2} \sup_{v\in \cU}\|\widehat\bW(v)\| \|\bA\| = o_\P(p^{1-\delta}),\\
            I_3 & = \left\| \iint \big\{\bA\bSigma_{xx}^{(1)}(u,v)\bA^{\T}+\bA\bSigma_{x\varepsilon}^{(1)}(u,v)\big\}\widehat{\bW}(v)\bSigma_{\varepsilon\varepsilon}^{(1)}(u,v)^{\T} {\rm d}u{\rm d}v\right\|=O_\P(p^{1-\delta}),\\
            I_4 & =\left\|\iint\bSigma_{\varepsilon\varepsilon}^{(1)}(u,v)\widehat{\bW}(v)\bSigma_{\varepsilon\varepsilon}^{(1)}(u,v)^{\T} {\rm d}u{\rm d}v\right\|=O_\P(1).
        \end{aligned}
    \end{equation}
    By Weyl's inequality, the decomposition (\ref{eq:cov-relation_bound}) and Condition~\ref{con:strg}, $\nu_r$ can be bounded from below by
    \begin{equation*}\label{thm:eigvec-pf2}
    \begin{aligned}
        \nu_r & = \lambda_{r}\Big\{\iint \bSigma^{(1)}_{yy}(u,v)\widehat\bW(v)\bSigma^{(1)}_{yy}(u,v)^{\T} \,{\rm d}u{\rm d}v\Big\}\\
        &\geq \lambda_{r}(\bA^\T\bA)\cdot\lambda_r\Big[\iint \big\{\bSigma^{(1)}_{xx}(u,v)\bA^{\T}+\bSigma^{(1)}_{x\varepsilon}(u,v)\big\}\widehat\bW(v)\big\{\bSigma^{(1)}_{xx}(u,v)\bA^{\T}+\bSigma^{(1)}_{x\varepsilon}(u,v)\big\}^{\T} \,{\rm d}u{\rm d}v\Big]\\
        &\ge \|\bA\|_{\min}^2 \cdot \lambda_{r}\Big\{\iint \bSigma^{(1)}_{xx}(u,v)\bA^{\T}\widehat\bW(v)\bA\bSigma^{(1)}_{xx}(u,v)^\T\,{\rm d}u{\rm d}v\Big\}\\
        &~~~~+2\|\bA\|_{\min}^2\cdot\sigma_{r}\Big\{\iint \bSigma^{(1)}_{x\varepsilon}(u,v)\widehat\bW(v)\bSigma^{(1)}_{xx}(u,v)\bA^{\T}\,{\rm d}u{\rm d}v\Big\}\\
        &~~~~+2\sigma_{r}\Big[\iint \big\{\bA\bSigma_{xx}^{(1)}(u,v)\bA^{\T}+\bA\bSigma_{x\varepsilon}^{(1)}(u,v)\big\}\widehat{\bW}(v)\bSigma_{\varepsilon\varepsilon}^{(1)}(u,v)^{\T} {\rm d}u{\rm d}v\Big]\\
        &~~~~+\lambda_r\Big\{\iint\bSigma_{\varepsilon\varepsilon}^{(1)}(u,v)\widehat{\bW}(v)\bSigma_{\varepsilon\varepsilon}^{(1)}(u,v)^{\T} {\rm d}u{\rm d}v\Big\}\\
        &\ge \|\bA\|_{\min}^2  \cdot (I_1 - 2I_2) -2I_3-I_4 \gtrsim p^{2-2\delta}\{1+o_\P(1)\}.
    \end{aligned}
    \end{equation*}
    Hence we obtain that  $\nu_r \asymp p^{2-2\delta}$ with high probability.

\color{black}
{\bf Step~3}.
We will show that 
\begin{equation}
\label{err.bM}
\|\widehat\bM - \widecheck\bM\| = O_\P(p^{2-\delta}n^{-1/2}),
\end{equation}
and further
\begin{equation}
    \label{eq:err_bM_L}
    \|\widehat\bM - \widecheck\bM_{\cL}\| = O_\P(p^{2-\delta}n^{-1/2}+p^{1-\delta}).
\end{equation}
\color{black}
Observe that
 \begin{equation*}\label{eq:tilde-M}
	\begin{aligned}
	 &\widehat\bSigma^{(1)}_{yy}(u,v) \widehat\bW(v) \widehat\bSigma^{(1)}_{yy}(u,v)^\T - \bSigma^{(1)}_{yy}(u,v) \widehat\bW(v)\bSigma^{(1)}_{yy}(u,v)^\T\\
	 =& \big\{\widehat\bSigma^{(1)}_{yy}(u,v) - \bSigma^{(1)}_{yy}(u,v)\big\} \widehat\bW(v) \big \{\widehat\bSigma^{(1)}_{yy}(u,v) - \bSigma^{(1)}_{yy}(u,v)\big\}^\T \\
	  & + \big \{\widehat\bSigma^{(1)}_{yy}(u,v) - \bSigma^{(1)}_{yy}(u,v)\big\} \widehat\bW(v) \bSigma^{(1)}_{yy}(u,v)^\T\\
	  & + \bSigma^{(1)}_{yy}(u,v)\widehat\bW(v) \big \{\widehat\bSigma^{(1)}_{yy}(u,v) - \bSigma^{(1)}_{yy}(u,v)\big\}^\T.
\end{aligned}
    \end{equation*}

By  $\|\bQ\| \asymp 1 + o_\P(1)$, Condition~\ref{con:rootn}(i) and Markov's inequality, we have that
 \begin{equation*}
	\begin{aligned}
&\iint \left\|\left \{\widehat\bSigma^{(1)}_{yy}(u,v) - \bSigma^{(1)}_{yy}(u,v)\right\}\bQ \right\|^2{\rm d}u{\rm d}v 
& \le \sum_{i=1}^p \sum_{j=1}^p \left \|\widehat\Sigma^{(1)}_{yy,ij} - \Sigma^{(1)}_{yy,ij}\right\|_{\cS}^2 \|\bQ\|^2 = O_\P(p^2n^{-1}).
\end{aligned}
    \end{equation*}
Furthermore,
\eqref{eq:cov-yy} and \eqref{eq:cov-xeps} imply that 
\begin{equation*}
	\begin{aligned}
&\iint \left\|\bSigma^{(1)}_{yy}(u,v)\bQ \right\|^2{\rm d}u{\rm d}v \le \iint \left\|\bSigma^{(1)}_{yy}(u,v)\right\|^2 \, {\rm d}u{\rm d}v \, \|\bQ\|^2 = p^{2-2\delta}\{ 1+ o_\P(1)\}.
\end{aligned}
    \end{equation*}
It follows from the above decomposition, two upper bound results, Cauchy-Schwartz inequality and $\sup_{v \in \cU}\lambda_{\max}\big(\{\bQ^\T\widehat{\bSigma}_{yy}^{(0)}(v,v)\bQ\}^{-1}\big)= O_\P(1)$ from Step~1 that (\ref{err.bM}) holds.
Then by the definition of $\widecheck{\bM}_{\cR}$ in \eqref{eq:M_R}, it follows from \eqref{eq:I_234_bound} that $\Vert\widecheck{\bM}_{\cR}\Vert\le k_0(2I_3+I_4)=O_\P(p^{1-\delta})$, which implies \eqref{eq:err_bM_L}.

{\bf Step~4}. Combining (\ref{nu_r_bound}) and (\ref{eq:err_bM_L}) from Steps~2 and 3 with Lemma 3 and the proof of Theorem 1 in \textcolor{blue}{Lam et al. (2011)}, we have that there exists $\widebar{\bA}\in\eR^{p\times r}$ and an orthogonal matrix $\bU$ such that $\widehat{\bA} = \widebar{\bA}\bU$ and 
\begin{equation*}
\label{thm:eigvec-eq}
        \|\widebar \bA - \bA\| = O_\P(p^{\delta}n^{-1/2}+p^{1-\delta}).
 \end{equation*}
It then follows from the orthogonality of $\bA$ that
    \begin{equation*}\label{thm:dist-pf1}
        \begin{aligned}
            \mathcal D^2\Big(\mathcal C(\bA), \mathcal C(\widehat\bA)\Big) &= r^{-1}\tr\Big(\bI - \bA\bA^{\T}\widehat\bA\widehat{\bA}^{\T}\Big) = r^{-1} \tr(\bA^{\T}(\bI-\widehat\bA\widehat{\bA}^{\T})\bA)\\
            &= r^{-1}\Big\{\tr(\bA^{\T}(\bI-\widehat\bA\widehat{\bA}^{\T})\bA) - \tr(\bA^{\T}(\bI-\bA\bA^{\T})\bA)\Big\}\\
            &=r^{-1}\tr\Big(\bA^{\T}( \bA\bA^{\T} - \widehat\bA\widehat{\bA}^{\T})\bA\Big)\leq \Big\|\bA^{\T}( \bA\bA^{\T} - \widebar\bA\widebar{\bA}^{\T})\bA\Big\|,
        \end{aligned}
    \end{equation*}
    where the last term is further bounded by
    \begin{equation*}\label{thm:dist-pf2}
        \begin{aligned}
            \|-\bA^{\T}(\bA-\widebar{\bA})(\bA-\widebar{\bA})^{\T}\bA + (\bA-\widebar{\bA})^{\T}(\bA-\widebar{\bA})\| \leq 2\|\widebar{\bA}-\bA\|^2 = O_\P(p^{2\delta}n^{-1}+p^{2-2\delta}).
        \end{aligned}
    \end{equation*}
    The result in the first part of Theorem~\ref{thm:eigvec} follows immediately. 
    
    If Condition~\ref{con:white} holds, then $\widecheck{\bM}_{\cR}=\bzero,$ and the rate in \eqref{eq:err_bM_L} can be reduced to $\|\widehat\bM - \widecheck\bM_{\cL}\| = O_\P(p^{2-\delta}n^{-1/2}).$ Following the similar procedures, we can show that $\mathcal D\big(\mathcal C(\bA), \mathcal C(\widehat\bA)\big)=O_\P(p^{\delta}n^{-1/2}),$   
    which completes the proof. \hfill$\square$

\color{black}
\subsection{Proof of Theorem~\ref{thm:eigval}}
\label{supp.pf.thm2}
(i) Denote $(\nu_j,\bgamma_j),(\tilde\nu_j,\widetilde\bgamma_j)$ and $(\hat \nu_j,\widehat {\bgamma}_j)$ by the eigenpairs of $\widecheck{\bM},\widecheck{\bM}_{\cL}$ and $\widehat \bM$, respectively. For simplicity we assume $\bgamma_j^\T\widehat\bgamma_j\geq 0$ and $\widetilde\bgamma_j^\T\widehat\bgamma_j\geq 0$ here. First note that for $j=1,2,\dots,r,$
    \begin{equation*}
        \begin{aligned}
            \hat\nu_j-\nu_j= \widehat{\bgamma}_j^\T\widehat\bM\widehat{\bgamma}_j-\bgamma_j^\T\widecheck\bM\bgamma_j=I_1+I_2+I_3+I_4+I_5,
        \end{aligned}
    \end{equation*}
    where
     $       I_1 = (\widehat{\bgamma}_j-\bgamma_j)^\T(\widehat\bM-\widecheck\bM)\widehat{\bgamma}_j,~ I_2 = (\widehat{\bgamma}_j-\bgamma_j)^\T\widecheck\bM(\widehat{\bgamma}_j-\bgamma_j),
            ~I_3 = (\widehat{\bgamma}_j-\bgamma_j)^\T\widecheck\bM\bgamma_j, 
            ~I_4 = \bgamma_j^\T(\widehat\bM-\widecheck\bM)\widehat{\bgamma}_j,
            ~I_5 = \bgamma_j^\T\widecheck\bM(\widehat{\bgamma}_j-\bgamma_j).$

    By $\|\widehat\bM-\widecheck\bM\| =O_\P(p^{2-\delta} n^{-1/2})$ and $\|\widecheck \bM\| = O_\P(p^{2-2\delta})$, we have $\|\widehat{\bgamma}_j-\bgamma_j\|=O_\P(p^{\delta}n^{-1/2})$. Then,
    $$|I_1| + |I_2| = O_\P(p^2n^{-1}), \qquad |I_3|+|I_4|+|I_5|=O_\P(p^{2-\delta}n^{-1/2}),$$
    which means that $|\hat{\nu}_j-\nu_j| = O_\P(p^{2-\delta}n^{-1/2})$ under $p^{\delta}n^{-1/2}=o(1).$

    To prove the result for $j=r+1,\dots,p$, we introduce 
    \begin{equation*}\label{eq:Mtilde}
        \widebar\bM = \sum_{k=1}^{k_0} \iint \widehat{\bSigma}^{(k)}_{yy}(u,v)\widehat \bW(v)\bSigma^{(k)}_{yy}(u,v)^{\T} \,\mathrm d u\mathrm d v.
    \end{equation*}
 For each $j = r+1, \ldots, p,$ consider the following decomposition 
    \begin{equation*}
        \begin{aligned}
            \hat\nu_j= \widehat{\bgamma}_j^\T\widehat\bM\widehat{\bgamma}_j = I_6+I_7+I_8,
        \end{aligned}
    \end{equation*}
    where
    $       I_6 = \widehat{\bgamma}_j^\T(\widehat\bM-\widebar{\bM} - \widebar{\bM}^\T  + \widecheck\bM_{\cL})\widehat{\bgamma}_j,
            ~I_7 = 2\widehat{\bgamma}_j^\T(\widebar{\bM}-\widecheck\bM_{\cL})(\widehat{\bgamma}_j-\widetilde\bgamma_j), ~I_8 = (\widehat{\bgamma}_j-\widetilde\bgamma_j)^\T\widecheck\bM_{\cL}(\widehat{\bgamma}_j-\widetilde\bgamma_j).$
    Note that, under Condition~\ref{con:rootn}(i), 
    $$
    |I_6| \le \|\widehat\bM-\widebar{\bM} - \widebar{\bM}^\T  + \widecheck \bM\| \lesssim \sum_{k=1}^{k_0} \|\widehat{\bSigma}^{(k)}_{yy} - \bSigma^{(k)}_{yy}\|_{\cS}^2 =O_\P(p^2n^{-1}). 
    $$
    By the similar arguments to prove Theorem \ref{thm:eigvec},  we can show that $\|\widebar\bM -  \widecheck\bM_{\cL}\| =O_\P(p^{2-\delta}n^{-1/2}+p^{1-\delta})$ and 
    $\|\widehat{\bgamma}_j-\widetilde{\bgamma}_j\|=O_\P(p^{\delta}n^{-1/2}+p^{\delta-1}),$ which together imply that
    $|I_7|+|I_8|=O_\P(p^2n^{-1}+1)$. Combining the above results yields that $\hat{\nu}_j=O_\P(p^2n^{-1}+1)$ for $j=r+1,\dots,p,$ which completes the  proof of the first part.

    \color{black}
    (ii) We next turn to prove the second part. It follows from $\nu_j\asymp p^{2-2\delta}$ for $j =1,\ldots,r$ in (\ref{nu_r_bound}) and $p^{\delta}n^{-1/2} = o(1)$ that
    $$
    \frac{|\hat{\nu}_j-\nu_j|}{\nu_j} = O_\P(p^{2-\delta}n^{-1/2}p^{2\delta -2}) = O_\P(p^{\delta}n^{-1/2}) = o_\P(1),~~j=1, \dots, r.
    $$
    As a result, for $j = 1,\ldots,r-1,$
    $$
    \frac{\hat \nu_{j+1}}{\hat \nu_j} = \frac{\nu_{j+1}}{\nu_j}\frac{1+\nu_{j+1}^{-1}(\hat \nu_{j+1} - \nu_{j+1})}{1+\nu_{j}^{-1}(\hat \nu_{j} - \nu_{j})} = 
    \frac{\nu_{j+1}}{\nu_j} \frac{1+o_\P(1)}{1+o_\P(1)} \asymp 1
    $$
    with probability tending to one, and
    $$
    \frac{\hat \nu_{r+1}}{\hat \nu_r} = \frac{1}{\nu_{r}}\frac{\hat \nu_{r+1}}{1+\nu_{r}^{-1}(\hat \nu_{r} - \nu_{r})} = O_\P(p^{2\delta-2}p^2n^{-1}+p^{2\delta-2}) = O_\P(p^{2\delta} n^{-1}+p^{2\delta-2}).
    $$
    We complete the proof of the second part. \hfill$\square$
    \color{black}


\subsection{Proof of Theorem~\ref{thm.r}}
\label{supp.pf.thm3}
By Theorem~\ref{thm:eigval}, $|\hat \nu_j -\nu_j| = O_\P(\beta_n)$ for $j=1, \dots, r$ and $|\hat \nu_j| = O_\P(\tilde \beta_n)$ for $j=r+1, \dots, p,$ where $\beta_n=p^{2-\delta}n^{-1/2}$ and $\widetilde \beta_n = p^2n^{-1}+1.$ Under the event $\Omega_n =\{\nu_j \asymp p^{2-2\delta},j=1, \dots, r\},$ it follows from (\ref{nu_r_bound}) that $\pr(\Omega_n) \to 1.$ We next verify the following three conditions (i), (ii) and (iii):
\begin{enumerate}
\item[(i)] $(\vartheta_n + \widetilde \beta_n)/(\nu_r^2/\nu_1) \asymp (\vartheta_n + p^2n^{-1}+1)/p^{2-2\delta} \to 0;$
\item[(ii)] $\beta_n/\nu_r \asymp p^{2-\delta}n^{-1/2}/p^{2-2\delta} = p^{\delta}n^{-1/2}\to 0;$
\item[(iii)] $\widetilde \beta_n^2/(\vartheta_n \nu_r)\asymp(p^2n^{-1}+1)^2/(\vartheta_n p^{2-2\delta})\asymp\vartheta_n^{-1} n^{-2} p^{2+2\delta} +\vartheta_n^{-1}  p^{-2+2\delta} \to 0.$
\end{enumerate}
Under (i), (ii) and (iii), we apply Proposition~1 of \textcolor{blue}{Han et al. (2022)} and obtain 
$\pr(\hat r=r\mid\Omega_n) \to 1$ with $\hat r$ defined in (\ref{mod.ratio.est}). Noting that
$\pr(\hat r \neq r) \leq \pr(\hat r \neq r\mid\Omega_n) + \pr(\Omega_n^C) \to 0,$ we complete the proof of this theorem.\hfill$\square$


\subsection{Proof of Theorem~\ref{thm:hdeigspace}}
\label{supp.pf.thm4}%
We organize our proof in the following three steps.

{\bf Step~1}. With the choice of $\eta_k \asymp \mathcal{M}_{y}(\log p/n)^{1/2},$ we will show that     
\begin{equation}
\label{err.threshold}
    \begin{aligned}
         \big\|\mathcal T_{\eta_k}\{\widehat \bSigma^{(k)}_{yy}\} - \bSigma^{(k)}_{yy}\big\|_{\cS,1} &= O_\P\left\{c_1(p)\mathcal{M}_{y}^{1-\tau}\Big(\frac{\log p}{n}\Big)^{\frac{1-\tau}{2}}\right\},\\
         \big\|\mathcal T_{\eta_k}\{\widehat \bSigma^{(k)}_{yy}\} - \bSigma^{(k)}_{yy}\big\|_{\cS,\infty} &= O_\P\left[\big\{c_1(p)+c_2(p)\big\}\mathcal{M}_{y}^{1-\tau}\Big(\frac{\log p}{n}\Big)^{\frac{1-\tau}{2}}\right].
    \end{aligned}
\end{equation}

By the definition of functional matrix $\ell_1$ norm $\|\cdot\|_{\cS, 1}$ and the triangle inequality, 
        \begin{equation*}
            \begin{aligned}
                \big\|\mathcal T_{\eta_k}\{\widehat \bSigma^{(k)}_{yy}\} - \bSigma^{(k)}_{yy}\big\|_{\cS,1} \leq \big\|\thresh - \bSigma^{(k)}_{yy}\big\|_{\cS,1} + \big\|\mathcal T_{\eta_k}\{\widehat \bSigma^{(k)}_{yy}\} - \thresh\big\|_{\cS,1}.
            \end{aligned}
        \end{equation*}
        
        For the first term, it follows from Lemma \ref{lem:sparse} that
        \begin{equation}
            \begin{aligned}
                 \big\|\thresh - \bSigma^{(k)}_{yy}\big\|_{\cS,1} &= \max\limits_{1\leq j\leq p}\sum_{i=1}^p \|\Sigma^{(k)}_{yy,ij}\|_{\cS}I\big\{\|\Sigma^{(k)}_{yy,ij}\|_{\cS}< \eta_k\big\}\\ 
                 &< \eta_k^{1-\tau}\max\limits_{1\leq j\leq p}\sum_{i=1}^p\|\Sigma^{(k)}_{yy,ij}\|_{\cS}^\tau\lesssim \eta_k^{1-\tau}c_1(p).
                 \label{bd.first}
            \end{aligned}
        \end{equation}
        For the second term, observe that
        \begin{align*}
                 \big\|\mathcal T_{\eta_k}\{\widehat \bSigma^{(k)}_{yy}\} - \thresh\big\|_{\cS,1} &= \max\limits_{1\leq j\leq p}\sum_{i=1}^p \big\|\widehat{\Sigma}_{yy,ij}^{(k)} I\{\|\widehat{\Sigma}_{yy,ij}^{(k)}\|_{\cS}\geq\eta_k\}- \Sigma_{yy,ij}^{(k)} I\{\|\Sigma_{yy,ij}^{(k)}\|_{\cS}\geq\eta_k\}\big\|_{\cS}\\
                 &\leq \max\limits_{1\leq j\leq p}\sum_{i=1}^p \big\|\widehat{\Sigma}_{yy,ij}^{(k)} - \Sigma_{yy,ij}^{(k)}\big\|_{\cS}I\big\{\|\widehat{\Sigma}_{yy,ij}^{(k)}\|_{\cS}\geq\eta_k,\|\Sigma_{yy,ij}^{(k)}\|_{\cS}\geq\eta_k\big\}\\
                 &~~~~+\max\limits_{1\leq j\leq p}\sum_{i=1}^p \big\|\widehat{\Sigma}_{yy,ij}^{(k)} \|_{\cS}I\big\{\|\widehat{\Sigma}_{yy,ij}^{(k)}\big\|_{\cS}\geq\eta_k,\|\Sigma_{yy,ij}^{(k)}\|_{\cS}<\eta_k\big\}\\
                 &~~~~+\max\limits_{1\leq j\leq p}\sum_{i=1}^p \big\|\Sigma_{yy,ij}^{(k)}\big\|_{\cS}I\big\{\|\widehat{\Sigma}_{yy,ij}^{(k)}\|_{\cS}<\eta_k,\|\Sigma_{yy,ij}^{(k)}\|_{\cS}\geq\eta_k\big\}\\
                 &:= I_1+I_2+I_3.
        \end{align*}
        
        Denote $Z = \max\limits_{1\leq i,j\leq p}\big\|\widehat{\Sigma}_{yy,ij}^{(k)} - \Sigma_{yy,ij}^{(k)}\big\|_{\cS}$. We first bound $I_1.$ By Lemma \ref{lem:sparse},
        \begin{equation*}
            \begin{aligned}
                 I_1 &\leq \max\limits_{1\leq i,j\leq p}\big\|\widehat{\Sigma}_{yy,ij}^{(k)} - \Sigma_{yy,ij}^{(k)}\big\|_{\cS}\cdot\max\limits_{1\leq j\leq p}\sum_{i=1}^p I\big\{\|\widehat{\Sigma}_{yy,ij}^{(k)}\|_{\cS}\geq\eta_k,\|\Sigma_{yy,ij}^{(k)}\|_{\cS}\geq\eta_k\big\}\\
                 &\leq Z\cdot \max\limits_{1\leq j\leq p}\sum_{i=1}^p I\big\{\|\Sigma_{yy,ij}^{(k)}\big\|_{\cS}\geq\eta_k\}\lesssim Z\eta_k^{-\tau}c_1(p).
            \end{aligned}
        \end{equation*}
        We next bound $I_2.$ By the triangle inequality and Lemma \ref{lem:sparse}, 
        \begin{equation*}
            \begin{aligned}
                 I_2 &\leq \max\limits_{1\leq j\leq p}\sum_{i=1}^p \big\|\widehat{\Sigma}_{yy,ij}^{(k)} - \Sigma_{yy,ij}^{(k)}\big\|_{\cS}I\big\{\|\widehat{\Sigma}_{yy,ij}^{(k)}\|_{\cS}\geq\eta_k,\|\Sigma_{yy,ij}^{(k)}\|_{\cS}<\eta_k\big\}\\
                 &~~~~+\max\limits_{1\leq j\leq p}\sum_{i=1}^p \big\|\Sigma_{yy,ij}^{(k)}\big\|_{\cS}I\big\{\|\Sigma_{yy,ij}^{(k)}\|_{\cS}<\eta_k\big\}\\
                 &\lesssim \max\limits_{1\leq j\leq p}\sum_{i=1}^p \big\|\widehat{\Sigma}_{yy,ij}^{(k)} - \Sigma_{yy,ij}^{(k)}\big\|_{\cS}I\big\{\|\widehat{\Sigma}_{yy,ij}^{(k)}\|_{\cS}\geq\eta_k,\|\Sigma_{yy,ij}^{(k)}\|_{\cS}<\eta_k\big\} + \eta_k^{1-\tau}c_1(p)\\
                 &:= I_4 + \eta_k^{1-\tau}c_1(p).
            \end{aligned}
        \end{equation*}
        
        We take certain $\theta\in (0,1).$ 
        Let $N(t)=\max\limits_{1\leq j\leq p} \sum_{i=1}^p I\big\{\|\widehat{\Sigma}_{yy,ij}^{(k)}-\Sigma_{yy,ij}^{(k)}\|_{\cS}\geq t\big\}.$ By the triangle inequality and Lemma~\ref{lem:sparse},
        \begin{equation*}
            \begin{aligned}
                 I_4 &\leq \max\limits_{1\leq j\leq p}\sum_{i=1}^p \big\|\widehat{\Sigma}_{yy,ij}^{(k)} - \Sigma_{yy,ij}^{(k)}\big\|_{\cS}I\big\{\|\widehat{\Sigma}_{yy,ij}^{(k)}\|_{\cS}\geq\eta_k,\|\Sigma_{yy,ij}^{(k)}\|_{\cS}\leq \theta\eta_k\big\}\\
                 &~~~~+ \max\limits_{1\leq j\leq p}\sum_{i=1}^p \big\|\widehat{\Sigma}_{yy,ij}^{(k)} - \Sigma_{yy,ij}^{(k)}\big\|_{\cS}I\big\{\|\widehat{\Sigma}_{yy,ij}^{(k)}\|_{\cS}\geq\eta_k,\theta\eta_k<\|\Sigma_{yy,ij}^{(k)}\|_{\cS}< \eta_k\big\}\\
                 &\leq Z \cdot \max\limits_{1\leq j\leq p}\sum_{i=1}^p I\big\{\|\widehat{\Sigma}_{yy,ij}^{(k)}-\Sigma_{yy,ij}^{(k)}\|_{\cS}\geq(1-\theta)\eta_k\big\} + Z\cdot \max\limits_{1\leq j\leq p}\sum_{i=1}^p I\big\{\|\Sigma_{yy,ij}^{(k)}\|_{\cS}>\theta\eta_k\big\}\\
                 &\lesssim Z \cdot N\big\{(1-\theta)\eta_k\big\} + Z(\theta\eta_k)^{-\tau}c_1(p).
            \end{aligned}
        \end{equation*}
        The above bounds imply that
        \begin{equation*}
            \begin{aligned}
                 I_2 \lesssim Z\cdot N\big\{(1-\theta)\eta_k\big\} + \eta_k^{1-\tau}c_1(p) +  Z(\theta\eta_k)^{-\tau}c_1(p).
            \end{aligned}
        \end{equation*}
        We finally bound $I_3.$ By the triangle inequality and Lemma \ref{lem:sparse}, 
        \begin{equation*}
            \begin{aligned}
                 I_3&\leq  \max\limits_{1\leq j\leq p}\sum_{i=1}^p \big\|\widehat{\Sigma}_{yy,ij}^{(k)} - \Sigma_{yy,ij}^{(k)}\big\|_{\cS}I\big\{\|\widehat{\Sigma}_{yy,ij}^{(k)}\|_{\cS}<\eta_k,\|\Sigma_{yy,ij}^{(k)}\|_{\cS}\geq \eta_k\big\}\\
                 &~~~~+\max\limits_{1\leq j\leq p}\sum_{i=1}^p \big\|\widehat{\Sigma}_{yy,ij}^{(k)}\big\|_{\cS}I\big\{\|\widehat{\Sigma}_{yy,ij}^{(k)}\big\|_{\cS}<\eta_k,\|\Sigma_{yy,ij}^{(k)}\|_{\cS}\geq \eta_k\}\\
                 &\leq Z\cdot \max\limits_{1\leq j\leq p}\sum_{i=1}^p I\big\{\|\widehat{\Sigma}_{yy,ij}^{(k)}\|_{\cS}\geq\eta_k\big\} + \eta_k \max\limits_{1\leq j\leq p}\sum_{i=1}^p I\big\{\|\Sigma_{yy,ij}^{(k)}\|_{\cS}\geq\eta_k\big\}\\
                 &\lesssim Z \eta_k^{-\tau}c_1(p) + \eta_k^{1-\tau}c_1(p).
            \end{aligned}
        \end{equation*}
        Combining the bounds for $I_1, I_2$ and $I_3$ yields that
        \begin{equation}
        \begin{aligned}
        \label{bd.second}
                 \big\|\mathcal T_{\eta_n}\{\widehat \bSigma^{(k)}_{yy}\} - \thresh\big\|_{\cS,1} \lesssim &Z \eta_k^{-\tau}c_1(p) + \eta_k^{1-\tau}c_1(p) \\ &+Z(\theta\eta_k)^{-\tau}c_1(p)
              + Z \cdot N\big\{(1-\theta)\eta_k\big\}.
        \end{aligned}
        \end{equation}
        By Lemma \ref{lem:nonasym} and the choice of $\eta_k =c_3 {\cal M}_y (\log p/n)^{1/2}=o(1)$ for sufficiently large $c_3,$
        \begin{equation*}
            \begin{aligned}
            \mathbb P\Big(N\big\{(1-\theta)\eta_k\big\}>0\Big)&=\mathbb P\Big(\max\limits_{1\leq i,j\leq p} \|\widehat{\Sigma}_{yy,ij}^{(k)}-\Sigma_{yy,ij}^{(k)}\|_{\cS}\geq (1-\theta)\eta_k\Big)\\
            &\leq c_1p^2\exp\{-c_2(1-\theta)^2c_3^2\log p\}\to 0.
            \end{aligned}
        \end{equation*}
        Hence $N\big\{(1-\theta)\eta_k\big\}=o_\P(1).$ Applying Lemma \ref{lem:nonasym} again, we obtain that $Z = O_\P(\eta_k),$ which together with (\ref{bd.first}) and (\ref{bd.second}) implies that
        $$ \big\|\mathcal T_{\eta_n}\{\widehat \bSigma^{(k)}_{yy}\} - \bSigma^{(k)}_{yy}\big\|_{\cS,1} = O_\P\left\{c_1(p)\mathcal{M}_{y}^{1-\tau}\Big(\frac{\log p}{n}\Big)^{\frac{1-\tau}{2}}\right\}.
        $$
        By the similar procedure, the second result of (\ref{err.threshold}) can be derived and the proof is omitted.
        
        {\bf Step~2}. We will show that
\begin{equation}
\label{err.Mtilde}
    \big\|\widetilde\bM -\widecheck \bM\big\| = O_\P\left[c_1(p)\big\{c_1(p)+c_2(p)\big\}\mathcal{M}_{y}^{1-\tau}\Big(\frac{\log p}{n}\Big)^{\frac{1-\tau}{2}}\right].
\end{equation}
Note that
\begin{equation*}
                \begin{aligned}
                   \big\|\widetilde \bM-\widecheck\bM\big\| =& O_\P\Bigg(\Big\|\iint\big[\mathcal T_{\eta_k}\{\widehat \bSigma^{(k)}_{yy}\}(u,v)-\bSigma^{(k)}_{yy}(u,v)\Big]\widehat\bW(v)\big[\mathcal T_{\eta_k}\{\widehat \bSigma^{(k)}_{yy}\}(u,v)-\bSigma^{(k)}_{yy}(u,v)\Big]^{\T}\,{\rm d}u{\rm d}v\Big\| \\
                    &~~~~+ \Big\|\iint\big[\mathcal T_{\eta_k}\{\widehat \bSigma^{(k)}_{yy}\}(u,v)-\bSigma^{(k)}_{yy}(u,v)\big]\widehat\bW(v)\bSigma^{(k)}_{yy}(u,v)^{\T}\,{\rm d}u{\rm d}v\Big\|\Bigg). 
                \end{aligned}
            \end{equation*} 
            
            By Lemma \ref{lem:int-mat} and (\ref{bd.Q}), (\ref{err.threshold}), we obtain that
            \begin{align*}
                 &\Big\|\iint\big[\mathcal T_{\eta_k}\{\widehat \bSigma^{(k)}_{yy}\}(u,v)-\bSigma^{(k)}_{yy}(u,v)\big]\widehat\bW(v)\big[\mathcal T_{\eta_k}\{\widehat \bSigma^{(k)}_{yy}\}(u,v)-\bSigma^{(k)}_{yy}(u,v)\big]^{\T}\,{\rm d}u{\rm d}v\Big\|\\
                 \lesssim& \big\|\mathcal T_{\eta_k}\{\widehat \bSigma^{(k)}_{yy}\}-\bSigma^{(k)}_{yy}\big\|_{\cS,1}\big\|\mathcal T_{\eta_k}\{\widehat \bSigma^{(k)}_{yy}\}-\bSigma^{(k)}_{yy}\big\|_{\cS,\infty}\big\{1+o_\P(1)\big\}\\
                 =& O_\P\left[c_1(p)\big\{c_1(p)+c_2(p)\big\}\mathcal{M}_{y}^{2(1-\tau)}\Big(\frac{\log p}{n}\Big)^{1-\tau}\right],
            \end{align*}
            \begin{equation*}
                \begin{aligned}
                   &\Big\|\iint\big[\mathcal T_{\eta_k}\{\widehat \bSigma^{(k)}_{yy}\}(u,v)-\bSigma^{(k)}_{yy}(u,v)\big]\widehat\bW(v)\bSigma^{(k)}_{yy}(u,v)^{\T}\,{\rm d}u{\rm d}v\Big\| \\
                   \lesssim & \Big[\|\mathcal T_{\eta_n}\{\widehat \bSigma^{(k)}_{yy}\}-\bSigma^{(k)}_{yy}\|_{\cS,1}\|\mathcal T_{\eta_n}\{\widehat \bSigma^{(k)}_{yy}\}-\bSigma^{(k)}_{yy}\|_{\cS,\infty}\Big]^{1/2} \Big[\|\bSigma^{(k)}_{yy}\|_{\cS,1}\|\bSigma^{(k)}_{yy}\|_{\cS,\infty}\Big]^{1/2}\big\{1+o_\P(1)\big\}.
                \end{aligned}
            \end{equation*}
            For the second term above, it follows that
            \begin{equation*}
                \begin{aligned}
                    \|\bSigma^{(k)}_{yy}\|_{\cS,1} = \max\limits_{1\leq j\leq p}\sum_{i=1}^p \|\Sigma_{yy,ij}^{(k)}\|_{\cS}\leq \max\limits_{1\leq j\leq p}\sum_{i=1}^p \|\Sigma_{yy,ij}^{(k)}\|^\tau_{\cS}\|\Sigma_{yy,ij}^{(k)}\|^{1-\tau}_{\cS}\lesssim c_1(p),
                \end{aligned}
            \end{equation*}
            where $\tau \in (0,1)$ and $\|\Sigma_{yy,ij}^{(k)}\|_{\cS}$ is uniformly bounded under Condition~\ref{con:flp}. Following the similar procedure, we can show that $\|\bSigma^{(k)}_{yy}\|_{\cS,\infty}\lesssim c_1(p)+c_2(p)$. For the truly sparse case with $\tau=0$, it follows from the derivation in Section~\ref{lem:sparse-pf1} under Conditions~\ref{con:x}, \ref{con:eps}, \ref{con:colsparse} with fixed $r$ and $L$ that
            \begin{equation*}
                \begin{aligned}
                    \max\limits_{1\leq j\leq p}\sum_{i=1}^p\|\Sigma_{yy,ij}^{(k)}\|_{\cS} 
                    &\leq \max\limits_{1\leq j\leq p}\sum_{1 \leq m,m' \leq r} \sum_{i=1}^p |A_{im}A_{jm'}|\|\Sigma_{xx,mm'}^{(k)}\|_{\cS} + \max\limits_{1\leq j\leq p}\sum_{m=1}^r\sum_{i=1}^p|A_{im}|\|\Sigma^{(k)}_{x\epsilon,mj}\|_{\cS}\\
                    &\lesssim r^2L c_1(p) + rc_1(p)\asymp c_1(p).
                \end{aligned}
            \end{equation*}
            
            Similarly, under Condition~\ref{con:xeps2}, we can show that
            \begin{equation*}
                \begin{aligned}
                    \max\limits_{1\leq i \leq p}\sum_{j=1}^p\|\Sigma_{yy,ij}^{(k)}\|_{\cS} \lesssim r^2Lc_1(p) + rL c_2(p)\asymp c_1(p)+c_2(p).
                \end{aligned}
            \end{equation*}
            The above results together with (\ref{err.threshold}) from Step~1 implies (\ref{err.Mtilde}). 
            
{\bf Step~3}.
    Denote $\hat \epsilon = \mathcal D\big(\mathcal C(\bK), \mathcal C(\widetilde \bK)\big).$
    By the fact that $\hat \epsilon^2 = (2r)^{-1}\|\bK \bK^\T - \widetilde \bK \widetilde \bK^\T\|_\tF^2$ and Corollary 4.1 in \textcolor{blue}{Vu and Lei (2013)}, we obtain that
    \begin{equation*}
        \begin{aligned}
             2 r \nu_r \hat \epsilon^2 &\leq |\langle \widetilde\bM-\widecheck\bM, \widetilde\bK\widetilde\bK^\T-\bK\bK^\T \rangle|\\
             &\leq \left|\Big\langle \widetilde\bM-\widecheck\bM, \frac{\widetilde\bK\widetilde\bK^\T -\bK\bK^\T}{\|\widetilde\bK\widetilde\bK^\T -\bK\bK^\T\|_\tF}\Big\rangle \right|\cdot \sqrt{2r}\hat \epsilon, 
        \end{aligned}
    \end{equation*}
which further implies that
$$
             \nu_r \hat \epsilon \leq (2r)^{-1/2} \Big|\big\langle \widetilde\bM-\widecheck\bM, \boldsymbol {\widehat \Delta}\big\rangle \Big|,  
$$
    where $\boldsymbol{\widehat \Delta} = \frac{\widetilde\bK\widetilde\bK^\T -\bK\bK^\T}{\|\widetilde\bK\widetilde\bK^\T -\bK\bK^\T\|_\tF}$ and $\|\boldsymbol{ \widehat \Delta}\|_\tF = 1.$
    
    Note that $\boldsymbol{\widehat \Delta}$ is a matrix with rank at most $2 r$ and the associated singular values $\hat \sigma_1 \ge \hat \sigma_2 \ge \cdots \ge \hat \sigma_{2r}\ge 0.$ Then by the von-Neumann's trace inequality, we obtain that 
    \begin{align*}
             \nu_r \hat \epsilon & \leq (2r)^{-1/2} \big\|\widetilde\bM-\widecheck\bM\big\| \sum_{i=1}^{2r} \hat \sigma_i \\
             & \leq  (2r)^{-1/2} \big\|\widetilde\bM-\widecheck\bM\big\| \sqrt{2r \sum_{i=1}^{2r} \hat \sigma_i^2} 
              = \big\|\widetilde\bM-\widecheck\bM\big\| \|\boldsymbol{ \widehat \Delta}\|_\tF = \big\|\widetilde\bM-\widecheck\bM\big\|.
    \end{align*}
By (\ref{err.Mtilde}) from Step~2, we complete the proof of Theorem~\ref{thm:hdeigspace}.\hfill$\square$ 

\section{Auxiliary lemmas}
\label{supp.lem}
In this section, we present some auxiliary lemmas that are used in the proofs of theoretical results in Sections~\ref{sec:theory} and \ref{sec:ultrahigh}. 

The first lemma specifies the functional sparsity patterns in $\bSigma_{yy}^{(k)}$'s.

\begin{lemma}\label{lem:sparse}
    Under Conditions~\ref{con:x}, \ref{con:eps}, \ref{con:xeps2} and \ref{con:colsparse}, it holds that, for $k=1, \dots, k_0,$
    \begin{equation*}
    \begin{aligned}
        \max\limits_{1\leq j\leq p}\sum_{i=1}^p\|\Sigma^{(k)}_{yy,ij}\|^\tau_{\cS}\lesssim c_1(p)~~\text{and}~~ \max\limits_{1\leq i\leq p}\sum_{j=1}^p\|\Sigma^{(k)}_{yy,ij}\|^\tau_{\cS} \lesssim c_1(p)+c_2(p).
    \end{aligned}
    \end{equation*}
\end{lemma}

\begin{proof}[Proof of Lemma~\ref{lem:sparse}]
Denote by $\bA = (A_{ij})_{p \times r}$. By the decomposition (\ref{eq:cov-relation}) and the fact $(x+y)^\tau \leq x^{\tau}+y^\tau$ for $x,y \geq 0$ and $\tau \in [0,1),$ we have that
    \begin{equation*}\label{lem:sparse-pf1}
        \begin{aligned}
            \big\|\Sigma^{(k)}_{yy,ij}\big\|^{\tau}_{\cS}&\leq \Big\|\sum_{1\leq m,m'\leq r} A_{im}\Sigma^{(k)}_{xx,mm'}A_{jm'}\Big\|_{\cS}^{\tau} + \Big\|\sum_{m=1}^r A_{im}\Sigma^{(k)}_{x\epsilon,mj}\Big \|_{\cS}^{\tau}\\
            &\leq \sum_{1\leq m,m'\leq r}\big\|A_{im}\Sigma^{(k)}_{xx,mm'}A_{jm'}\big\|_{\cS}^{\tau} + \sum_{m=1}^r\big\|A_{im}\Sigma^{(k)}_{x\epsilon,mj}\big\|_{\cS}^{\tau}\\
            &= \sum_{1\leq m,m'\leq r}\big|A_{im}A_{jm'}\big|^{\tau}\big\|\Sigma^{(k)}_{xx,mm'}\big\|_{\cS}^{\tau} + \sum_{m=1}^r\big|A_{im}\big|^{\tau}\big\|\Sigma^{(k)}_{x\epsilon,mj}\big\|_{\cS}^{\tau},
        \end{aligned}
    \end{equation*}
    where, under Conditions~\ref{con:x} and \ref{con:eps}, it holds that
    \begin{equation}
    \label{bd.moment}
    \begin{split}
    \big\|\Sigma^{(k)}_{xx,mm'}\big\|_{\cS}^2 =& \iint [E \{X_{(t+k)m}(u)X_{t,m'}(v)\}]^2\,{\rm d}u{\rm d}v \\
    \leq& E\{\|X_{(t+k)m}\|_{2}^2\}\, E(\|X_{tm'}\|_{2}^2) =O(1),\\
    \max\limits_{1\leq j\leq p}\big\|\Sigma^{(k)}_{x\epsilon,mj}\big\|_{\cS}^2 \leq& E\{\|X_{(t+k)m}\|_{2}^2\}\max\limits_{1\leq j\leq p}E(\|\epsilon_{tj}\|_{2}^2) =O(1).
    \end{split}
    \end{equation}
    
    For each $j$, summing over $i$ and applying (\ref{bd.moment})  under Condition~\ref{con:colsparse} with fixed $L$ and $r,$ we obtain that
    \begin{equation*}
        \begin{aligned}
            \max\limits_{1\leq j\leq p}\sum_{i=1}^p\big\|\Sigma^{(k)}_{yy,ij}\big\|^{\tau}_{\cS}
            &\leq \max\limits_{1\leq j\leq p}\sum_{1\leq m,m'\leq r}\sum_{i=1}^p\Big\{|A_{im}A_{jm'}|^{\tau}\|\Sigma^{(k)}_{xx,mm'}\|_{\cS}^{\tau}\Big\} \\
            &~~~~+ \max\limits_{1\leq j\leq p}\sum_{m=1}^r\sum_{i=1}^p\Big\{|A_{im}|^{\tau}\|\Sigma^{(k)}_{x\epsilon,mj}\|_{\cS}^{\tau}\Big\}\\
            &\leq c_1(p)\Big\{r^2\cdot\max\limits_{\substack{1\leq j\leq p,\\ 1\leq m,m' \leq r}}|A_{jm'}|^{\tau}\|\Sigma^{(k)}_{xx,mm'}\|_{\cS}^{\tau} + r\cdot\max\limits_{\substack{1\leq j\leq p,\\ 1\leq m\leq r}} \|\Sigma^{(k)}_{x\epsilon,mj}\|_{\cS}^{\tau}\Big\} \lesssim c_1(p).
        \end{aligned}
    \end{equation*}
    By the similar arguments above and Condition~\ref{con:xeps2},
        \begin{equation*}
        \begin{aligned}
            \max\limits_{1\leq i\leq p}\sum_{j=1}^p\big\|\Sigma^{(k)}_{yy,ij}\big\|^{\tau}_{\cS}
            &\leq \max\limits_{1\leq i\leq p}\sum_{1\leq m,m'\leq r}\sum_{j=1}^p\Big\{|A_{im}A_{jm'}|^{\tau}\|\Sigma^{(k)}_{xx,mm'}\|_{\cS}^{\tau}\Big\}\\
            &~~~~+ \max\limits_{1\leq i\leq p}\sum_{m=1}^r\sum_{j=1}^p\Big\{|A_{im}|^{\tau}\|\Sigma^{(k)}_{x\epsilon,mj}\|_{\cS}^{\tau}\Big\}\\
            &\leq c_1(p)r^2 \cdot\max_{\substack{1\leq i\leq p,\\ 1\leq m,m'\leq r}}|A_{im}|^{\tau}\|\Sigma^{(k)}_{xx,mm'}\|_{\cS}^{\tau} + c_2(p)r\cdot \max\limits_{\substack{1\leq i\leq p,\\ 1\leq m\leq r}} |A_{im}|^\tau\\
            &\lesssim c_1(p) + c_2(p).
        \end{aligned}
    \end{equation*}
Hence we complete the proof of this lemma. 
\end{proof}

\begin{lemma}\label{lem:intnorm}
    For a measurable function $f: \mathcal D \rightarrow \mathcal B$, where $\mathcal D\in \mathbb R^p$ is a compact set and $\mathcal B$ is a separable Banach space with norm $\|\cdot\|$, if the norm of $f$ has finite integral, then $\|\sint_{\mathcal D}f(x)\mathrm d x\|\leq \sint_{\mathcal D}\|f(x)\|\mathrm d x$.
\end{lemma}
\begin{proof}
    Let $\nu=\int_{\cD}f(x){\rm d}x\in\cB,\cB^*$ is the dual space of $\cB$ and $\Vert\cdot\Vert_*$ is the dual norm. By the Hahn-Banach theorem, there exists $\phi\in\cB^*$ such that $\Vert\phi\Vert_*=1$ and $\Vert\phi(\nu)\Vert_*=\Vert\nu\Vert=\big\Vert\int_{\cD}f(x){\rm d}x\big\Vert.$ Thus, we have
    $$
    \left\Vert\int_{\cD}f(x){\rm d}x\right\Vert=\left\Vert\phi\Big(\int_{\cD}f(x){\rm d}x\Big)\right\Vert_*=\left\Vert\int_{\cD}\phi\{f(x)\}{\rm d}x\right\Vert_*\le\int_{\cD}\Vert\phi\{f(x)\}\Vert_*{\rm d}x\le \int_{\cD}\Vert f(x)\Vert{\rm d}x,
    $$
    where the second equality follows from that the dual space $\cB^*$ of $\cB$ is a Banach space (also a linear space) and the last inequality follows from $\Vert\phi(f)\Vert_*\le\Vert\phi\Vert_*\cdot\Vert f\Vert=\Vert f\Vert$.
\end{proof}

\color{black}

The third lemma specifies the order of $\bSigma^{(k)}_{xx}(u,v)$.
\begin{lemma}\label{lem:xeps}
    Suppose that Condition~\ref{con:x} holds, then for $ k=1,\dots,k_0,$ $\iint\|\bSigma^{(k)}_{xx}(u,v)\|^2 \, {\rm d}u{\rm d}v=O(1)$ and $\|\iint \bSigma^{(k)}_{xx}(u,v)\, {\rm d}u{\rm d}v\|=O(1).$
\end{lemma}
\begin{proof}[Proof of Lemma~\ref{lem:xeps}]
By Cauchy-Schwartz inequality, Condition~\ref{con:x} and Fubini Theorem,
\begin{equation*}\label{lem:xeps-pf1}
\begin{aligned}
    \iint\|\bSigma^{(k)}_{xx}(u,v)\|^2\,{\rm d}u{\rm d}v &\leq \iint\| \eE\{\bX_{t+k}(u) \bX_t(v)^{\T}\}\|_\tF^2 \,{\rm d}u{\rm d}v\\
    &= \sum_{i=1}^r \sum_{j=1}^r \iint \big[{\mathbb E}\{ X_{(t+k)i}(u)X_{tj}(v)\}\big]^2 \,{\rm d}u{\rm d}v\\
    &\leq \Big\{\sum_{i=1}^r {\mathbb E}\big(\| X_{(t+k)i}\|_{2}^2\big) \Big\}\Big\{\sum_{j=1}^r {\mathbb E}\big(\|X_{tj}\|_{2}^2\big)\Big\} = O(1).
\end{aligned}
\end{equation*}
On the other hand, by Lemma~\ref{lem:intnorm}, we have that
\[
    \left\|\iint \bSigma^{(k)}_{xx}(u,v)\,{\rm d}u{\rm d}v\right\|^2\leq \Big(\iint \big\|\bSigma^{(k)}_{xx}(u,v)\big\|\,{\rm d}u{\rm d}v\Big)^2\leq  \text{Vol}(\cU\times\cU)\iint\big\|\bSigma^{(k)}_{xx}(u,v)\big\|^2 \,{\rm d}u{\rm d}v = O(1),
\] 
where $\text{Vol}(\cdot)$ represents the volume.
\end{proof}

The fourth lemma provides the non-asymptotic error bound on $\widehat\Sigma_{yy,ij}^{(k)}(u,v).$ 

\begin{lemma}\label{lem:nonasym}
Suppose that Conditions~\ref{con:flp} and \ref{con:fsm} hold. Then there exists some positive constants $c_1$ and $c_2$ such that for $k=1, \dots, k_0$ and $t>0,$ 
\begin{equation*}
    \label{max.rate}
    \mathbb P\left(\max\limits_{1\leq i,j\leq p} \big\|\widehat{\Sigma}_{yy,ij}^{(k)}-\Sigma_{yy,ij}^{(k)}\big\|_{\cS}>\mathcal M_{y} t \right)\leq c_1p^2\exp\big\{-c_2n\min(t,t^2)\big\},
\end{equation*}
If $n\geq \rho^2\log p$ with $\rho>\sqrt 2c_2^{-1/2},$ then 
\begin{equation*}
    \label{max.rate-eq}
    \max\limits_{1\leq i,j\leq p} \big\|\widehat{\Sigma}_{yy,ij}^{(k)}-\Sigma_{yy,ij}^{(k)}\big\|_{\cS}\leq \rho \mathcal M_{y}\sqrt{\frac{\log p}{n}}
\end{equation*}
holds with probability greater than $1-c_1p^{2-c_2\rho^2}$.
\end{lemma}
\begin{proof}[Proof of Lemma~\ref{lem:nonasym}]
    This lemma follows directly from \textcolor{blue}{Guo and Qiao (2023)}; \textcolor{blue}{Fang et al. (2022)} and the choice of $t=\rho (\log p/n)^{1/2}\leq 1,$ and hence the proof is omitted here. 
\end{proof}

The lemma below generalizes the following inequality between matrix norms
\begin{equation}
\label{mat.ineq}
\|\bE\|^2\leq \|\bE\|_{\infty}\|\bE\|_1~~\text{for any}~~\bE \in {\mathbb R}^{p \times p} 
\end{equation}
to the functional domain and is used in the proof of Theorem~\ref{thm:hdeigspace}. 

\begin{lemma}\label{lem:int-mat}
For $\bS(u,v)=\{S_{ij}(u,v)\}_{p \times p}$ and $\bT(u,v)=\{T_{ij}(u,v)\}_{p \times p}$ with each $S_{ij}$ and $T_{ij} \in L_2(\cU\times\cU),$ it holds that
        $$\Big\|\iint  \bS(u,v)\bT^{\T}(u,v)\,{\rm d}u{\rm d}v\Big\|\leq \Big(\|\bS\|_{\cS,\infty}  \|\bS\|_{\cS,1}\Big)^{1/2}\Big(\|\bT\|_{\cS,\infty}  \|\bT\|_{\cS,1}\Big)^{1/2}.$$
\end{lemma}
\begin{proof}[Proof of Lemma \ref{lem:int-mat}]
        Notice that
        \begin{equation}\label{lem:int-mat-pf1}
            \begin{aligned}
                \Big\|\iint \bS(u,v)\bT(u,v)^{\T}\,{\rm d}u{\rm d}v\Big\|_1&= \max\limits_{1\leq j\leq p}\sum_{i=1}^p\left|\iint \sum_{k=1}^p S_{ik}(u,v)T_{jk}(u,v) \, {\rm d}u {\rm d}v \right|\\
                &\leq \max\limits_{1\leq j\leq p}\sum_{i=1}^p\sum_{k=1}^p\|S_{ik}\|_{\cS}\|T_{jk}\|_{\cS}\\ 
                &\leq \Big\{\max\limits_{1\leq j\leq p}\sum_{k=1}^p\|T_{jk}\|_{\cS}\Big\} \Big\{\max\limits_{1\leq k\leq p}\sum_{i=1}^p\|S_{ik}\|_{\cS}\Big\} = \|\bS\|_{\cS,1}  \|\bT\|_{\cS,\infty}.
            \end{aligned}
        \end{equation}
        By the similar arguments, we obtain that
        \begin{equation}\label{lem:int-mat-pf2}
            \Big\|\iint \bS(u,v)\bT(u,v)^{\T} \,{\rm d}u{\rm d}v\Big\|_{\infty}\leq \|\bS\|_{\cS,\infty}  \|\bT\|_{\cS,1}.
        \end{equation}
        By setting $\bE=\iint \bS(u,v)\bT(u,v)^{\T}\,{\rm d}u{\rm d}v$, Lemma~\ref{lem:int-mat} follows immediately from (\ref{mat.ineq})--(\ref{lem:int-mat-pf2}).
\end{proof}

\section{Additional empirical results}
\label{supp.emp}

When the factors are weak with $\delta=0.5,$ Figure~\ref{fig3} plots average relative frequencies of $\hat r=r$ as the factor strength increases, supporting Section~\ref{sim:ordi}. Figure~\ref{fig4} plots the corresponding average estimation errors for $\mathcal C(\bA).$

\begin{figure}[!hbt]
\captionsetup[subfigure]{labelformat=empty}
\centering
\begin{subfigure}{0.32\linewidth}
  \caption{\scriptsize{$p=50$}}
  \vspace{-0.5cm}
  \includegraphics[width=5cm,height=4cm]{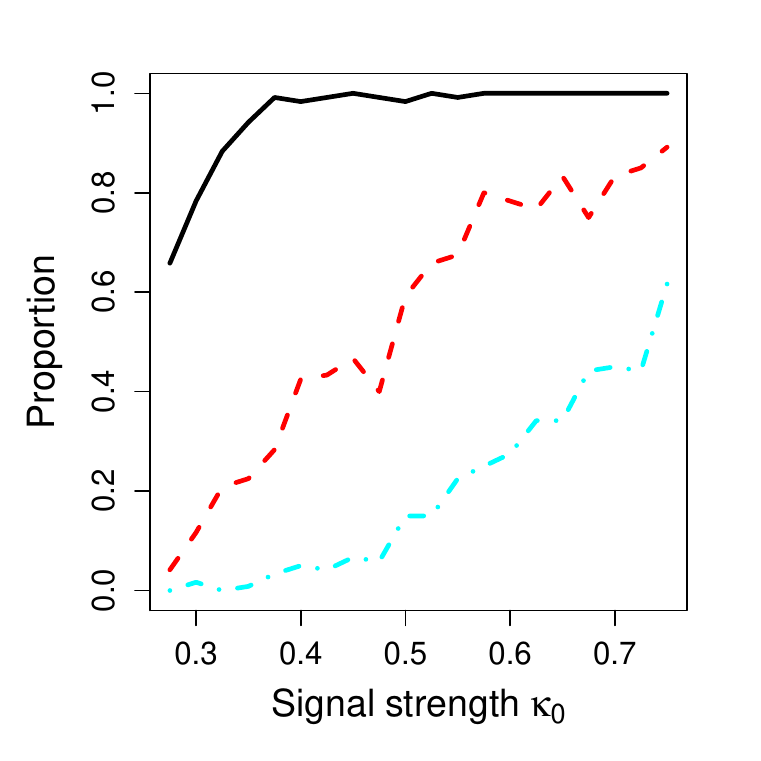}
  \label{fig:2-1-50}
  \vspace{-0.25cm}
  \caption{\scriptsize{$p=50$}}
  \vspace{-0.5cm}
  \includegraphics[width=5cm,height=4cm]{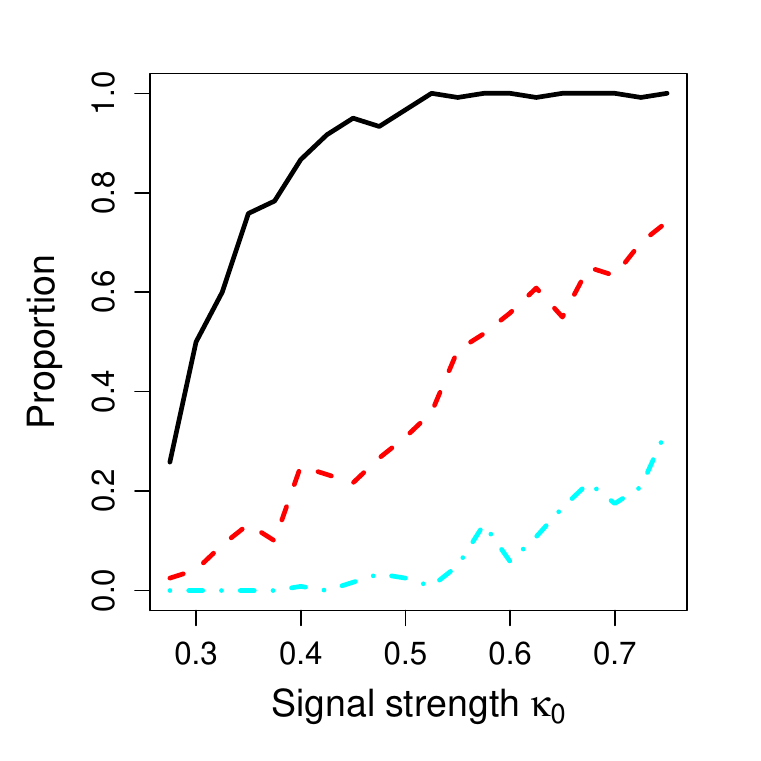}
  \label{fig:2-2-50}
  \vspace{-0.25cm}
  \caption{\scriptsize{$p=50$}}
  \vspace{-0.5cm}
  \includegraphics[width=5cm,height=4cm]{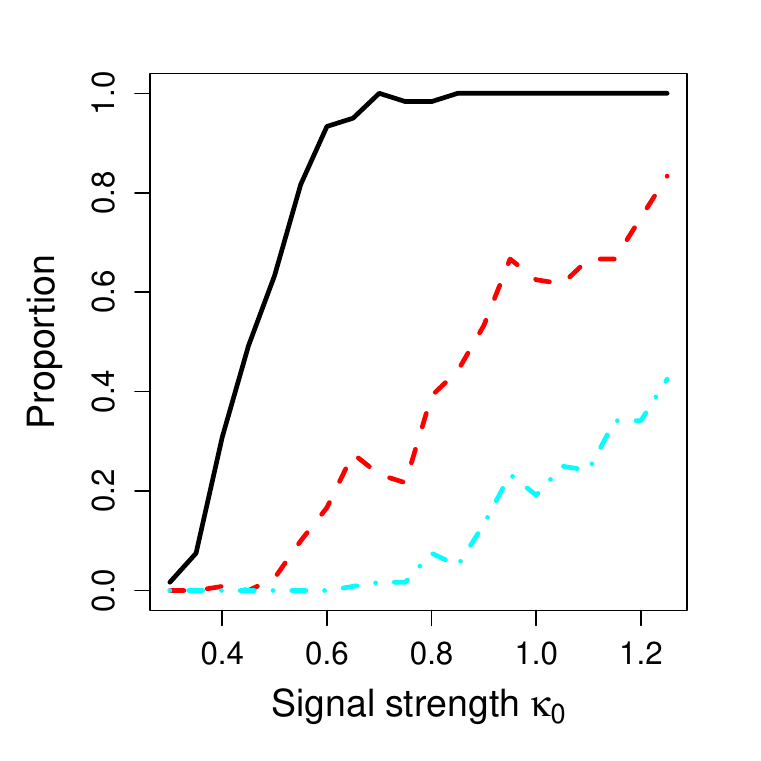}
  \label{fig:2-4-50}
\end{subfigure}
\begin{subfigure}{0.32\linewidth}
  \caption{\scriptsize{$p=100$}}
  \vspace{-0.5cm}
  \includegraphics[width=5cm,height=4cm]{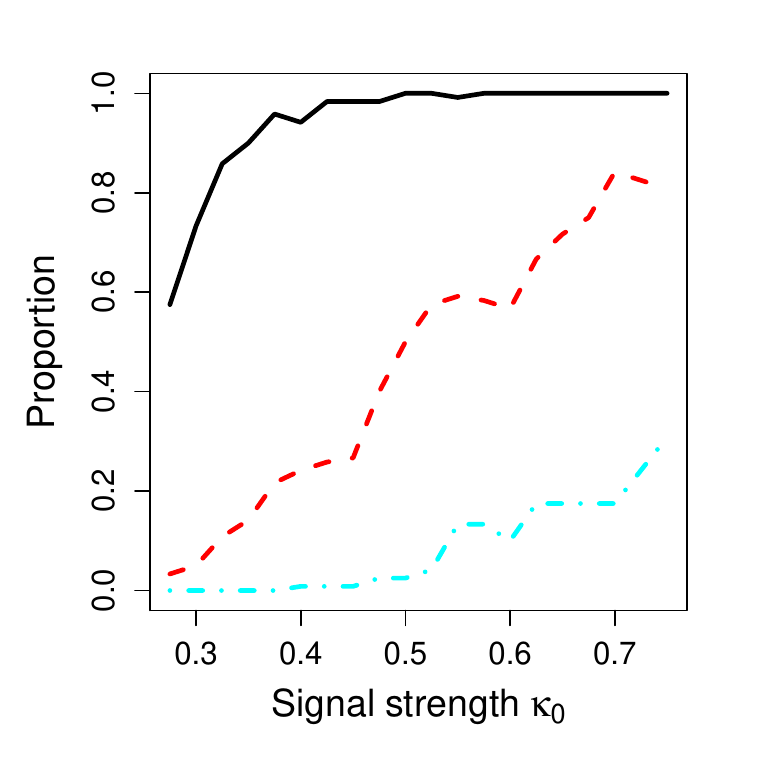}
  \label{fig:2-1-100}
  \vspace{-0.25cm}
  \caption{\scriptsize{$p=100$}} 
  \vspace{-0.5cm}
  \includegraphics[width=5cm,height=4cm]{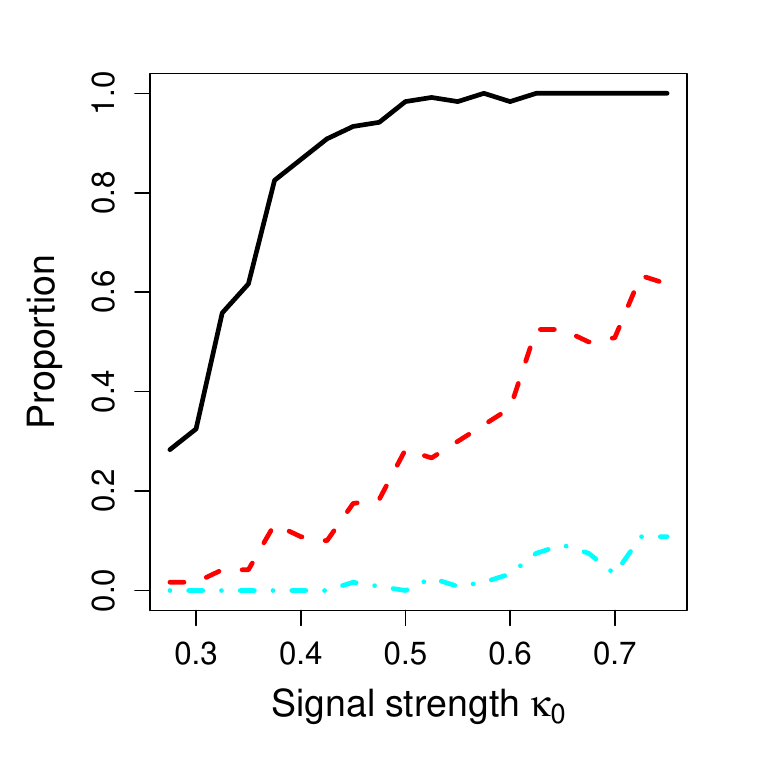}
  \label{fig:2-2-100}
  \vspace{-0.25cm}
  \caption{\scriptsize{$p=100$}} 
  \vspace{-0.5cm}
  \includegraphics[width=5cm,height=4cm]{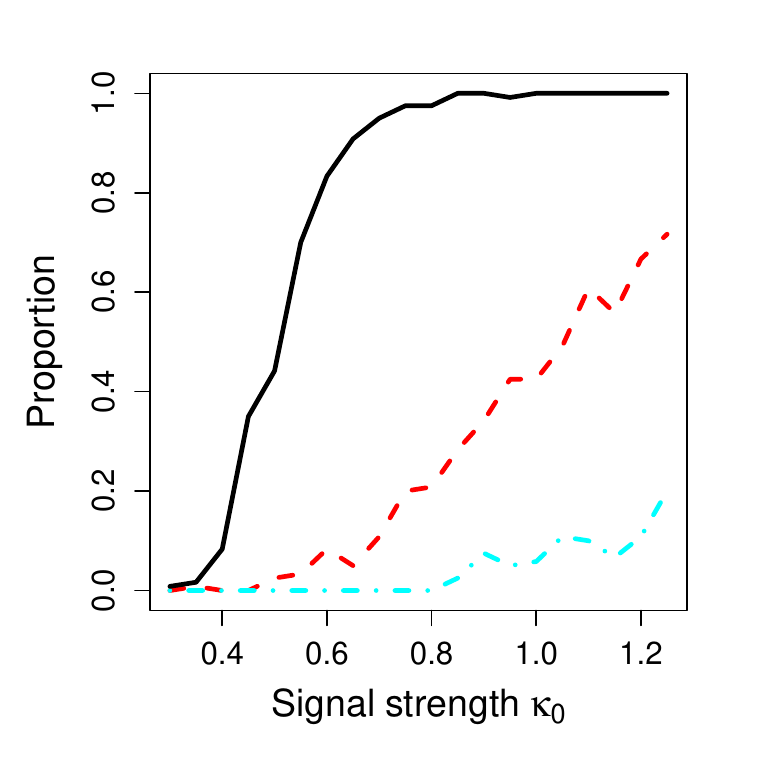}
  \label{fig:2-4-100}
\end{subfigure}
\begin{subfigure}{0.32\linewidth}
  \caption{\scriptsize{$p=200$}}
  \vspace{-0.5cm}
  \includegraphics[width=5cm,height=4cm]{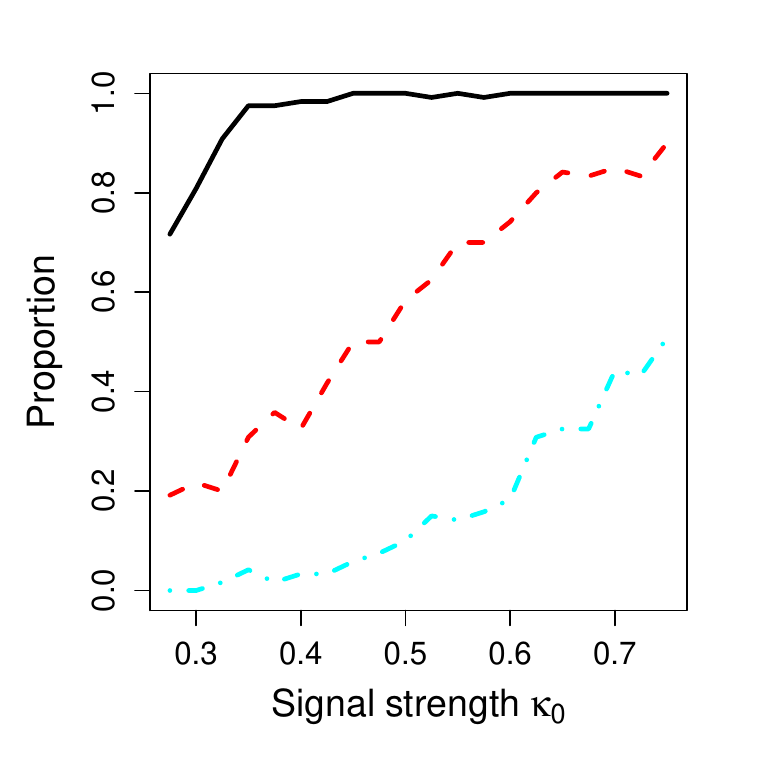}
  \label{fig:2-1-200}
  \vspace{-0.25cm}
  \caption{\scriptsize{$p=200$}}
  \vspace{-0.5cm}
  \includegraphics[width=5cm,height=4cm]{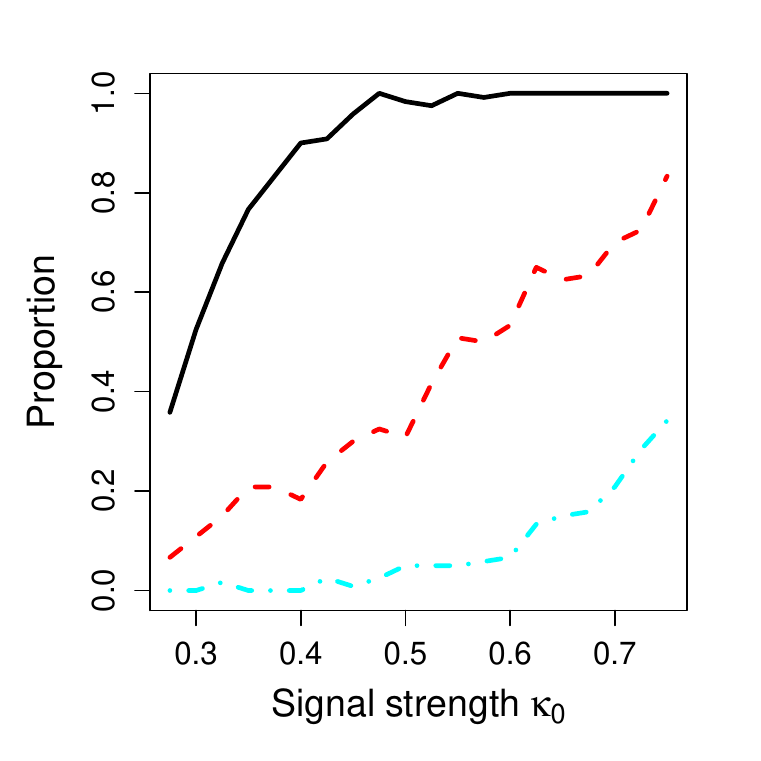}
  \label{fig:2-2-200}
  \vspace{-0.25cm}
  \caption{\scriptsize{$p=200$}}
  \vspace{-0.5cm}
  \includegraphics[width=5cm,height=4cm]{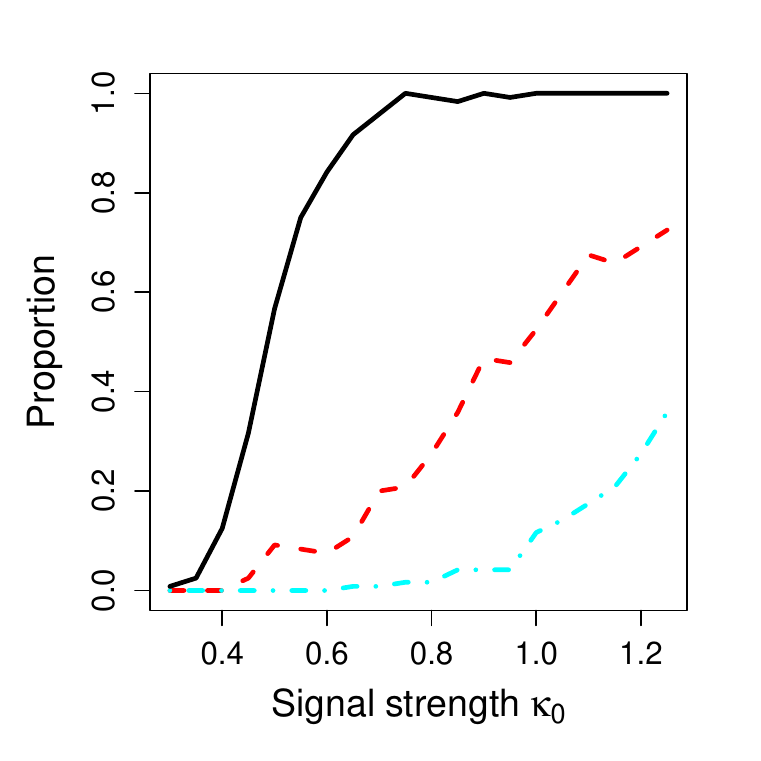}
  \label{fig:2-4-200}
\end{subfigure}
\centering
\caption{\label{fig3}{\it \small Scenario~1 (top row) Scenario~2 (middle row) and Scenario~4 (bottom row) for $p=50,$ $100$ and $200$ with correct $r$ when the factors are weak: 
Plots of relative frequency estimates for $\pr(\hat r = r)$ against $\kappa_0$ for three methods based on $\widehat\bM$ (black solid), $\widehat\bM_1$ (red dashed) and $\widehat\bM_2$ (cyan dash dotted).
}}
\end{figure}
\begin{figure}[!hbt]
\captionsetup[subfigure]{labelformat=empty}
\centering
\begin{subfigure}{0.32\linewidth}
  \caption{\scriptsize{$p=50$}}
  \vspace{-0.5cm}
  \includegraphics[width=5cm,height=4cm]{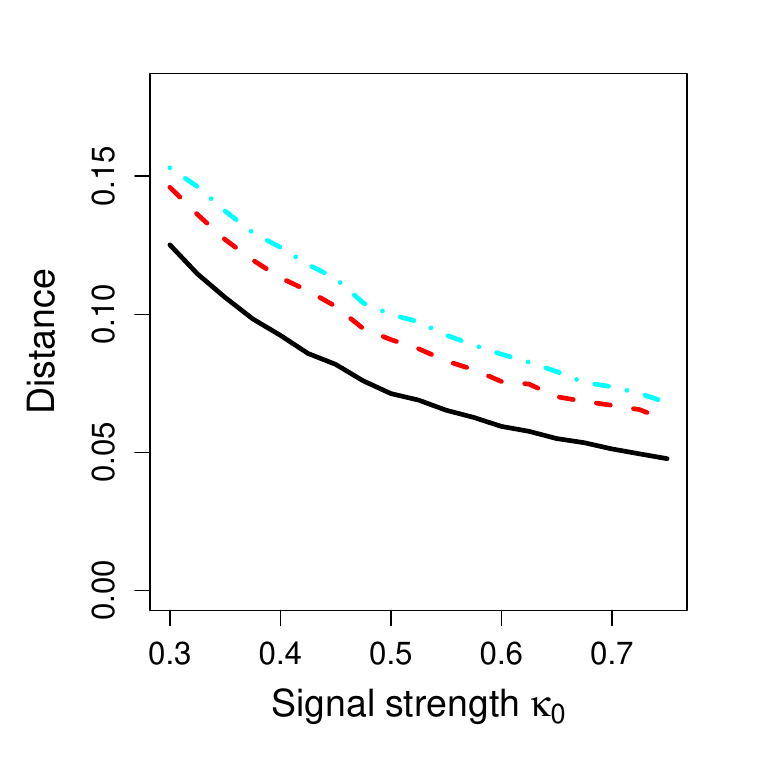}
  \label{fig:4-1-50}
  \vspace{-0.25cm}
  \caption{\scriptsize{$p=50$}}
  \vspace{-0.5cm}
  \includegraphics[width=5cm,height=4cm]{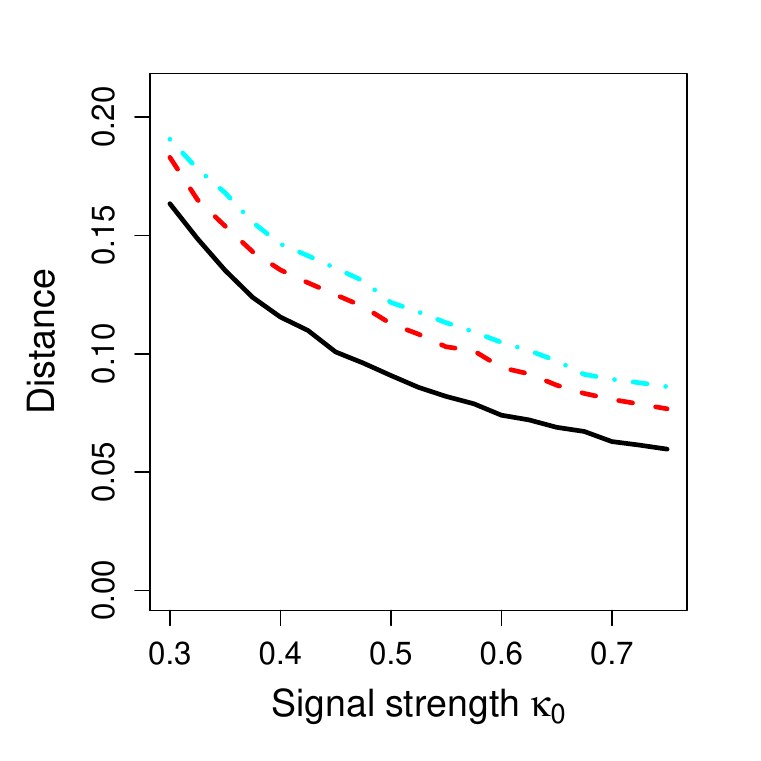}
  \label{fig:4-2-50}
  \vspace{-0.25cm}
  \caption{\scriptsize{$p=50$}}
  \vspace{-0.5cm}
  \includegraphics[width=5cm,height=4cm]{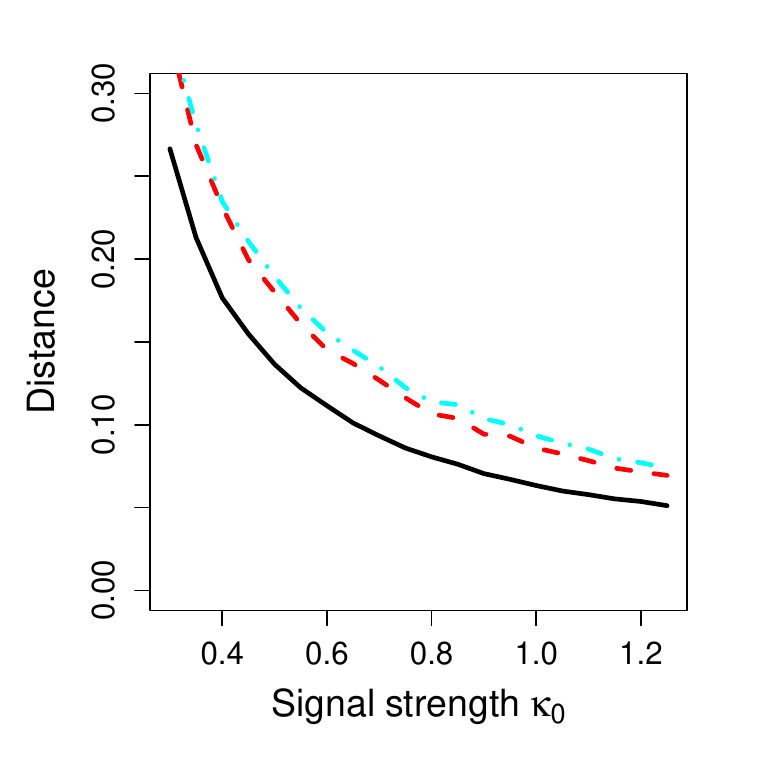}
  \label{fig:4-4-50}
\end{subfigure}
\begin{subfigure}{0.32\linewidth}
  \caption{\scriptsize{$p=100$}}
  \vspace{-0.5cm}
  \includegraphics[width=5cm,height=4cm]{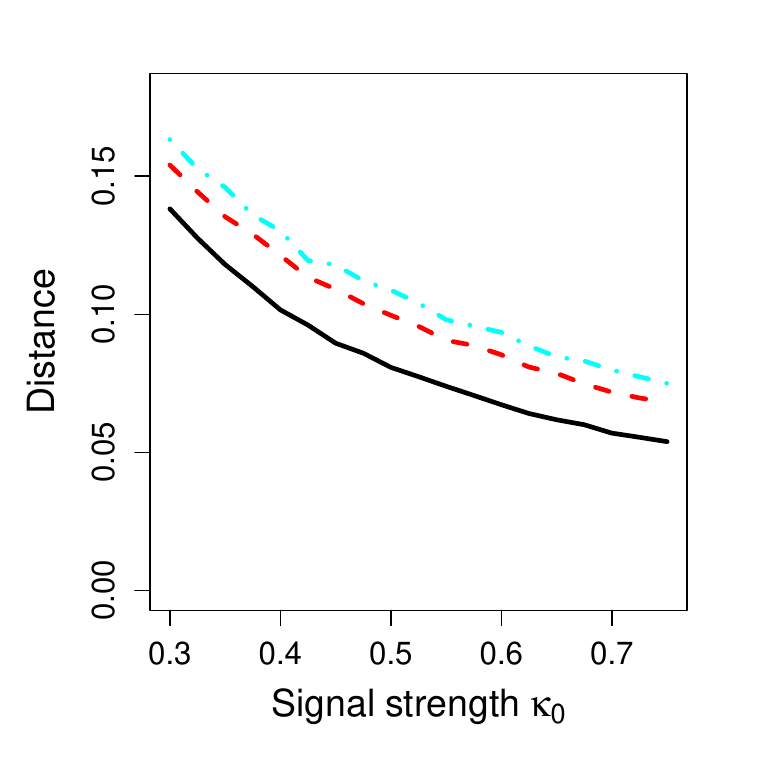}
  \label{fig:4-1-100}
  \vspace{-0.25cm}
  \caption{\scriptsize{$p=100$}} 
  \vspace{-0.5cm}
  \includegraphics[width=5cm,height=4cm]{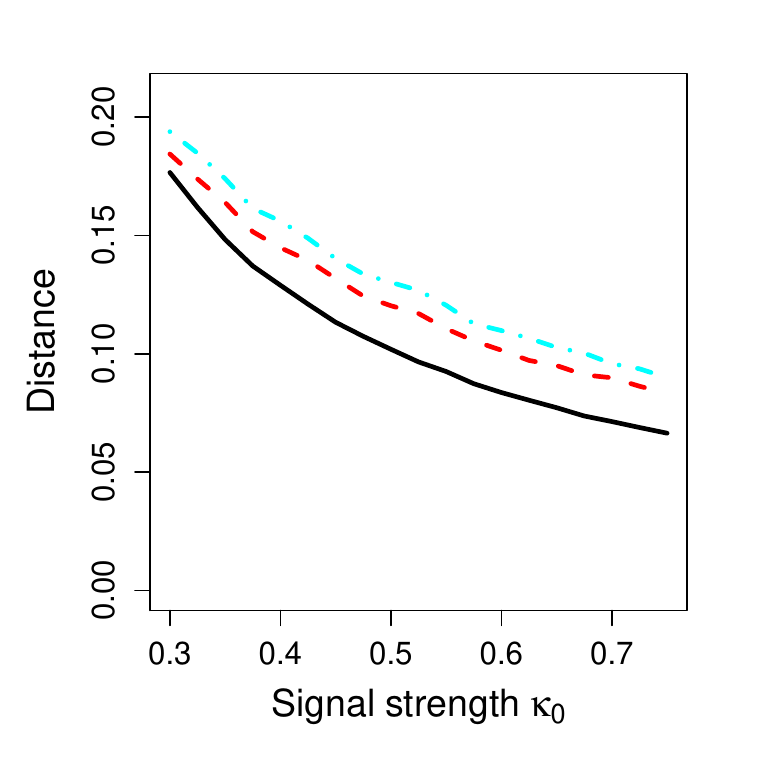}
  \label{fig:4-2-100}
  \vspace{-0.25cm}
  \caption{\scriptsize{$p=100$}} 
  \vspace{-0.5cm}
  \includegraphics[width=5cm,height=4cm]{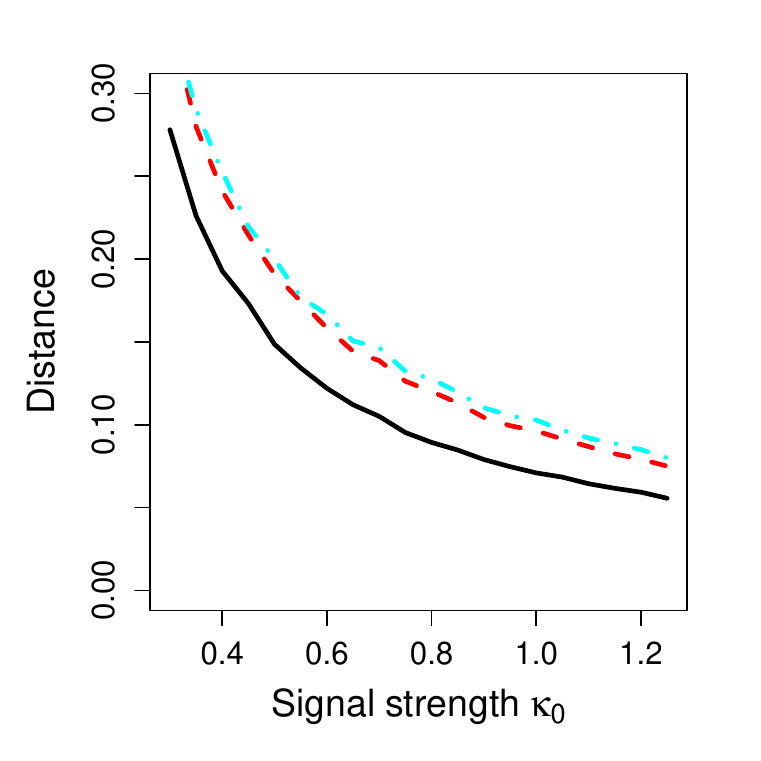}
  \label{fig:4-4-100}
\end{subfigure}
\begin{subfigure}{0.32\linewidth}
  \caption{\scriptsize{$p=200$}}
  \vspace{-0.5cm}
  \includegraphics[width=5cm,height=4cm]{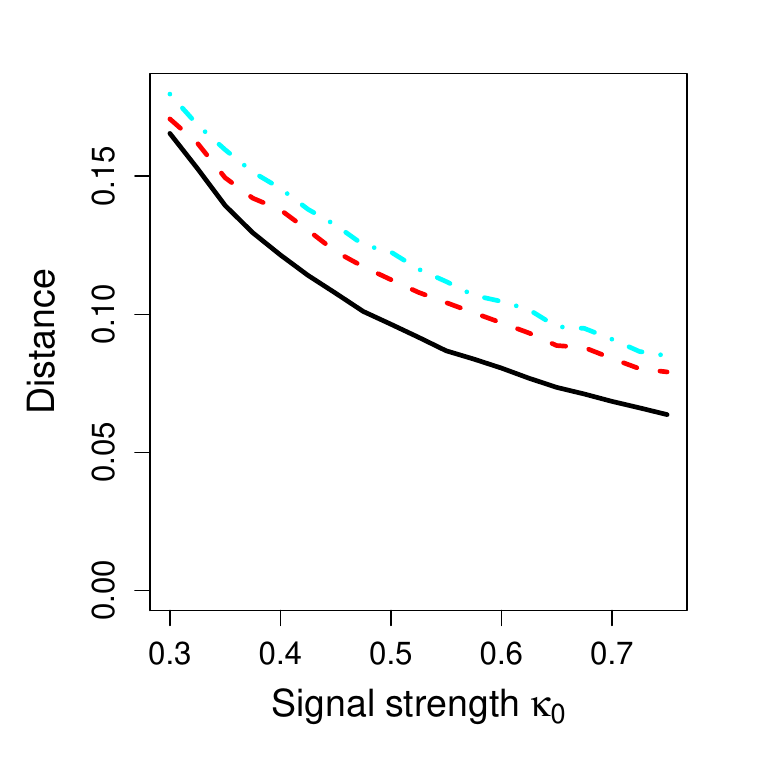}
  \label{fig:4-1-200}
  \vspace{-0.25cm}
  \caption{\scriptsize{$p=200$}}
  \vspace{-0.5cm}
  \includegraphics[width=5cm,height=4cm]{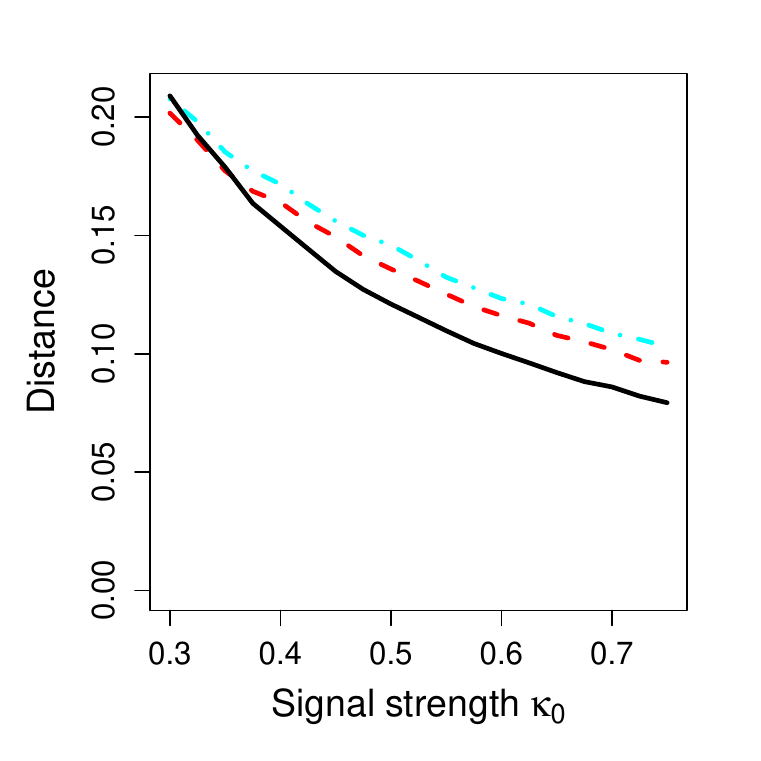}
  \label{fig:4-2-200}
  \vspace{-0.25cm}
  \caption{\scriptsize{$p=200$}}
  \vspace{-0.5cm}
  \includegraphics[width=5cm,height=4cm]{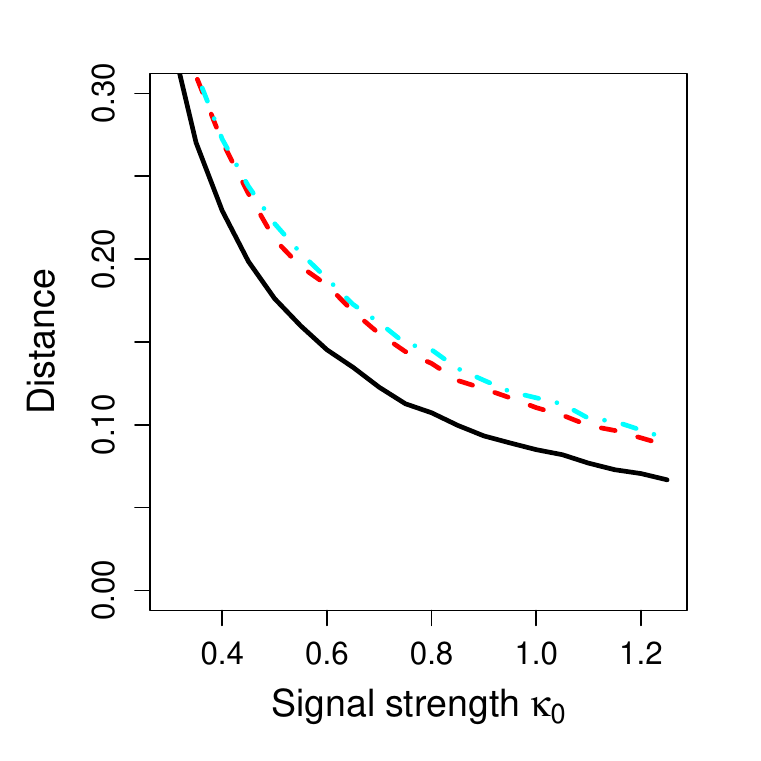}
  \label{fig:4-4-200}
\end{subfigure}
\centering
\caption{\label{fig4}{\it \small Scenario~1 (top row), Scenario~2 (middle row) and Scenario~4 (bottom row) for $p=50,$ $100$ and $200$ with correct $r$ when the factors are weak: Plots of average $\mathcal D \big(\mathcal C(\bA), \mathcal C(\widehat \bA)\big)$ against $\kappa_0$ for three methods based on $\widehat\bM$ (black solid), $\widehat\bM_1$ (red dashed) and $\widehat\bM_2$ (cyan dash dotted).
}}
\end{figure}

As suggested by one referee, we consider the following Scenario~5, where the idiosyncratic components are cross-correlated.
\begin{itemize}
    \item[] \textsc{Scenario} 5. We replace $\bvarepsilon_t(\cdot)$ in \textsc{Scenario}~4 by $\mathbf{C}\bvarepsilon_t(\cdot)$, where $\mathbf{C}=(0.5^{|i-j|})_{1\leq i,j\leq p}.$ 
\end{itemize}
Figure~\ref{fig-cross} plots average relative frequencies of $\hat r = r$ as the factor strength increases, and the corresponding average estimation errors for ${\cal C}(\bA)$ 
under Scenario~5 with $p=100$ when factors are strong. We observe in Figure~\ref{fig-cross} that our method still significantly outperform the two competitors.
\begin{figure}[!hbt]
\captionsetup[subfigure]{labelformat=empty}
\centering
\begin{subfigure}{0.45\linewidth}
  \includegraphics[width=7cm,height=5cm]{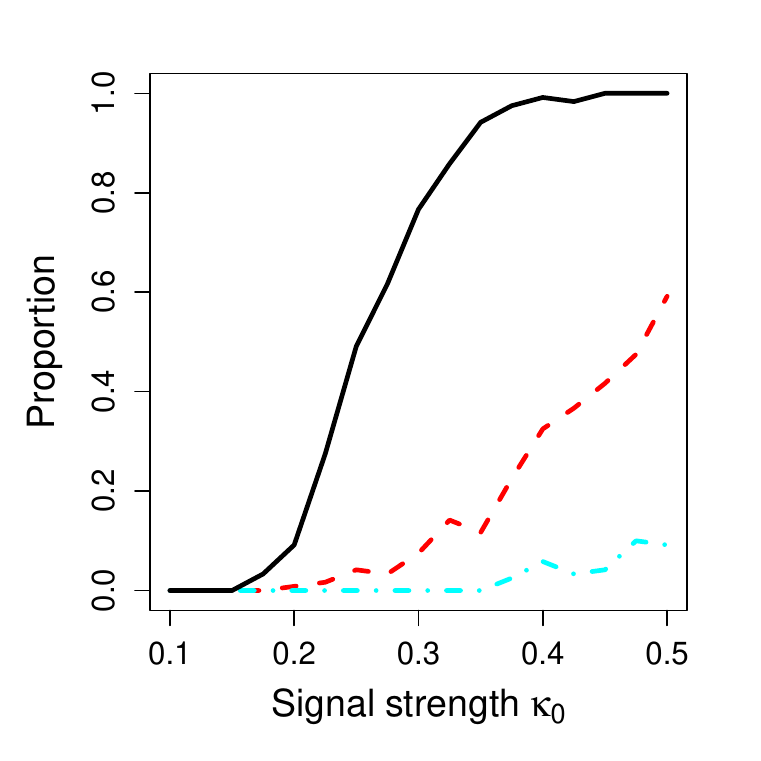}
  \label{fig:5-5-100(1)}
\end{subfigure}
\begin{subfigure}{0.45\linewidth}
  \includegraphics[width=7cm,height=5cm]{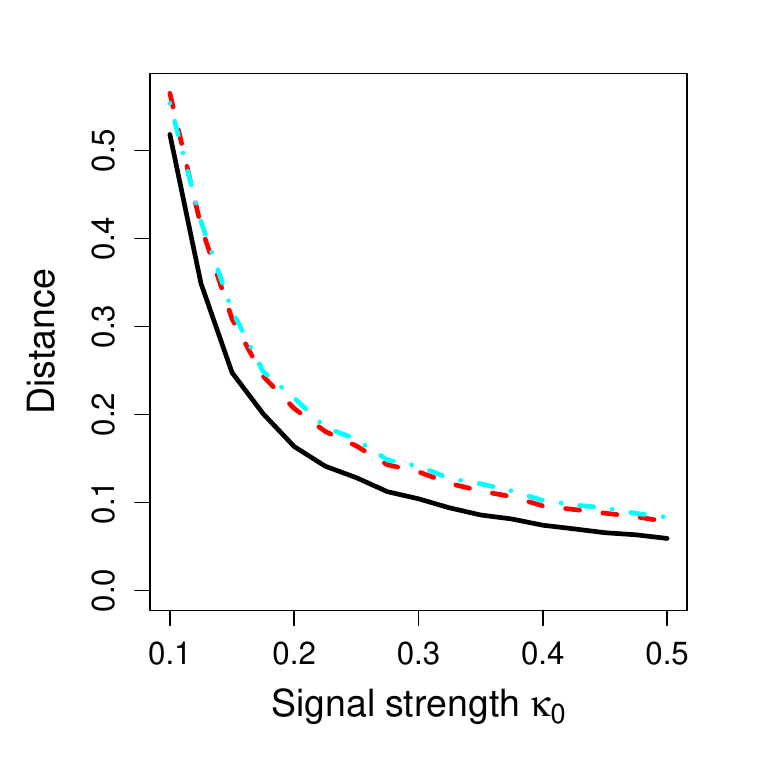}
  \label{fig:5-5-100(2)}
\end{subfigure}
\centering
\caption{\label{fig-cross}{\it \small Scenario~5 for $p=100$ when factors are strong: Plots of average relative frequency estimates for $\pr(\hat r = r)$ (Left) and average $\mathcal D \big(\mathcal C(\bK), \mathcal C(\widehat \bK)\big)$ (Right) against $\kappa_0$ for three methods based on $\widehat\bM$ (black solid), $\widehat\bM_1$ (red dashed) and $\widehat\bM_2$ (cyan dash dotted).
}}
\vspace{-0.2cm}
\end{figure}

To conduct a sensitivity analysis of our method with respect to $q$ and $k_0$, we repeat the simulated experiment under Scenario~1 with $p=100$ by setting $q=8,12,16,20$ and $k_0=2,3,4,5.$ We observe in Figure~\ref{fig-rob} that the performance is generally insensitive to the different values of $q$ and $k_0$ that vary within reasonable ranges.

\begin{figure}[!hbt]
\captionsetup[subfigure]{labelformat=empty}
\centering
\begin{subfigure}{0.45\linewidth}
  \vspace{-0.5cm}
  \includegraphics[width=7cm,height=5cm]{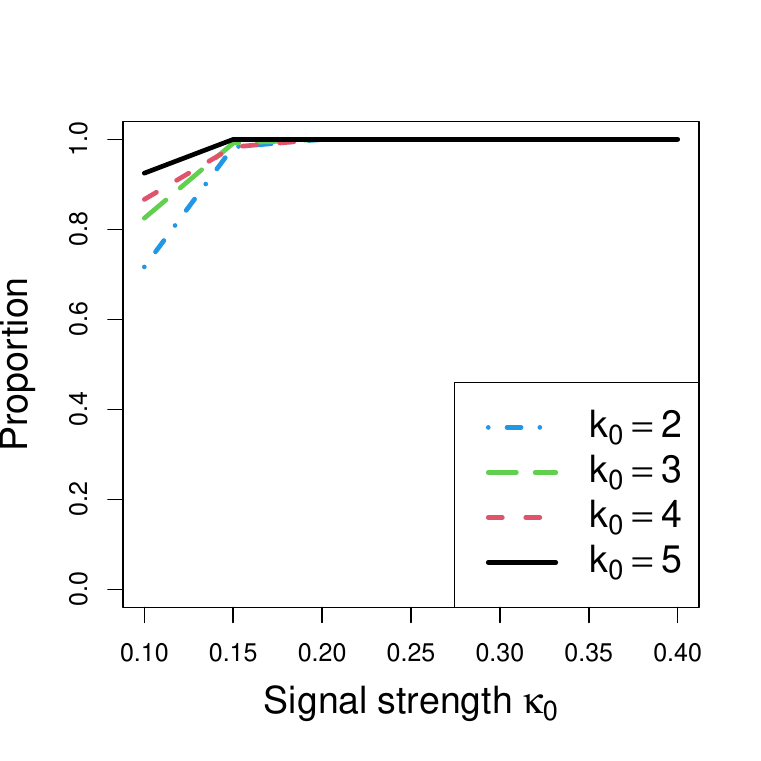}
  \label{fig:6-1-100(1)}
\end{subfigure}
\begin{subfigure}{0.45\linewidth}
  \vspace{-0.5cm}
  \includegraphics[width=7cm,height=5cm]{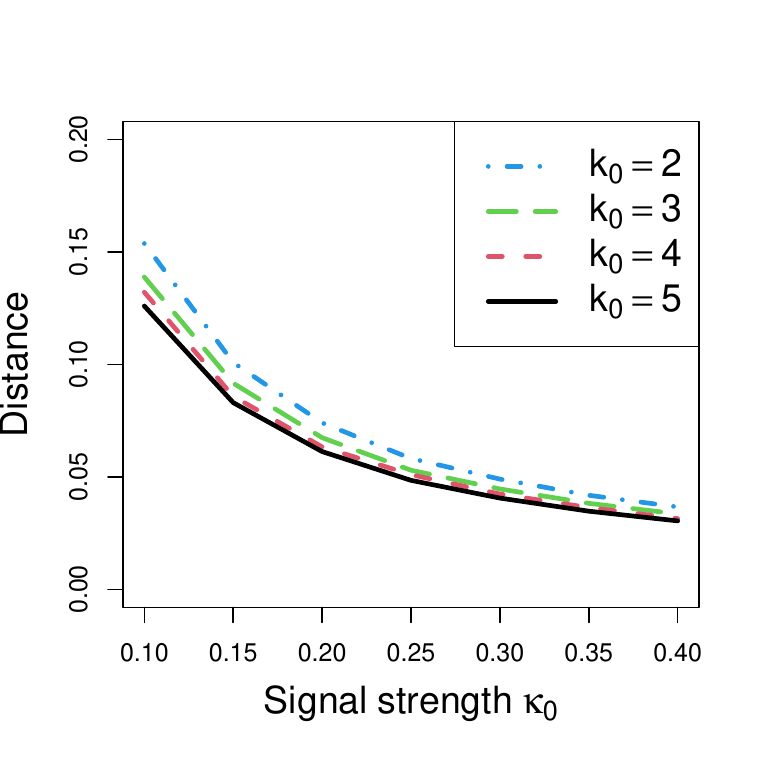}
  \label{fig:6-1-100(2)}
\end{subfigure}\\
\begin{subfigure}{0.45\linewidth}
  \vspace{-0.5cm}
  \includegraphics[width=7cm,height=5cm]{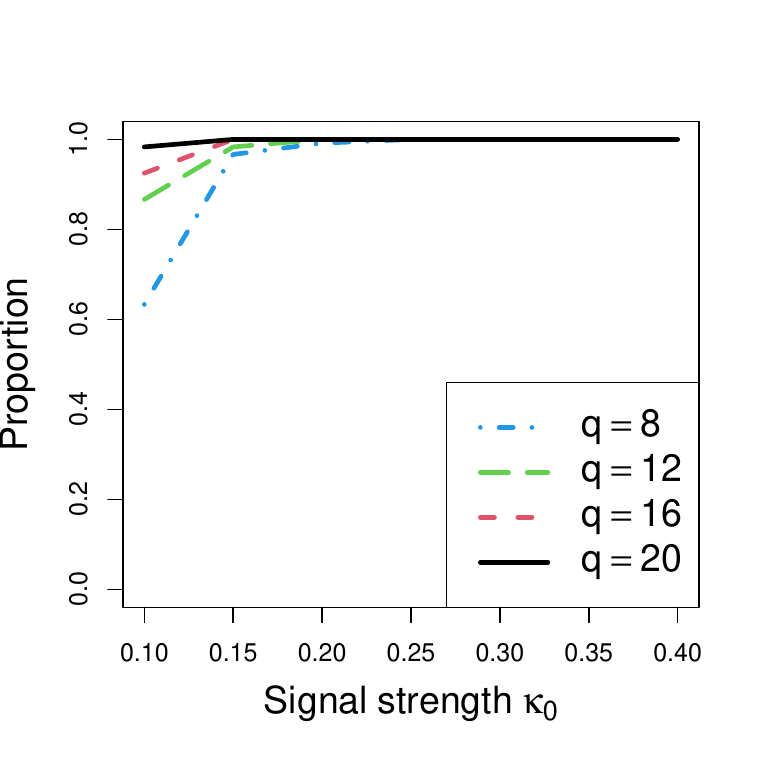}
  \label{fig:6-2-100(1)}
\end{subfigure}
\begin{subfigure}{0.45\linewidth}
  \vspace{-0.5cm}
  \includegraphics[width=7cm,height=5cm]{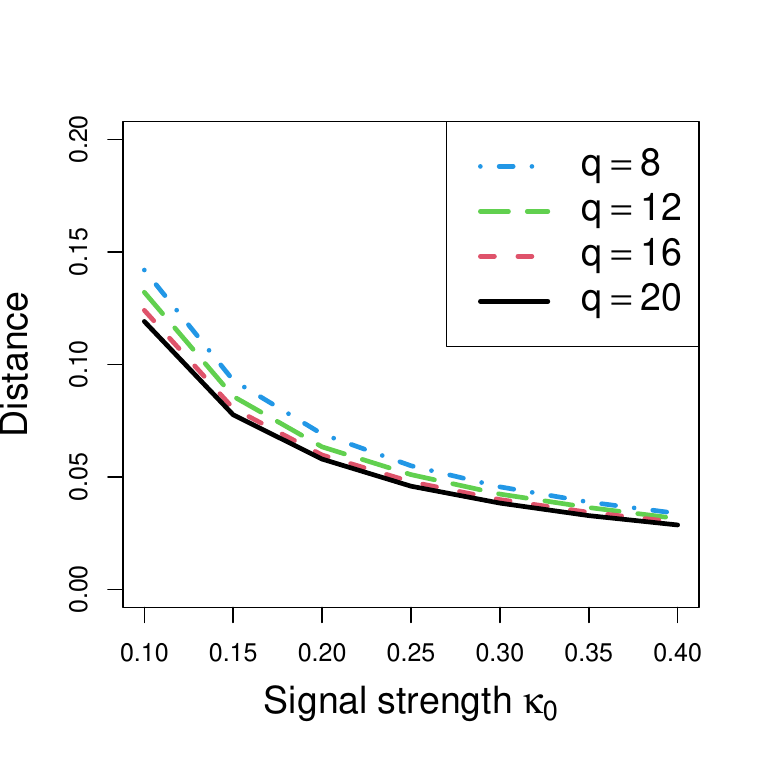}
  \label{fig:6-2-100(2)}
\end{subfigure}
\centering
\caption{\label{fig-rob}{\it \small Scenario~1 for $p=100$ when factors are strong: Plots of average relative frequency estimates for $\pr(\hat r = r)$ (Left) and average $\mathcal D \big(\mathcal C(\bK), \mathcal C(\widehat \bK)\big)$ (Right) against $\kappa_0$. We set $q=12, k_0=2,3,4,5$ in the top row and  $q=8, 12, 16, 20, k_0=4$ in the bottom row.
}}
\vspace{-0.2cm}
\end{figure}

Figure~\ref{fig:male} provides spatial heatmaps of varimax-rotated loading matrix and sparse loading matrices with different sparsity levels for Japanese males, supporting Section~\ref{sec:mort.data}.
\begin{figure}[!htb]
	\centering
	\begin{subfigure}{0.32\linewidth}
	  \caption{\scriptsize{1st factor (rotated)}}
	  \vspace{-0.5cm}
	  \includegraphics[width=6cm,height=5cm]{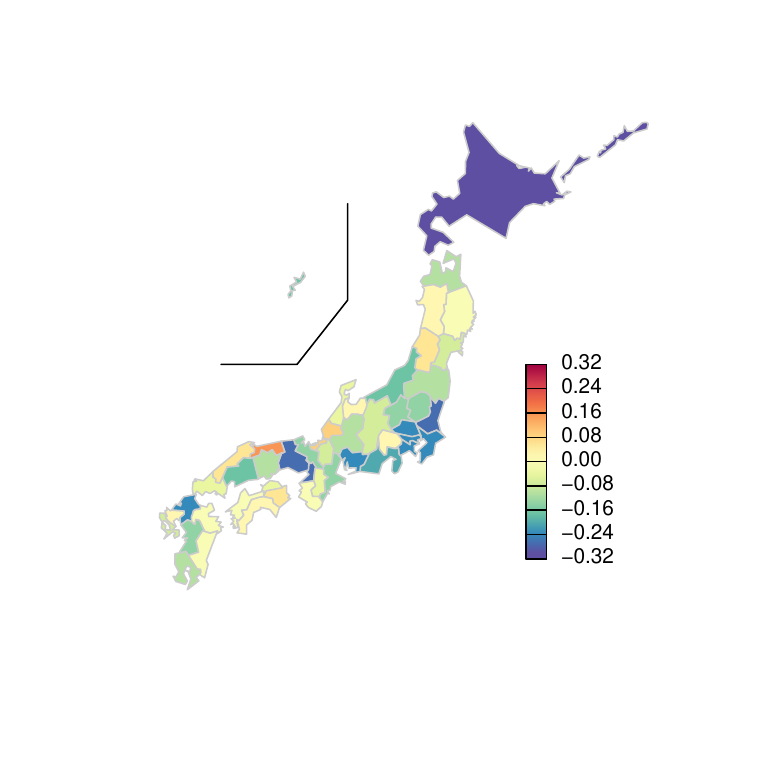}
	  \vspace{-1cm}

	  \caption{\scriptsize{2nd factor (rotated)}}
	  \vspace{-0.5cm}
	  \includegraphics[width=6cm,height=5cm]{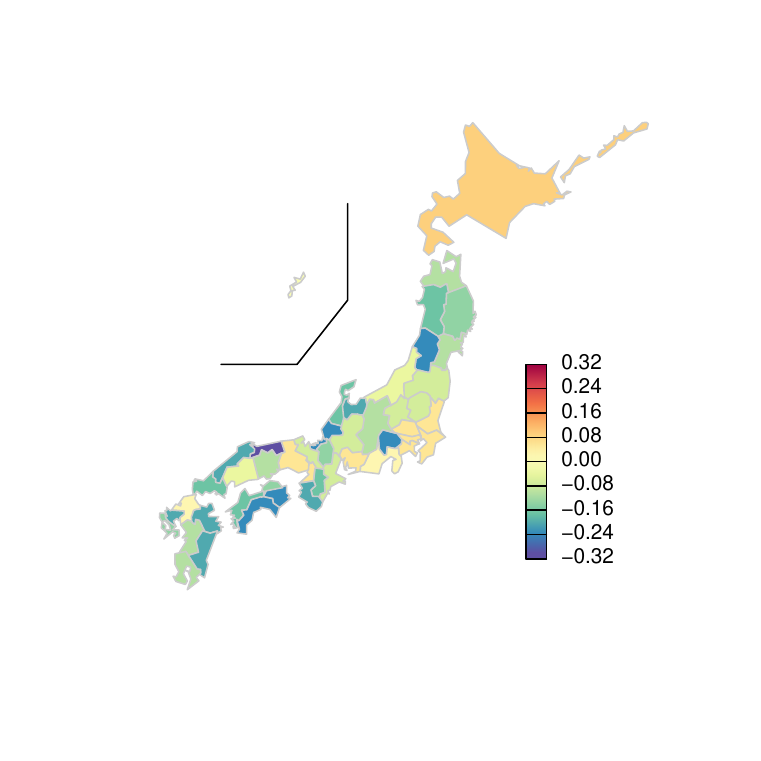}
	  \vspace{-1cm}
	\end{subfigure}\hspace{-0.35cm}	
	\begin{subfigure}{0.32\linewidth}
	  \caption{\scriptsize{1st factor (27 zeros)}}
	  \vspace{-0.5cm}
	  \includegraphics[width=6cm,height=5cm]{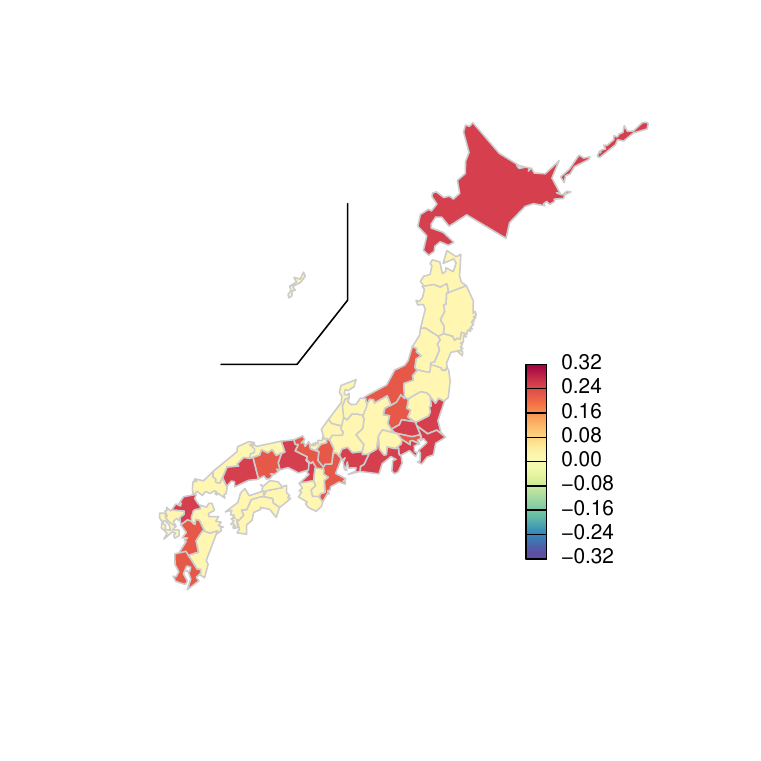}
	  \vspace{-1cm}

	  \caption{\scriptsize{2nd factor (27 zeros)}}
	  \vspace{-0.5cm}
	  \includegraphics[width=6cm,height=5cm]{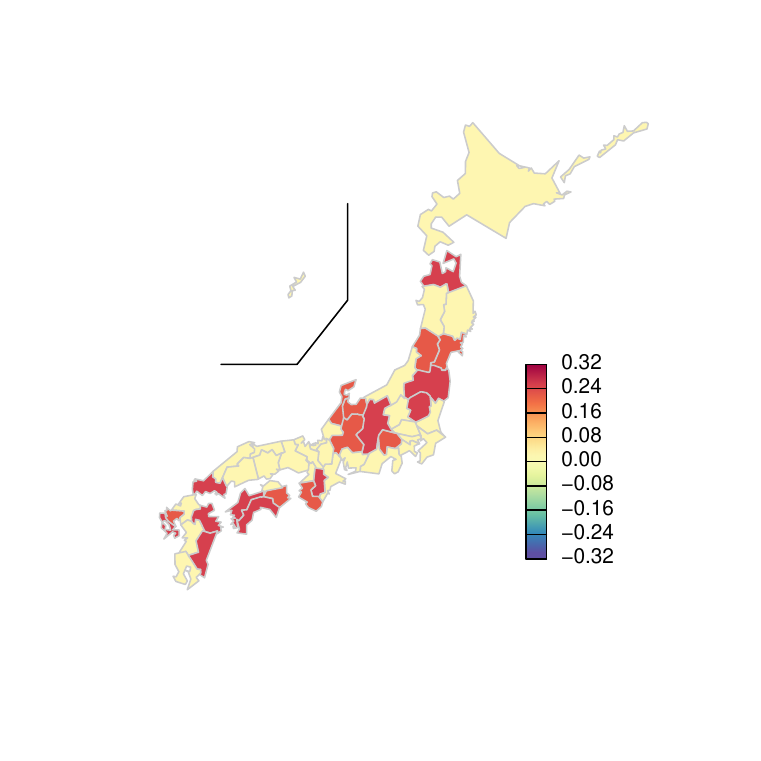}
	  \vspace{-1cm}
	\end{subfigure}\hspace{-0.35cm}
	\begin{subfigure}{0.32\linewidth}
		\caption{\scriptsize{1st factor (32 zeros)}}
	  \vspace{-0.5cm}
	  \includegraphics[width=6cm,height=5cm]{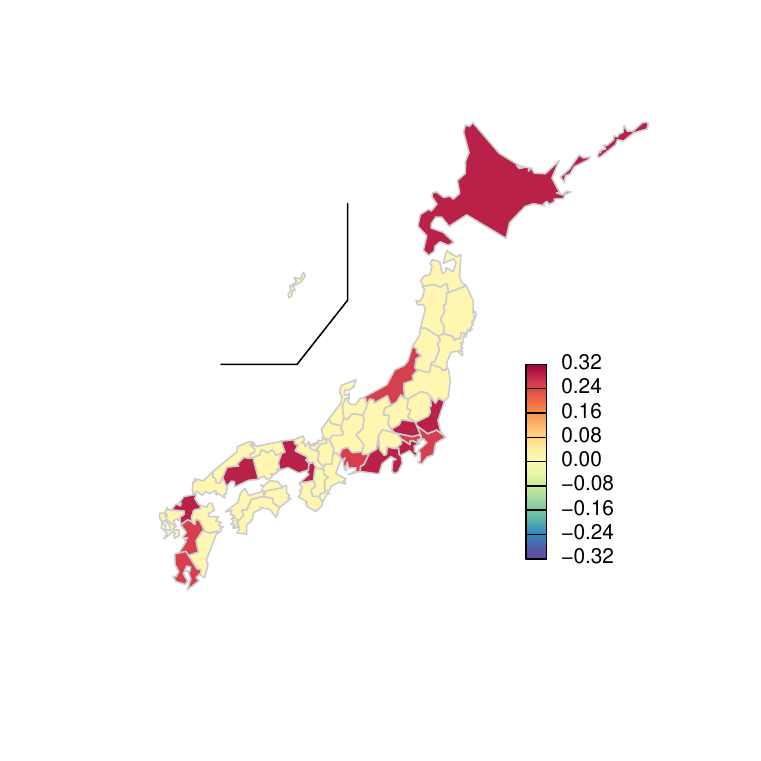}
	  \vspace{-1cm}

	  \caption{\scriptsize{2nd factor (32 zeros)}}
	  \vspace{-0.5cm}
	  \includegraphics[width=6cm,height=5cm]{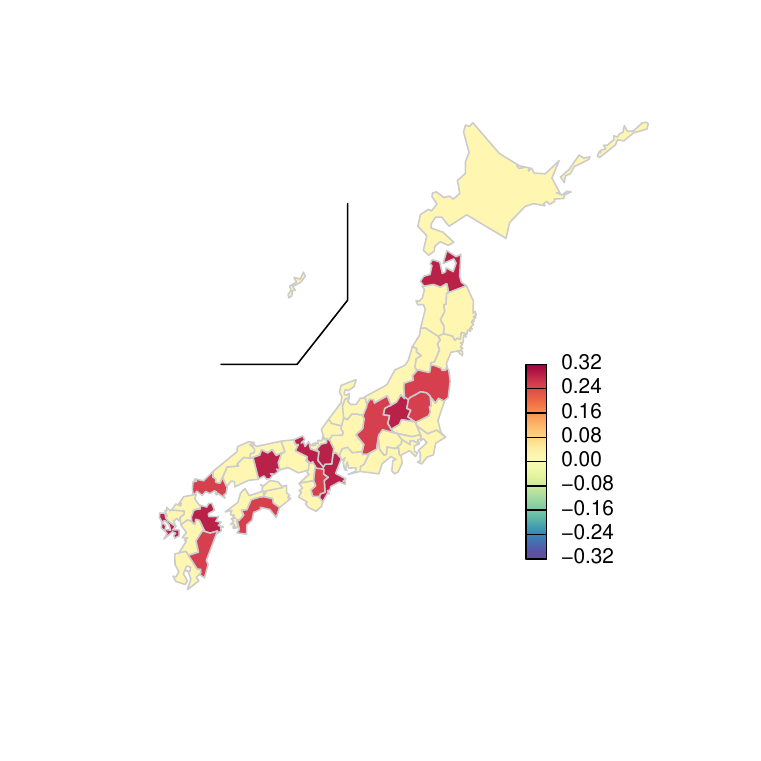}
	  \vspace{-1cm}
	\end{subfigure}
	\centering
	\caption{\label{fig:male}{\it Spatial heatmaps based on varimax-rotated loadings (left column) and sparse loadings with 27 zeros (middle column) and 32 zeros (right column) of 47 prefectures on two factors for Japanese males.
	}}
\end{figure}

\newpage
\spacingset{1.4}
\section*{References}
\begin{description}
    \item
	Fang, Q., Guo, S. and Qiao, X. (2022). Finite sample theory for high-dimensional functional/scalar time series with applications, {\it Electronic Journal of Statistics} {\it 16}(1), 527-591.
	\item 
	Guo, S. and Qiao, X. (2023). On consistency and sparsity for high-dimensional functional time series with application to autoregressions, {\it Bernoulli} {\it 29}(1), 451-472.
	\item
	Han, Y., Chen, R. and Zhang, C. H. (2022). Rank determination in tensor factor model. {\it Electronic Journal of Statistics} {\it 16}(1), 1726–1803.
	\item 
	Lam, C., Yao, Q. and Bathia, N. (2011). Estimation of latent factors for high-dimensional time series, {\it Biometrika} {\it 98}(4), 901–918.
	\item 
	Vu, V. Q. and Lei, J. (2013). Minimax sparse principal subspace estimation in high dimensions, {\it The Annals of Statistics} {\it 41}(6) 2905–2947.
	\item
	Wang, B.-Y. and Xi, B.-Y. (1997). Some inequalities for singular values of matrix products, {\it Linear Algebra and Its Applications} {\it 264} 109–115.
\end{description}
\end{document}